\documentclass[a4paper,fleqn]{cas-sc}

\usepackage[authoryear,longnamesfirst]{natbib}

\usepackage{tikz}
\usepackage{xcolor}
\usepackage{amsthm}
\usepackage{amsfonts}
\usepackage{amsmath}
\usepackage{amssymb}

\usepackage{harmony}

\usepackage{stmaryrd}
\usepackage{color}

\usetikzlibrary{positioning}
\usepackage{bm}
\usepackage{hyperref}
\usepackage{float}
\usepackage{ebproof}
\usepackage{multicol}
\usepackage{appendix}
\usepackage{comment}
\usepackage{subcaption}

\newtheorem{definition}{Definition}
\newtheorem{theorem}{Theorem}
\newtheorem{lemma}[theorem]{Lemma}

\newtheorem{example}{Example}
\newdefinition{rmk}{Remark}

\newcommand{\M}{\mathcal{M}}
\newcommand{\ModelM}{\M=(W,\lbrace\sim_a \rbrace_{a\in \mathcal{A}},V)}
\newcommand{\actionModel}{\mathrm{M=(E^M},\lbrace \sim^{\mathrm{M}}_a \rbrace_{a \in \mathcal{A}}, \mathrm{pre^M)}}
\newcommand{\lang}{\mathcal L}
\newcommand{\lDEL}{\lang_{DEL}}
\newcommand{\satisfies}{\vDash}
\newcommand{\lAct}{\lang_{\mathrm{Evt}}}
\newcommand{\evt}{\mathrm{e}}
\newcommand{\Var}{\mathrm{Var}}
\newcommand{\Evt}{\mathrm{Evt}}

\newcommand{\st}{\ \ViPa \ }
\newcommand{\stSpace}{\quad \ViPa \quad}
\newcommand{\store}{\mathfrak{E}}

\newcommand{\eq}{\leftrightarrow}

\newcommand{\imp}{\rightarrow}
\newcommand{\Imp}{\Rightarrow}

\newcommand{\T}{\top}

\renewcommand{\phi}{\varphi}
\newcommand{\union}{\cup}

\newcommand{\inter}{\cap}

\newcommand{\update}{\sslash}

\newcommand{\Sep}{\ \bm{\;\Big|\;} \ }
\newcommand{\SepC}{ \bm{\;\Big|\;}  }  

\newcommand{\calculus}{\textbf{$\mathbf{IDHS}_{\mathrm{DEL}}$}}
\newcommand{\xMerge}{\otimes}
\newcommand{\indexMerge}{\oplus}

\newcommand{\prefix}{\sqsubseteq}

\newcommand{\byIH}{\overset{(IH)}{\rightsquigarrow} }

\newcommand{\Rule}{\mathcal{R}}

\ebproofnewstyle{regular}{
left label template = \scriptsize\inserttext,
right label template = \small\inserttext}

\ebproofnewstyle{compact}{
left label template = \scriptsize\inserttext,
right label template = \small\inserttext,
separation=0.6em}

\ebproofnewstyle{inter}{
left label template = \scriptsize\inserttext,
right label template = \scriptsize\inserttext,
separation=0.8em}

\ebproofnewstyle{small}{
separation = 0.5em, rule margin = .5ex,
right label template = \scriptsize\inserttext,
template = \small$\inserttext$ }

\ebproofnewstyle{tiny}{
separation = 0.5em, rule margin = .5ex,
right label template = \scriptsize\inserttext,
template = \footnotesize$\inserttext$ }

\ebproofnewrulestyle { no~rule } {
rule~code =
}

\tikzstyle{index on}=[inner sep=2pt, white, circle, fill=black]
\tikzstyle{index off}=[inner sep=2pt, black, circle, draw]
\tikzstyle{index gray}=[inner sep=2pt, black, circle, fill=lightgray]
\tikzstyle{opaque}=[fill=gray,fill opacity=.1]
\tikzstyle{counter}=[densely dashed]

\begin{document}
\let\WriteBookmarks\relax
\def\floatpagepagefraction{1}
\def\textpagefraction{.001}
\shorttitle{Dynamic Proof Theory for DEL}
\shortauthors{C. Lerouvillois \& F. Poggiolesi}

\title [mode = title]{Dynamic Proof Theory for Dynamic Epistemic Logic}                   

\author[1]{Clara Lerouvillois}
\ead{clara.lerouvillois@irit.fr}

\affiliation[1]{IRIT, CNRS -- INP -- University of Toulouse, France & IHPST UMR 8590, CNRS -- University Paris 1 Pantheon Sorbonne, France}

\author[2]{Francesca Poggiolesi}

\affiliation[2]{IHPST UMR 8590, CNRS -- University Paris 1 Pantheon Sorbonne, France}


\begin{abstract}
This paper proposes a calculus for Dynamic Epistemic Logic (DEL) based on dynamic hypersequents, a new framework for capturing the dynamics of epistemic models in a purely structural and natural way. We enrich dynamic hypersequents with agent indices and an event store, yielding a uniform representation of epistemic events. The resulting calculus enjoys a set of logical rules that are provably invertible. Moreover, all structural rules, including the contraction rules, as well as the cut-rule, are shown to be admissible. Finally, the calculus constitutes, to the best of our knowledge, the only proof-theoretic calculus for DEL that is both grounded in a classical S5 framework and inherently dynamic.
\end{abstract}

\begin{keywords}
proof theory \sep dynamic epistemic logic \sep action models \sep hypersequents \sep cut-elimination \sep
\end{keywords}

\maketitle

\section{Introduction}

Traditional epistemic logic, typically based on S5 modal systems, provides a static representation of agents’ knowledge by characterising what is known in a given epistemic situation, while abstracting away from the processes through which knowledge is acquired, communicated, or modified. Although this perspective has proven fundamental for the study of knowledge, it is insufficient for modelling the evolution of epistemic states resulting from communication, observation, and other forms of interaction. Dynamic Epistemic Logic (DEL)\footnote{In this paper we call DEL the logic which is also known as \emph{Action Model Logic}, \emph{e.g.} in \cite{DEL}, or \emph{Logic of Epistemic Actions}, \emph{e.g.} in \cite{baltagandothers}.} addresses this limitation by offering an explicit representation of epistemic actions, or events, that transform agents’ information states, see \textit{e.g.} \cite{DEL}, \cite{baltagandothers}. 

From a semantic perspective, DEL is a very rich framework. Based on Kripke models and event models, which have been extensively studied and refined in \cite{plaza1989, baltagandothers},  DEL provides a unified setting where different kinds of epistemic changes can be represented and combined. Public announcements, private announcements, and other general epistemic actions can be modelled as events transforming epistemic models and modifying agents’ information states. This ability to represent a wide variety of knowledge-transforming processes within a single framework is one of the main sources of the expressive power of DEL. 

By contrast, the proof-theoretic landscape of dynamic epistemic logic remains comparatively fragmented. Existing approaches include tableaux systems, i.e. \cite{aucherandothers1, aucherandothers2}, as well as sequent-based systems. Among the latter, there are labelled calculi, as the one introduced in \cite{Nomuraandothers2}, which further develops the approach of \cite{nomura2015revising}, and display calculi, as those presented in \cite{Frittellandothers1, Frittellandothers2}. 
To the best of our knowledge, the only proof-theoretic treatment of DEL with an epistemic base logic S5  prior to the present work is \cite{Sano2}. In particular, Liu and Sano develop non-labelled sequent calculi for public announcement logic and action model logic over both K45 and S5. The expressive power of DEL is captured by introducing rules corresponding to so-called reduction axioms, namely axioms that translate any dynamic formula into a static formula of epistemic logic. For the K45-based systems, full cut elimination is established. In the S5 case, however, full cut elimination fails, and cuts can instead be restricted to formulas belonging to a suitably extended subformula closure.

The main aim of this paper is to develop a novel proof-theoretical approach to DEL which offers a twofold advantage. On the one hand, it captures the expressive strength of DEL as an extension of the epistemic logic base S5 and allows a full cut-elimination result. On the other hand, this approach internalises the dynamic transformations induced by epistemic events by purely syntactic means, without relying on auxiliary semantic devices nor on reduction axioms. Hence, the resulting system is, as far as we know, the first non-labelled sequent calculus that is both S5 based and inherently dynamic.


To address the first challenge, we build on the proof system for S5 presented in \cite{Poggiolesi2008}, whose structural simplicity makes it particularly suitable as a foundation for our extension. The key feature of this system is that the structural organisation of hypersequents already internalises the characteristic properties of S5. In particular, the conditions corresponding to reflexivity, transitivity, and symmetry are not enforced through additional logical rules, but are encoded in the structure of hypersequents themselves. This makes the calculus a natural basis for incorporating the dynamic mechanisms required by DEL.

To address the second challenge, namely, to capture the dynamic aspect of DEL, we rely on the framework introduced in \cite{lerouvillois2025dynamichypersequentspublicannouncement}, in which hypersequents are extended so as to become dynamic. In the present paper, we further enrich dynamic hypersequents by incorporating \emph{indices} for agents'uncertainty and by introducing a new structure for representing arbitrary epistemic events, which we call \emph{event store}. This allows epistemic transformations to be represented directly at the structural level, rather than by means of additional semantic machinery. Thus, by combining enriched dynamic hypersequents with a structural proof system for S5, we obtain a calculus for DEL with an S5 epistemic base. In this calculus, we prove that all logical rules are invertible and that all structural rules -- cut-rule included -- are admissible.  

The paper is structured in the following way. \emph{Section} \ref{sectionDEL} defines the basic notions and notations of DEL, which will be useful for what follows. \emph{Section} \ref{sectionIDHSwAS} introduces the notion of indexed dynamic hypersequent with event store, together with its interpretation. In \emph{Section} \ref{sectionCalculus} we present our calculus for DEL, whilst in \emph{Section} \ref{sectionAdmissibility} we prove that structural rules of weakening, contraction and merge are admissible in it. \emph{Section} \ref{sectionCompleteness} serves to prove that the calculus is sound and complete with respect to DEL, whilst \emph{Section} \ref{sectionCut} is dedicated to show the cut-rule is admissible in the calculus. We end the paper with \emph{Section} \ref{sectionConclusion} where we sketch some interesting paths of future research.

\section{Dynamic Epistemic Logic}\label{sectionDEL}

We use this section to introduce Dynamic Epistemic Logic (DEL), both from a semantic and from a syntactic -- Hilbert-style system -- perspectives.

\begin{definition}[Language of DEL]
The language $\lang$ of DEL is composed of a countable set $P$ of atomic sentences $p, q,etc.$, a finite set $\mathcal{A}$ of agents $\mathcal{A} = \{a,b,\dots\}$ , the connectives $\neg$, $\wedge$, modalities $K_{a}$, for each agent $a$ belonging to $\mathcal{A}$, and the event modality $[\cdot]$. 
\end{definition}

\begin{definition}[Formulas of DEL]
    Given the language $\lang$, we simultaneously define the formulas $A \in \lDEL$ and pointed event models $\mathrm{(M,e)}\in \lAct$ by induction.
    On the one hand, \emph{formulas} $A\in \lDEL$ are defined by the BNF
\begin{center}
    $A \::=\ p \ | \ \bot \ | \ \lnot A \ |\ (A \land A)\ |\ K_a A \ |\ [\mathrm{(M,e)}]A$
\end{center}
where $p\in P$, $a\in \mathcal{A}$, and $\mathrm{(M,e)}\in \lAct$.

On the other hand, \emph{event models} $\mathrm{M}=(\mathrm{E^M}, \lbrace\sim^{\mathrm{M}}_a\rbrace _{a\in \mathcal{A}},\mathrm{pre^M})$ are defined by a \emph{finite domain} $\mathrm{E^M}$ of event points $\mathrm{e,f}, etc.$, equivalence relations $\sim^{\mathrm{M}}_a$ on $\mathrm{E^M}$ and a \emph{precondition} function $\mathrm{pre^M}:\mathrm{E^M} \imp \lDEL$ mapping each event point $\evt\in \mathrm{E^M}$ to a $\lDEL$-formula $\mathrm{pre^{M}(e})$ that has been constructed in a previous stage of the induction.
More precisely, let $\lDEL^0$ be the formulas of epistemic logic and $\lAct^0$ the set of pointed event models with preconditions in $\lDEL^0$. Then, whenever $\lDEL^n$ and $\lAct^n$ are both defined, $\lDEL^{n+1}$ is defined by the BNF $A \::=\ p\ |\ \lnot A \ |\ (A \land A)\ |\ K_a A \ |\ [\mathrm{(M,e)}]A $ where $\mathrm{(M,e)} \in \lAct^n$, whereas $\lAct^{n+1}$ is defined as the set of pointed event models with precondition formulas in $\lDEL^{n}$. Finally, $\lDEL$ is defined as $\lDEL:= \bigcup_{n \in \mathbb{N}} \lDEL^n$ and $\lAct := \bigcup_{n\in \mathbb{N}}\lAct^n$. See \cite{sepDEL} for further detail.

If $\mathrm{M}$ is an event model and $\evt\in \mathrm{E^M}$ is an event point in the domain of that model, then $\mathrm{(M,e)} \in \lAct$ is a \emph{pointed event model}. 
We follow the standard rules for omission of the parentheses. The connectives $\T, \vee, \rightarrow, \leftrightarrow, \hat{K}_a, \langle\cdot \rangle$ are defined by abbreviation as usual. In particular $\hat{K}_a \phi := \lnot K_a \lnot \phi$ and $\langle (\mathrm{M,e})\rangle \phi := \lnot [(\mathrm{M,e})]\lnot \phi$. Formulas of the form $[(\mathrm{M,e})]A$ are read `after event $(\mathrm{M,e})$ occurs, $A$ holds'. 

\end{definition}


For clarity and succinctness, let us use the following notational conventions. We use $\mathrm{e^M}$ as an abbreviation to denote the \emph{pointed event model} $(\mathrm{M,e})$ and we omit the parentheses in $[(\mathrm{e^M})]A$ when there is no ambiguity, thereby getting $[\mathrm{e^M}]A$. Analogously, instead of $\actionModel$ we will simply write $\mathrm{M}=(\mathrm{E}, \lbrace\sim_a \rbrace_{a \in \mathcal{A}},\mathrm{pre})$ if there is no ambiguity.
Finally, the cardinality of $\mathrm{E^M}$ will be denoted by $|\mathrm{M}|$.

\begin{definition}\label{def:espitemicmodel}[Epistemic model]
    An \emph{epistemic model} is a tuple $\ModelM$ where $W$ is a set of possible worlds, all $\sim_a$ are equivalence relations on $W$ and $V:P \imp \mathcal{P}(W)$ is a valuation. For a model $\M$ and world $w \in W$, we call $(\M,w)$ a \emph{pointed epistemic model}, or simply a \emph{pointed model}.
\end{definition}

\begin{definition}[Semantics of DEL]\label{defSemDEL}
Given an epistemic model $\ModelM$, and an event model $\actionModel$, the \emph{semantic relation} $\satisfies$ and \emph{the updated model} $\M' = \M \times \mathrm{M}$ are simultaneously defined by induction as below.
\begin{align*}
&\M,w \satisfies p                &   &\text{iff}  &   &w \in V(p) \\
&\M, w \not\satisfies \bot \\
&\M,w \satisfies \lnot A          &   &\text{iff}  &   &\M,w \nvDash A \\
&\M,w \satisfies A\land B         &   &\text{iff}  &   &\M,w \satisfies A \text{ and } \M,w \satisfies B \\ 
&\M,w \satisfies K_a A             &   &\text{iff}  &   &\M,v \satisfies A \text{ for all } v \in W \text{ s.t. } w \sim_a v\\
&\M,w\satisfies[\mathrm{e^M}]A & &\text{iff} & &\text{if } \M,w\satisfies \mathrm{pre^M}(\evt) \text{ then } \M',(w,\evt)\satisfies A
\end{align*}

\noindent $\M' := \M \times \mathrm{M}$ is the restricted modal product of $\M$ and $\mathrm{M}$, defined as $\M'=(W',\lbrace \sim'_a\rbrace_{a\in \mathcal{A}},V')$ where:
    \begin{align*}
        &W' & &= & &\lbrace (w, \evt) \ | \ w\in W, \mathrm{e\in E^M} \text{ and } \M,w\satisfies \mathrm{pre^M(e}) \rbrace \\
        &(w,\evt)\sim_a'(v,\mathrm{f})  & &\text{iff} & &w\sim_a v \text{ and } \mathrm{e\sim^M}_a \mathrm{f} \\
        &(w,\evt)\in V'(p) & &\text{iff} & &w\in V(p)
    \end{align*}
A formula $A$ is \emph{valid}, denoted $\satisfies_{DEL} A$, if, and only if, for all models $\ModelM$ and all worlds $w \in W$, $\M,w\satisfies A$. The set of all validities is denoted \emph{DEL}.
\end{definition}

\noindent Note that each $\sim_a'$ in the updated model $\M' = \M \times \mathrm{M}$ is an equivalence relations, since both $\sim_a$ and $\sim_a^{\mathrm{M}}$ are equivalence relations. Also note $[\mathrm{e^M}]A$ is true at $(\M,w)$ just in case whenever the event $\evt^{\mathrm{M}}$ is \emph{executable} in $(\M,w)$, namely $\M,w \satisfies \mathrm{pre^M(\evt)}$, $A$ holds in the updated model $\M' = \M \times \mathrm{M}$ and corresponding world $(\M', (w,\evt))$. The dynamic modality therefore acts similarly as an implication, though across updated models. In particular, to evaluate a dynamic formula $[\mathrm{e^M}]A$, not only do we need to look at a world $w$ in a model $\M$ but we also need to move to the updated model $\M'$ and evaluate $A$ in world $(w, \evt)$.   

Before we proceed, let us consider a simple example of event model and update. This is a mere adaptation of Example 6.13 of \cite{DEL}.

\begin{example}[$\mathrm{Mayread}$]
Consider two agents Anne and Bob, and an epistemic model $\M=(\lbrace w,v\rbrace, \sim_a, \sim_b, V)$ where $V(p)=\lbrace w\rbrace$ (see Figure \ref{fig:ex} (a)).

\begin{figure}[pos=h]
    \centering
\begin{subfigure}[c]{0.45\linewidth}
\begin{tikzpicture}[scale=0.8]
    \draw (0,1.5) node {};
    \draw (-2,0) node (0) {\underline{$w$}};
    \draw (2,0) node (1) {$v$};
    \draw (0,-1.5) node {};
    \draw[<->] (0) --node[fill=white, inner sep=1pt] {$a,b$} (1);
    \draw      (0) edge [loop left] node {$a,b$}  (0);
    \draw      (1) edge [loop right] node {$a,b$}  (1);
\end{tikzpicture}
\caption{Epistemic model $\M$ where both Anne and Bob are uncertain about $p$. The actual world is $w$.}
\end{subfigure}
\hfill
\begin{subfigure}[c]{0.45\linewidth}
\begin{tikzpicture}[scale=0.8]
    \draw (-2,2)  node   (w)    {\small $\mathrm{read_{p}}$};
    \draw (2,2)   node   (u)    {\small $\mathrm{read_{np}}$};
    \draw (0,0)   node   (v)    {\small \underline{$\mathrm{\lnot read}$}};
    \draw[<->] (w) --node[fill=white, inner sep=1pt] {$b$} (u);
    \draw[<->] (u) --node[fill=white, inner sep=1pt] {$b$} (v);
    \draw[<->] (w) --node[fill=white, inner sep=1pt] {$b$} (v);
    \draw (w) edge [loop left] node {$a,b$}   (w);
    \draw (u) edge [loop right] node {$a,b$}  (u);
    \draw (v) edge [loop right] node {$a,b$}  (v);
\end{tikzpicture}
\caption{Event model $\mathrm{M}$ where Anne may read the letter. The actual event is $\lnot\mathrm{read}$.}
\end{subfigure}
\caption{Example $\mathrm{Mayread}$.}
\label{fig:ex}
\end{figure}
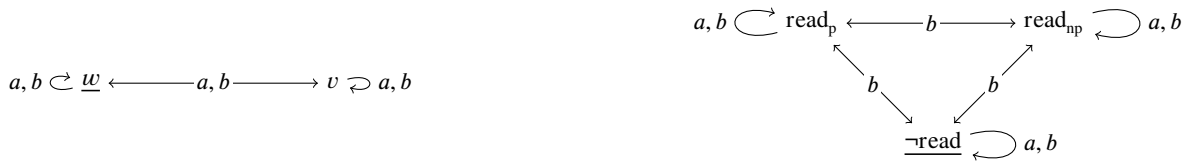

In the epistemic state $w$, $p$ is true but both Anne and Bob do not know $p$. Consider now the situation where Anne \emph{may read} a letter containing the truth value of $p$ \emph{i.e.} either $p$ or $\lnot p$. From Bob's point of view, three alternative scenarios are indistinguishable: (i) Anne does not open the letter, (ii) Anne reads the letter and learns $p$, and (iii) Anne reads the letter and learns $\lnot p$. These three scenarios are represented by three different points in the event model $\mathrm{M}$ depicted in Figure \ref{fig:ex} (b), respectively $\mathrm{\lnot read, read_{np}}$ and $\mathrm{read_p}$, with precondition $\mathrm{pre^M(\lnot read)}= \top, \mathrm{pre^M(read_{np})} =\lnot p$ and $\mathrm{pre(read_p)}=p$. Bob cannot distinguish between the three events but obviously Anne can. Hence, their relation in the event model are the universal one and the identity relation, respectively. Finally, by the preconditions, it is easy to see that event $\mathrm{read_{p}}$ can only be executed in $w$, $\mathrm{read_{np}}$ only in state $v$, while $\mathrm{\lnot read}$ can be executed in both states. Hence the resulting updated model, see Figure \ref{fig:exUpdatedModel}, will have four states.

Let us assume Anne does not open the letter, so the first scenario is the actual one (underlined in Figure \ref{fig:ex}), \emph{i.e.} the pointed event model $(\mathrm{M},\lnot \mathrm{read})$. In the updated model $\M':=\M \times \mathrm{M}$, see again Figure \ref{fig:exUpdatedModel}, the actual state is therefore $(w,\lnot \mathrm{read})$ wherein $p$ is true but neither Anne nor Bob knows it.

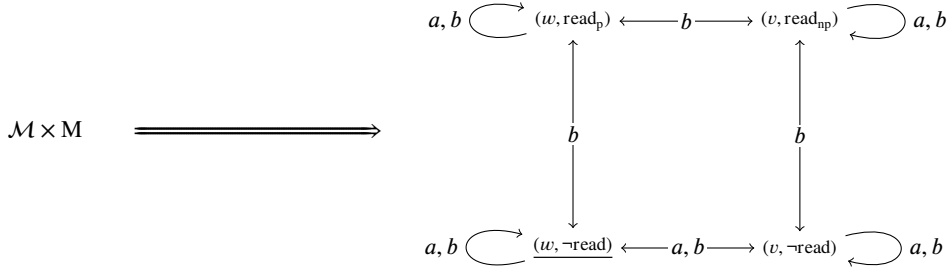
\begin{figure}[pos=h]
\centering
\begin{tikzpicture}
    \draw (0,0)    node (x) {$\M \times \mathrm{M} \qquad \xRightarrow{\hspace{3cm}}$};
    \draw (5,1.5)   node (1w) {\scriptsize$(w,\mathrm{read_{p}})$};
    \draw (8,1.5)   node (0u) {\scriptsize$(v,\mathrm{read_{np}})$};
    \draw (5,-1.5)  node (1v) {\scriptsize\underline{$(w,\mathrm{\lnot read})$}};
    \draw (8,-1.5)  node (0v) {\scriptsize$(v,\mathrm{\lnot read})$};
%
    \draw[<->] (1w) --node[fill=white, inner sep=1pt] {$b$} (0u);
    \draw[<->] (1w) --node[fill=white, inner sep=1pt] {$b$} (1v);
    \draw[<->] (0v) --node[fill=white, inner sep=1pt] {$b$} (0u);
    \draw[<->] (0v) --node[fill=white, inner sep=1pt] {$a,b$} (1v);
    \draw (1w) edge [loop left] node {$a,b$} (1w);
    \draw (0u) edge [loop right] node {$a,b$} (1w);
    \draw (1v) edge [loop left] node {$a,b$} (1w);
    \draw (0v) edge [loop right] node {$a,b$} (1w);
\end{tikzpicture}
\caption{Resulting updated model $\M'= \M \times \mathrm{M}$ in the example $\mathrm{Mayread}$. The actual world is $(w,\mathrm{\lnot read})$, where Anne has not read the letter but Bob does not know it.}
\label{fig:exUpdatedModel}
\end{figure}

\end{example}

From atomic events we can construct composed events in the following way.

\begin{definition}[Event composition]
    Given event models $\actionModel$ and $\mathrm{N=(E^N},\lbrace\sim^{\mathrm{N}}_a\rbrace_{a\in \mathcal{A}}, \mathrm{pre^N)}$, their \emph{composition} $\mathrm{(M;N)}$ is the event model $\mathrm{L=(E^L}, \lbrace \sim^{\mathrm{L}}_a\rbrace_{a\in \mathcal{A}},\mathrm{pre^L)}$ such that $\mathrm{E^L} = \mathrm{E^M}\times\mathrm{E^N}$, $\sim^{\mathrm{L}}_a = \sim^{\mathrm{M}}_a \times \sim^{\mathrm{N}}_a$ (for all $a\in \mathcal{A}$) and for all $\evt\in \mathrm{E^M}, \mathrm{f} \in \mathrm{E^N}$, $pre^{\mathrm{L}}(\mathrm{(e,f)}) = \mathrm{pre^{M}(e)} \land [\mathrm{e^M]pre^{N}(f)}$.
\end{definition}

We now move to the axiomatization of Dynamic Epistemic Logic.

\begin{definition}[Axiomatisation \textbf{DEL}]
    Based on the language $\lDEL$, we introduce the Hilbert-style axiomatic system \textbf{DEL}, which is composed of the axioms and rules of inference in Table \ref{AxDEL}. As standard, we will write $\vdash_{\textbf{DEL}} A$, for the formula $A$ is derivable in the Hilbert system \textbf{DEL}.
\end{definition}

\begin{table}[width=.6\linewidth, pos=h]
\centering
\caption{The axiomatisation \textbf{DEL}}\label{AxDEL}
\begin{tabular}{|ll|}
\hline 
All instantiations of propositional tautologies &  \\

Axioms of \textbf{S5m} & \\
\textbf{Reduction axioms:} & \\
$ [\mathrm{e^M}]p \eq (\mathrm{pre(e^M}) \imp p) $   & atomic permanence \\ 
$ [\mathrm{e^M}]\lnot A \eq (\mathrm{pre(e^M}) \imp \lnot[\mathrm{e^M}] A) $ 	& event and negation \\ 
$ [\mathrm{e^M}](A \land B) \eq ([\mathrm{e^M}]A \land [\mathrm{e^M}]B) $& event and conjunction \\ 
$ [\mathrm{e^M}][\mathrm{f^N}]A \eq [(\mathrm{e^M;f^N})]A $ 		& event composition  \\
$\displaystyle [\mathrm{e^M}]K_a A \eq (\mathrm{pre^M(a)} \imp \bigwedge_{\evt\sim^{\mathrm{M}}_e \mathrm{f}} K_a[\mathrm{f^M}]A) $   & event and knowledge \\ 

\textbf{Inference rules:} & \\
From $ A $ and $ A \imp B $, infer $ B $ 	& modus ponens (MP) \\ 
From $A$, infer $K_a A$ 					& necessitation rule (Nec) \\ 
\hline
\end{tabular}
\end{table}

\begin{theorem}[Soundness and Completeness of \textbf{DEL}]\label{AxComp}
For all $A \in \lDEL$, $\vdash_{\textbf{DEL}} A$ if, and only if, $\satisfies_{DEL} A$.
\end{theorem}

\begin{proof}
    The proof can be found in \cite{DEL}.
\end{proof}

For our purpose, we need to introduce notations to deal with sequences of events.

\begin{definition}
    If $\evt^{\mathrm{M}_1}_1, \evt^{\mathrm{M}_2}_2, \dots, \evt^{\mathrm{M}_n}_n$ are events, then $\alpha = \evt^{\mathrm{M}_1}_1 \cdot \evt^{\mathrm{M}_2}_2 \cdots \evt^{\mathrm{M}_n}_n$ is a \emph{sequence of events}. We denote such sequences of events with Greek letters $\alpha, \beta, etc.$, where the empty sequence is denoted $\epsilon$. Sequences of events formulas $[\alpha]A$ are inductively defined by $[\epsilon]A:= A$ and $[\alpha\cdot \mathrm{e^M}]A := [\alpha][\mathrm{e^M}]A$.

    By $\alpha' \prefix \alpha$ we denote a \emph{prefix} of the sequence $\alpha$, namely an initial subsequence of $\alpha$.
    Moreover, for a sequence $\alpha = \evt^{\mathrm{M}_1}_1 \cdot \evt^{\mathrm{M}_2}_2 \cdots \evt^{\mathrm{M}_n}_n$, the model $\M \times \mathrm{M}_1\times \cdots \times \mathrm{M}_n$, obtained after successive updates in $\alpha$, is denoted $\M^{\mathrm{M}(\alpha)}$. 
    Subsequently, $(w,\alpha)$ denotes the world $((\dots((w,\evt_1^{\mathrm{M}_1}),\evt_2^{\mathrm{M}_2})\dots),\evt_n^{\mathrm{M}_n})$ belonging to $\M^{\mathrm{M}(\alpha)}$.
    By $\M, w\satisfies \mathrm{pre}(\alpha)$ we denote that $\M,w \satisfies \mathrm{pre}^{\mathrm{M}_1}(\evt_1)$, $\M\times\mathrm{M}_1,(w,\mathrm{e_1})\satisfies \mathrm{pre}^{\mathrm{M}_2}(\evt_2) \dots$ and $\M\times\mathrm{M}_1\times \cdots \times \mathrm{M}_{n-1}, (( \dots(w,\evt_1),\dots), \evt_{n-1}) \satisfies \mathrm{pre}^{\mathrm{M}_{n}}(\evt_n) $.
    Hence, $\M,w \satisfies [\alpha]A$ if, and only if, if $\M,w \satisfies \mathrm{pre}(\alpha)$ then $\M^{\mathrm{M}(\alpha)},(w,\alpha)\satisfies A$.

    For any two sequences $\alpha=\evt^{\mathrm{M}_1}_1 \cdot \evt^{\mathrm{M}_2}_2 \cdots \evt^{\mathrm{M}_n}_n$ and $\beta=\mathrm{f}^{\mathrm{M}_1}_1 \cdot \mathrm{f}^{\mathrm{M}_2}_2 \cdots \mathrm{f}^{\mathrm{M}_n}_n$, we denote by $\alpha \sim_a \beta$ the fact that $\evt_i \sim_a^{\mathrm{M}_i} \mathrm{f}_i$ for all $1 \leq i \leq n$.
    Subsequently, for two worlds $(w,\beta)$ and $(v,\gamma)$ in the updated model $\M^{\mathrm{M}(\alpha)}$ -- where $\beta= \mathrm{f}_1^{\mathrm{M}_1}, \dots, \mathrm{f}_n^{\mathrm{M}_n}$ and $\gamma = \mathrm{g}_1^{\mathrm{M}_1}, \dots, \mathrm{g}_n^{\mathrm{M}_n}$ -- it holds that $(w,\beta)\sim^{\mathrm{M}(\alpha)}_a (v,\gamma)$ if and only if $w\sim_a v$ and $\beta \sim_a \gamma$ \emph{i.e.} $\mathrm{f}_1\sim_a^{\mathrm{M}_1} \mathrm{g}_1$ and $\mathrm{f}_2 \sim_a^{\mathrm{M}_2} \mathrm{g}_2$ and $\dots$ and $\mathrm{f}_n \sim_a^{\mathrm{M}_{n}} \mathrm{g}_n$. 

    
\end{definition}


\section{Indexed Dynamic Hypersequents for DEL}\label{sectionIDHSwAS}

In this section, we introduce the syntactic objects on which our calculus is built, namely \emph{indexed dynamic hypersequents with event store}. Since these are rather sophisticated objects, we proceed gradually, introducing them step by step. Moreover, rather than presenting each construction in a purely syntactic manner -- which, although possible, would likely be difficult to follow on a first reading -- we motivate each step by the semantic intuition underlying it. While these semantic considerations are not strictly necessary for the formal development, we believe that they provide valuable intuition and make indexed dynamic hypersequents with event store significantly easier to be grasped.

\begin{definition}[Sequent]
A \emph{sequent} is a syntactic object of the form $\Gamma\Imp \Delta$, where $\Gamma$ and $\Delta$ are finite multisets of formulas.
\end{definition}

Building on the extensive literature on sequent calculi for modal logics, e.g. see \cite{Poggiolesi2010} or \cite{Restall}, we adopt the standard view that a sequent can be understood as the syntactic representation of a world in a model. In the dynamic setting, we maintain this interpretation while adding the observation that the same world may occur in different models. In particular, we may need to move from a world $w$ of a model $\M$ to the corresponding world $(w, \alpha)$ in an updated model $\M^{\mathrm{M}(\alpha)}$, where $\alpha$ is a finite sequence of epistemic events. Correspondingly, at the syntactic level, we move from a sequent representing the world $w$ in model $\M$, to a new sequent representing the world $(w, \alpha)$ now considered in the updated model $\M^{\mathrm{M}(\alpha)}$. To capture this distinction, we enhance the standard notion of sequent to the notion of \emph{dynamic sequent} which we introduce as follows.

\begin{definition}[Dynamic sequent]\label{def:DS}
Let $\Gamma_1 \Imp \Delta_1,\cdots,\Gamma_k \Imp \Delta_k$ be sequents and $\alpha_1,\cdots,\alpha_k$ be finite sequences of epistemic events pairwise distinct, \emph{i.e.} $\alpha_i\neq \alpha_j$ for all $1\leq i\neq j \leq k$. Then $\update_{\alpha_1}\Gamma_1  \Imp \Delta_1 \update \cdots \update_{\alpha_k}\Gamma_k \Imp \Delta_k$ is a \emph{dynamic sequent}.

For the empty sequence of events $\epsilon$ we write $\Gamma \Imp \Delta$ instead of $\update_\epsilon \ \Gamma \Imp \Delta$. In particular then $\Gamma \Imp \Delta$ is a dynamic sequent. Capital letters $X, Y, \dots$ will denote such dynamic sequents.
\end{definition}

We further introduce the following notation that will later prove useful.

\begin{definition}\label{def:ActionLabelsDS}
     Let $X = \update_{\alpha_1} \Gamma_1 \Imp \Delta_1 \update \cdots \update_{\alpha_k} \Gamma_k \Imp \Delta_k$ be a dynamic sequent. The set of \emph{event labels} of $X$, denoted $l_\evt(X)$, is defined as $l_{\evt}(X):=\{\alpha_{i}\;\mid\; \alpha_{i}\in X\}$. 
\end{definition}

We now have a way to represent not only a world in the initial model $\M$ but also  in possible updated models $\M^{\mathrm{M}(\alpha_1)}, \dots, \M^{\mathrm{M}(\alpha_k)}$, where it may occur. However, a model is made of several worlds, each possibly occurring in updated models. How can we represent that? We adapt hypersequents -- namely objects of the form $\Gamma_{1}\Rightarrow \Delta_{1} \vert \cdots \vert \Gamma_{k}\Rightarrow \Delta_{k}$  where each sequent $\Gamma_i \Imp \Delta_i$ corresponds to a world in a model -- to this new dynamic setting: vertical bar will now separate dynamic sequents, and not simply sequents. 

\begin{definition}[Dynamic hypersequent] Let $X_1,\cdots,X_k$ be dynamic sequents, then the syntactic object:

$$X_{1}\Sep \cdots \Sep X_{k}$$

\noindent is a dynamic hypersequent.
\end{definition}

Dynamic hypersequents  have been introduced in \cite{lerouvillois2025dynamichypersequentspublicannouncement} and already constitute a natural extension of hypersequents to a dynamic setting. They work adequately for Public Announcement Logic (PAL) in the single-agent case. 
 However, we now need to extend the framework of dynamic hypersequents in two ways to provide a suitable calculus for Dynamic Epistemic Logic (DEL).
First we need to enrich the framework of dynamic hypersequents to accommodate multiple agents. Second, since DEL encompasses not only public announcements but also a variety of other types of epistemic events, our syntactic framework must account for this broader class of events. Let us keep on proceeding step by step, first addressing the multi-agent aspect and then the broader class of epistemic events.

In order to introduce multi-agent systems into our framework, it is useful to first recall the following feature of the semantics of the Hilbert system S5, which is the only system in the modal cube that enjoys the property that we are about to explain. The standard Kripke semantics for S5m is given in terms of possible worlds connected by accessibility relations, one for each agent. Each accessibility relation is an equivalence relation, satisfying reflexivity, symmetry, and transitivity (see also Definition \ref{def:espitemicmodel} Section \ref{sectionDEL}).
For example, consider the following model with four possible worlds $w,u,v,x$, such that $w\sim_a u, v\sim_a x$ and $w\sim_b v$.

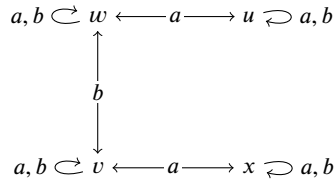
\begin{figure}[pos=h]
    \centering
    \begin{tikzpicture}

    \draw (-1,1)  node (w) {$w$};
    \draw (1,1)   node (u) {$u$};
    \draw (-1,-1) node (v) {$v$};
    \draw (1,-1)  node (x) {$x$};

    \draw[<->] (w) --node[fill=white, inner sep=1pt] {$a$} (u);
    \draw[<->] (v) --node[fill=white, inner sep=1pt] {$a$} (x);
    \draw[<->] (w) --node[fill=white, inner sep=1pt] {$b$} (v);
 
    \draw (w) edge [loop left] node {$a,b$} (w);
    \draw (v) edge [loop left] node {$a,b$} (v);
    \draw (x) edge [loop right] node {$a,b$} (x);
    \draw (u) edge [loop right] node {$a,b$} (u);
    
  	\end{tikzpicture}
    \caption{Example of Kripke model $\M$ with two agents $a$ and $b$}
    \label{fig:Kripke}
\end{figure}

\medskip\noindent Thus, a Kripke frame is naturally represented as a labelled graph, whose nodes are the possible worlds and whose edges correspond to the accessibility relations of the agents.

However, for the case of S5m, because the accessibility relation is an equivalence relation, it is possible to move to an alternative  representation based on partitions, closer in spirit to \cite{aumann}. In particular, for each agent $a\in \mathcal{A}$, the associated equivalence relation $\sim_{a}$ naturally induces a partition $\mathbb{R}_{a}=\{W_{1}, ..., W_{r}\}$ of the set of worlds $W$, that we call an $a$-partition. The \emph{cells} of the partition are the subsets $W_{1}, ..., W_{r}$ and they satisfy the following two requirements: cells are pairwise disjoint \emph{i.e.} $W_i \cap W_j = \emptyset$ for any $1 \leq i \neq j \leq r$, and their union is the whole set $W$. Hence the previous example in terms of partitions and cells becomes the following.

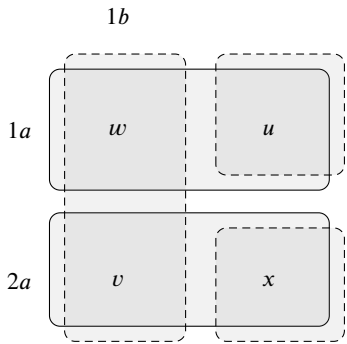
\begin{figure}[pos=h]
    \centering
    \begin{tikzpicture}
    
    \draw[opaque, rounded corners]           (-1.9, 1.8) rectangle (1.8, .2);
    \draw[opaque, rounded corners]           (-1.9, -.1) rectangle (1.8, -1.6);
    \draw[opaque, counter, rounded corners]  (-1.7, 2)   rectangle (-.1, -1.8);
    \draw[opaque, counter, rounded corners]  (.3, .4)    rectangle (2, 2);
    \draw[opaque, counter, rounded corners]  (.3, -.3)   rectangle (2, -1.8);

    \draw (-1,1)  node (w) {$w$};
    \draw (1,1)   node (u) {$u$};
    \draw (-1,-1) node (v) {$v$};
    \draw (1,-1)  node (x) {$x$};

    \draw (-2.3, 1)  node (1a) {$1a$};
    \draw (-2.3, -1) node (2a) {$2a$};
    \draw (-1, 2.5)   node (1b) {$1b$};
    
  	\end{tikzpicture}
    \caption{Example of model $\M$ with partitioning for two agents $a$ and $b$}
    \label{fig:partitioning}
\end{figure}

\medskip\noindent Note that each cell can, when necessary, be designated using indices. In particular, each  non-singleton cell is indexed with an index $na$ composed by the name of the agent $a$ corresponding to the partition that the cell belongs to, together with a natural number $n$, which differentiates it from the other cells of the same partition. So for example, consider again Figure \ref{fig:partitioning} where four worlds are divided into two partitions, one for agent $a$ and the other one for agent $b$.We have the partition for agent $a$, namely $\mathbb{R}_{a} = \{ \{w, u\} \{v, x\}\}$ with the first cell indexed by $1a$ and the second cell indexed by $2a$, and a partition for agent $b$, namely $\mathbb{R}_{b} = \{ \{w, v\} \{u\}, \{x\}\}$  with the first cell indexed by $1b$, and the second and third cells that do not need any index as they only contain one element.\footnote{We do not need to name cells that consist in singletons because reflexivity of the accessible relations $\sim_a$ can be very naturally captured by a specific rule, without any further condition on the indices. See Section \ref{sectionCalculus}.}

Both semantic representations are useful for understanding dynamic indexed hypersequents -- and this is most likely one of their main qualities. Indeed, when extending dynamic hypersequents to the multi-agent setting, it is more natural to think of our syntactic objects from the semantic presentation based on sets. Subsequently, however, when introducing an order on indices and defining the interpretation of dynamic indexed hypersequents with respect to this order, it becomes more convenient to adopt the Kripke-graph perspective. Since we have made both presentations explicit and they are equivalent, we shall freely move between them throughout the paper. We trust the reader to recognize which viewpoint is intended in each context, depending on the aspects under consideration.

\medskip

Let us return to the question of increasing dynamic hypersequents to cover multi-agent setting. To do so, we adopt  a set-theoretic perspective on S5 and enrich dynamic hypersequents with indices of the form just described to designate cells. In particular we equip each dynamic sequent with a set of indices:  each index denotes the cell to which the world-sequent belongs to for each given partition. So, once more with reference to our previous example, suppose we have four sequents $\Gamma_{1}\Rightarrow \Delta_{1}, \Gamma_{2}\Rightarrow \Delta_{2}, \Gamma_{3}\Rightarrow \Delta_{3}, \Gamma_{4}\Rightarrow \Delta_{4}$, each corresponding to one of the worlds $w$, $u$, $v$, $x$. respectively.  Then, in order to reflect the partitions of agents $a$ and $b$, we have the following \emph{indexed hypersequent}:

$$1a, 1b: \Gamma_{1}\Rightarrow \Delta_{1}\Sep 1a:  \Gamma_{2}\Rightarrow \Delta_{2}\Sep 2a, 1b: \Gamma_{3}\Rightarrow \Delta_{3}\Sep 2a: \Gamma_{4}\Rightarrow \Delta_{4}$$

Note that sets of indices do not change through model updates: if a world $w$ is in a certain cell indexed by $1a$ in a model $\M$, the corresponding world will remain in the same cell in the updated model $\M^{\mathrm{M}(\alpha)}$ -- as long as it survives the successive updates.  There could be cells which in the model $\M$ contain $k$ worlds and end up containing only one corresponding world -- or even none -- but this does not affect our notation. So, again with reference to the previous example, suppose we move to the model updated by the event model $\mathrm{M}$ consisting of a unique event $\evt$ that represents the public announcement of some formula $A \equiv \mathrm{pre^M(e)}$, and suppose $x$ is the only world that does not satisfy the precondition of that announcement. In this case, the corresponding updated model is depicted in Figure \ref{fig:partitioning_update} and we will have the following \emph{indexed dynamic hypersequent}:
$$1a, 1b: \Gamma_{1}\Rightarrow \Delta_{1} \update_{\mathrm{e^M}} \Gamma'_1 \Imp \Delta'_1 \Sep 1a:  \Gamma_{2}\Rightarrow \Delta_{2} \update_{\mathrm{e^M}} \Gamma_2' \Imp \Delta_2' \Sep 2a, 1b: \Gamma_{3}\Rightarrow \Delta_{3} \update_{\mathrm{e^M}} \Gamma_3' \Imp \Delta_3' \Sep 2a: \Gamma_{4}\Rightarrow \Delta_{4}$$
\noindent where the set of indices at the beginning of each dynamic sequent do not change. As the reader could observe, we may neglect curly brackets when writing the elements of the set $\eta$ of indices in an indexed dynamic hypersequent. Moreover, when $\eta$ is the empty set, we will also omit it. Thus a standard sequent $\Gamma \Imp \Delta$ is an indexed dynamic sequent.

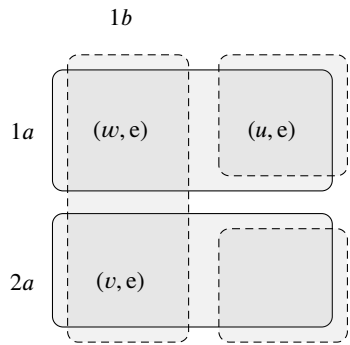
\begin{figure}[pos=h]
    \centering
    \begin{tikzpicture}
    
    \draw[opaque, rounded corners]           (-1.9, 1.8) rectangle (1.8, .2);
    \draw[opaque, rounded corners]           (-1.9, -.1) rectangle (1.8, -1.6);
    \draw[opaque, counter, rounded corners]  (-1.7, 2)   rectangle (-.1, -1.8);
    \draw[opaque, counter, rounded corners]  (.3, .4)    rectangle (2, 2);
    \draw[opaque, counter, rounded corners]  (.3, -.3)   rectangle (2, -1.8);

    \draw (-1,1)  node (w) {$(w,\evt)$};
    \draw (1,1)   node (u) {$(u,\evt)$};
    \draw (-1,-1) node (v) {$(v,\evt)$};

    \draw (-2.3, 1)  node (1a) {$1a$};
    \draw (-2.3, -1) node (2a) {$2a$};
    \draw (-1, 2.5)  node (1b) {$1b$};
    
  	\end{tikzpicture}
    \caption{Example of updated model $\M \times \mathrm{M}$ with partitioning for two agents $a$ and $b$ after update}
    \label{fig:partitioning_update}
\end{figure}

We have thus introduced indexed dynamic hypersequents from an intuitive point of view. Let us move to the technical definitions for those notions.

\begin{definition}[Indices]
    We denote with $\eta,\theta,\iota,\dots$ finite sets of indices of the form $na$ where $n \in \mathbb{N}$  and $a\in \mathcal{A}$. They are such that for all $\eta$ and $a\in \mathcal{A}$, there exists at most one index $na \in \eta$.
\end{definition}

\begin{definition}[Indexed dynamic sequent] Let $X$ be a dynamic sequent and $\eta$ a (possibly empty) set of indices. Then $\eta:X$ is an \emph{indexed dynamic sequent}, IDS for short.\end{definition}


    

\begin{definition}[Indexed dynamic hypersequent]\label{defIDHS}
    If $\eta_1\!:\! X_1,  \dots,  \eta_k\!:\!X_k$ are indexed dynamic sequents, then
    $$G := \eta_1: X_1 \Sep \cdots \Sep \eta_k : X_k$$
    is an \emph{indexed dynamic hypersequent} -- IDHS for short -- as long as the sets of indices respect the following conditions.
    \begin{itemize}
        \item[$(1)$] For all $1 \leq i, j \leq k$ such that $i \neq j$, $\eta_i\cap \eta_j$ contains at most one element;
        \item[$(2)$] For all $1 \leq i, j \leq k$ such that $i \neq j$, there exists a sequence $l_1,\dots,l_p$ with $l_1 =i$, $l_p = j$ and for all $1\leq q < p$, $\eta_{l_q}\cap \eta_{l_{q+1}} \neq \emptyset$;
        \item[$(3)$]  For all $1\leq i,j \leq k$, if there is a sequence $i=l_1,\dots,l_p=j$ (where $1 \leq l_q \leq k$ for all $1 \leq q \leq p$, and $l_{1}, ..., l_{p-1}$ are pairwise distinct) and a sequence of indexed dynamic sequents $\eta_{l_1}:X_{l_1}, \dots,\eta_{l_p}:X_{l_p}$ such that

        $\quad$ - $\eta_{l_q} \cap \eta_{l_{q+1}} \neq \emptyset$, for each pair of indexed sequents $\eta_{l_q}: X_{l_q}$ and $\eta_{l_{q+1}}:X_{l_{q+1}}$ with $1\leq q < p$;
        
        $\quad$ - $\eta_{l_1}: X_{l_1}$ is the same indexed dynamic sequent as $\eta_{l_p}: X_{l_p}$ \emph{i.e.} $i=j$;
        
    \noindent then there is an index $na$ shared by all those sets of indices, \emph{i.e.} $\bigcap_{1 \leq i \leq p} \eta_{l_i} = \lbrace na \rbrace$.
    \end{itemize}
 
     \noindent 
     We call \emph{disconnected indexed dynamic hypersequent} -- for short DIDHS -- an indexed dynamic hypersequent satisfying (1) and (3) but not necessarily (2).
\end{definition}

Let us call any indexed dynamic sequent $\eta:X$ a \emph{world component} and any dynamic sequent $\update_\alpha \Gamma \Imp \Delta$ occurring in $X$ a \emph{dynamic component} of $X$. A world-component might therefore be composed of several dynamic components, and a dynamic hypersequent will be composed of several world-components separated by a vertical bar.

For $a\in \mathcal{A}$ and $\eta$ a set of indices, we denote by $a\in \eta$ the fact that there is $n\in \mathbb{N}$ such that $na\in \eta$.
Finally, $\Vert G \Vert$ denotes the set of indices that occur in the indexed dynamic hypersequent $G$. Consequently, $na \in \Vert G \Vert$ means that the index $na$ occurs in $G$.

\medskip

We now comment on conditions (1), (2) and (3) which shape the type of indexed dynamic hypersequents that can be used in the calculus. First of all note that this construction is an extension to the dynamic framework of the hypersequent calculus for multiagent S5 proposed in \cite{Poggiolesi2013fromSingleToMany}. Recently, \emph{crossword sequents} for S5m have also been proposed in \cite{bílková2025agentinterpolationknowledge}; they enjoy a similar underlying structure.

Condition (1) ensures that any two world-components stand in the same cell for at most one given partition. This condition is needed to ensure the required arbitrariness when creating a new dynamic sequent, similarly to the arbitrariness involved in introducing a fresh world in the construction of a countermodel.

Conditions (2) and (3) are necessary to have a soundness result for our calculus.
To understand them, it is helpful to think of sets of indices (and the indexed dynamic sequents they are associated with) as forming graphs. Indeed, although indexed dynamic hypersequents are presented as set-theoretic objects, this set-based presentation encodes an underlying graph -- where the nodes are given by sets of indices and the edges are the common index between two given sets of indices.
In particular, condition (2) ensures that the graph induced by the sets of indices is connected, meaning for any two sets of indices $\eta$ and $\theta$ there is a path from $\eta$ to $\theta$, namely a sequence $\eta=\chi_0, \chi_1, \dots, \chi_k, \chi_{k+1}=\theta$ such that for all $0 \leq i \leq k$, $\chi_i \cap \chi_{i+1} \neq \emptyset$.
Finally, condition (3) is intended to ensure that cycles, \emph{i.e.} paths $\chi_0, \chi_1, \dots, \chi_k, \chi_{k+1}$ where $\chi_0 = \chi_{k+1}$, are allowed only when all the sets of indices in the path share a common index $na$ -- meaning the associated indexed dynamic sequents represent worlds that lie in the same $a$-cell.

Figure \ref{fig:ConditionsIDHS} illustrates all three conditions both from a graph perspective and a set-theoretic perspective.

%
%
%
%

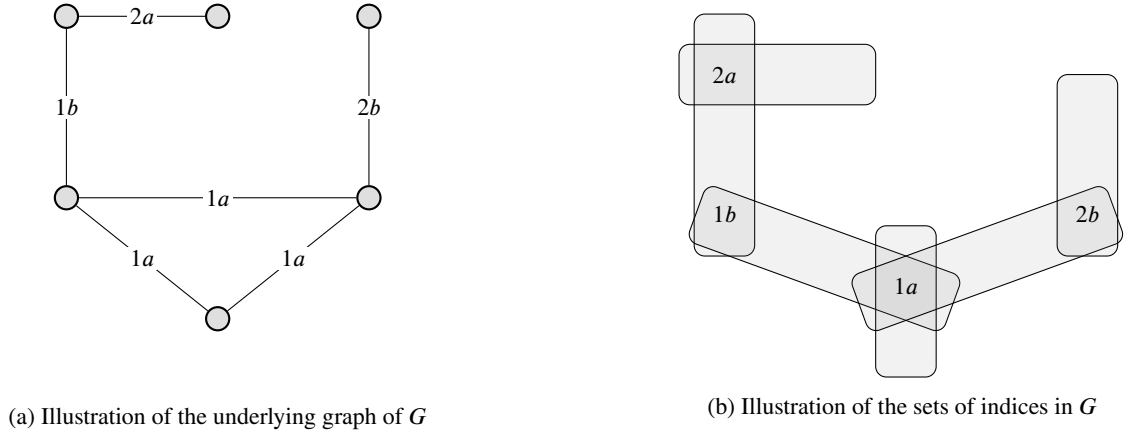
\begin{figure}[pos=h]
\begin{subfigure}[c]{0.45\linewidth}
\centering
\begin{tikzpicture}[scale=0.8,
roundnode/.style={circle, draw=black, fill=lightgray!50, thick},
]
	\node[fill=lightgray, roundnode] (1a1b) at (-2.5,0) 	{};
	\node[fill=lightgray, roundnode] (2a1b) at (-2.5,3) 	{};
	\node[fill=lightgray, roundnode] (2a)	 at (0, 3)		{};
	\node[fill=lightgray, roundnode] (1a)   at (0,-2) 	    {};
	\node[fill=lightgray, roundnode] (1a2b) at (2.5,0)		{};
	\node[fill=lightgray, roundnode] (2b)	 at	(2.5,3)		{};
    \node (void) at (0,-3) {};
	
	\draw (1a1b) -- node[fill=white,inner sep=1pt] {$1b$} (2a1b);
	\draw (1a1b) -- node[fill=white,inner sep=1pt] {$1a$} (1a);
	\draw (1a1b) -- node[fill=white,inner sep=1pt] {$1a$} (1a2b);
	\draw (2a1b) -- node[fill=white,inner sep=1pt] {$2a$} (2a);
	\draw (1a2b) -- node[fill=white,inner sep=1pt] {$2b$} (2b);
	\draw (1a2b) -- node[fill=white,inner sep=1pt] {$1a$} (1a);

\end{tikzpicture}
\caption{Illustration of the underlying graph of $G$}
\end{subfigure}
\hfill
\begin{subfigure}[c]{0.45\linewidth}
\centering
    \begin{tikzpicture}[scale=0.8]
\draw[opaque, rounded corners]  (-3.75, 3.5) rectangle (-0.5, 2.5);
\draw[opaque, rounded corners]  (-2.5, 4)  rectangle (-3.5, 0);
\draw[opaque, rounded corners, rotate=-20]  (-3.5, -1)  rectangle (1, 0);
\draw[opaque, rounded corners]  (-0.5, 0.5)   rectangle (0.5, -2);
\draw[opaque, rounded corners, rotate=20]  (-1, -1)  rectangle (3.5, 0);
\draw[opaque, rounded corners]  (2.5, 0)   rectangle (3.5, 3);

        \draw (-3, 3)  node (2a)   {$2a$};
        \draw (-3, 0.7)  node (1b)   {$1b$};
        \draw (0, -0.5)  node (1a)   {$1a$};
        \draw (3, 0.7)   node (2b)   {$2b$}; 
    \end{tikzpicture}
    \caption{Illustration of the sets of indices in $G$}
\end{subfigure}
\caption{Illustration of the underlying structure of $G:= 1a1b:X_1 \SepC 2a1b: X_2 \SepC 2a: X_3 \SepC 1a: X_4 \SepC 1a2b: X_5 \SepC 2b: X_6$ satisfying conditions (1), (2) and (3) of Definition \ref{defIDHS}}
\label{fig:ConditionsIDHS}
\end{figure}

\bigskip

The sets of indices we have introduced account for the multi-agent component of DEL in the hypersequent framework. We now introduce the last ingredient of our method, namely the so-called \emph{event store} which serves to capture the variety of all types of epistemic events which can occur in DEL and how they relate to each other in terms of agents' indistinguishability.

First of all, let $\mathrm{Var}=\lbrace \mathrm{x,y},\dots \rbrace$ be a countable set of \emph{event variables}.
Let now $\mathrm{Evt} = \lbrace \mathrm{e,f},\dots \rbrace$ be a set of so-called \emph{events} that will be either pointed event models $\mathrm{e^M}\in \lAct$ or variables $\mathrm{x}\in \mathrm{Var}$. From now on then, $\evt$ will denote either a pointed event model or an event variable, and sequences $\alpha \in \mathrm{Evt}^\ast$ will be composed of such events $\evt$. Henceforth, formulas can now contain variables. Eventually, we adopt the restriction that the empty sequence $\epsilon$ is only indistinguishable from itself \emph{i.e.} $\epsilon \sim_a \alpha$ if, and only if, $\alpha = \epsilon$, for all agents $a\in \mathcal{A}$.

\medskip\begin{definition}[Event store]\label{defActionStore}
    An \emph{event store} $\store$ is a finite set of \emph{syntactic} expressions of the form $\evt\sim_a \mathrm{f}$ for $\mathrm{e,f\in Evt}$ either pointed event models or event variables, and $a\in \mathcal{A}$.
    
    We will simply write $\evt\in \store$ to denote the fact that there is some $a\in \mathcal{A}, \mathrm{f}\in \mathrm{Evt}$ such that $\evt\sim_a \mathrm{f}\in \store$.  By $\evt\approx_a \mathrm{f}$ we denote the fact that there is a sequence $\evt=\mathrm{g}_1,\dots,\mathrm{g}_p = \mathrm{f}$ such that, for all $1\leq i < p$ either $\mathrm{g}_{i} \sim_a \mathrm{g}_{i+1}$ or $\mathrm{g}_{i+1} \sim_a \mathrm{g}_i$.
    This is to capture the symmetry and transitivity of the \emph{semantic} semantic relation $\sim^{\mathrm{M}}_a$ in an S5 event model $\actionModel$.\footnote{Although the semantic relation $\sim^{\mathrm{M}}_a$ is an equivalence relation -- and thus also reflexive -- we will not capture reflexivity through the relation $\approx$ in the action store because we will handle it directly through a specific logical rule.}
    We write $\evt\approx \mathrm{f}$ if there are sequences of agents $a_1,\dots,a_p$ and events $\mathrm{g_1,\dots,g_{p-1}}$ such that $\evt\approx_{a_1} \mathrm{g_1} \approx_{a_2}  \cdots \approx_{a_{p-1}} g_{p-1} \approx_{a_p} \mathrm{f}$.
    For any event variable $\mathrm{x}$ in a store $\store$, we say that it is \emph{bound} if there is a pointed event model $\evt^{\mathrm{M}}$ such that $\mathrm{x} \approx \evt^{\mathrm{M}}\in \store$. A variable is \emph{free} if it is not bound. Finally, we require all variables $\mathrm{x}$ in an event store $\store$ to be bound.
    
\end{definition}

\begin{definition}[IDHS with event store]
    An \emph{indexed dynamic hypersequent with event store} is an object $\ G \st \store$ where $G$ is an IDHS and $\store$ an event store containing at least all variables $\mathrm{x}$ appearing in $G$.
\end{definition}

We now have a definition of indexed dynamic hypersequent with event store. Before closing the section, we need to formulate its syntactic interpretation. As usual, we start by explaining the idea of this interpretation in intuitive terms and then move to the formal details. 

An IDHS with event store may be viewed as a set of possible worlds that are partitioned into epistemic equivalence classes and may evolve through model updates. The event store records all epistemic events while the indices specify the equivalence classes associated with each agent. By Conditions (1)–(3), these indices induce the structure of a connected graph. However, the indices alone do not specify at which world-component this graph is rooted. In order to define a sound interpretation, this information is essential. Indeed, the connected graph, by itself, does not determine a unique formula: one must also specify a distinguished world from which the interpretation is evaluated. To overcome this issue, we consider each world-component occurring in an IDHS with event store as a possible choice of root and define the interpretation relative to that root by constructing a \emph{rooted disconnected indexed dynamic hypersequent}.\footnote{This procedure is similar to the construction of a modally equivalent tree model $\M'$ from a model $\M$ rooted at some world $w$. See \emph{e.g.} \cite{BlackburnModalLogic}.} The interpretation of the entire IDHS with event store is then obtained by taking the conjunction of the interpretations associated with all possible choices of such a root.

Having introduced indexed dynamic hypersequents with an event store step by step, we now proceed to develop their interpretation in the same gradual fashion.

\begin{definition}[Interpretation of sequents]
    The interpretation $\tau$ of a sequent $\Gamma \Imp \Delta$ is standard, namely
    $$(\Gamma \Imp \Delta)^{\tau} := \bigwedge \Gamma \imp \bigvee \Delta$$
\end{definition}

\noindent The interpretation $\tau$ is then extended to dynamic sequents in the following way.

\begin{definition}[Interpretation of dynamic sequents]
    If $X=\update_{\alpha_1} \Gamma_1 \Imp \Delta_1 \update \cdots \update_{\alpha_k} \Gamma_k \Imp \Delta_k$ is a dynamic sequent, then its interpretation is defined as follows.
   $$X^\tau := [\alpha_1](\Gamma_1 \Imp \Delta_1)^\tau \lor \cdots \lor [\alpha_k](\Gamma_k \Imp \Delta_k)^\tau$$
\end{definition}

To extend the interpretation to indexed dynamic hypersequents, we first introduce some additional notation. Given an indexed dynamic sequent $\eta$: $X$ occurring in an IDHS $G$, we define the set of world-components accessible from it as the collection of indexed dynamic sequents $\theta:Y$ that share an index with $\eta:X$, \emph{i.e.} such that $\eta \cap \theta \neq \emptyset$. Formally, we have the following definition.

\begin{definition}\label{defn:setofsameindexes}
For any IDS $\eta : X$ in $G$, the set of all the IDSs in $G$ that have one common index with $\eta$ is defined as: 

$$\Sigma^{G}_{\eta:X} := \lbrace \theta : Y \in G \mid \eta \cap \theta \neq \emptyset \rbrace$$
\noindent For any $\eta,\theta$ with a single common index $na$, \emph{i.e.} such that $\eta \cap \theta = \lbrace na \rbrace$, $f(\eta,\theta)$ denotes that agent $a$.
\end{definition}


\medskip\begin{definition}
    Let us consider a dynamic hypersequent $G= \eta_1: X_1 \Sep \cdots \Sep \eta_k: X_k$.
    We construct $G\setminus \eta_i:X_i$ from $G$ by dropping (i) the dynamic sequent $\eta_i:X_i$ and (ii) each of the indices belonging to $\eta_i$ that occurs in other dynamic sequents of $G$, in the following way:
        $$G\setminus \eta_i:X_i := \eta'_1: X_1 \Sep \cdots \Sep \eta'_{i-1}:X_{i-1} \Sep \eta'_{i+1} :X_{i+1} \Sep \cdots \Sep \eta'_k : X_k$$
    where $\eta'_j= \eta_j \setminus \eta_i$. Note that $G\setminus \eta_i:X_i$ is a disconnected indexed dynamic hypersequent.
\end{definition}

\begin{definition}[Interpretation of rooted disconnected indexed dynamic hypersequents]
    Let $G$ be a disconnected IDHS and $X_i$ a dynamic sequent occurring in $G$. The \emph{interpretation of the DIDHS $G$ rooted at $X_i$}, denoted $(G')^{\tau}_{X_i}$, is inductively defined by:
    \begin{small}
    \begin{align*}
         &- \text{ if} \quad G= X_i = \update_{\alpha_1} \Gamma_1 \Imp \Delta_1 \update \cdots \update_{\alpha_k} \Gamma_k \Imp \Delta_k & &\text{then} \quad (G)^{\tau}_{X_i}:= X_i^\tau \\
        &- \text{ if} \quad G= X_i\SepC H, & &\text{then} \quad (G)^{\tau}_{X_i}:= X_i^\tau \\
        &- \text{ if} \quad G= \eta_1: X_1 \SepC \cdots \SepC \eta_i: X_i \SepC \cdots \Sep \eta_k: X_k  & &\text{then} \quad (G)^{\tau}_{X_i}:= (X_i)^{\tau}_{X_i} \lor \bigvee_{\eta_j:X_j \in \Sigma^{G}_{\eta_i:X_i}} K_{f(\eta_i,\eta_j)} (G\setminus \eta_i:X_i)^{\tau}_{X_j}
    \end{align*}
    \end{small}
\end{definition}






A disconnected indexed dynamic hypersequent $G$ rooted at $\eta_i : X_i$ may be viewed as representing the tree model obtained by unfolding the model represented by $G$ from the designated root $X_i$ (see, \emph{e.g.}, \cite{BlackburnModalLogic}). We then define the interpretation of a whole IDHS $G$ by taking the conjunction of all possible \emph{rooted} disconnected dynamic indexed hypersequents that can be extracted from $G$. 

\medskip\begin{definition}[Interpretation of indexed dynamic hypersequents]
Let $G$ be an IDHS. Its interpretation is defined by
$$G^\tau := \bigwedge_{X \in G} (G)^{\tau}_{X}$$
\end{definition}

As an example, consider the following IDHS:

$$G:= 1a: X_1 \Sep 1a,2b: X_2 \Sep 2b: X_3 \Sep 2a,2b: X_4 \Sep 2a: X_5  $$

Its interpretation is the following:
\begin{center}
\begin{align*}
    G^\tau \quad := \quad
    &X_1^\tau \lor K_a( X_2^\tau \lor K_b X_3^\tau \lor K_b( X_4^\tau \lor K_a X_5^\tau )) \\
    &X_2^\tau \lor K_a X_1^\tau \lor K_bX_3^\tau \lor K_b(X_4^\tau \lor K_a X_5^\tau) \\
    &X_3^\tau \lor K_b(X_2^\tau \lor K_aX_1^\tau) \lor K_b(X_4^\tau \lor K_a X_5^\tau)) \\
    &X_4^\tau \lor K_b(X_2^\tau \lor K_aX_1^\tau) \lor K_bX_3^\tau \lor K_a X_5^\tau \\
    &X_5^\tau \lor K_a(X_4^\tau \lor K_b(X_2^\tau \lor K_aX_1^\tau) \lor K_bX_3^\tau)
\end{align*}
\end{center}

\noindent Finally, we need to provide an interpretation for the event store $\store$ of an IDHS with event store $G \st \store$. To do so, we introduce the following notion of \emph{assignment function}. 

\begin{definition}[Assignment function]
    Let $G \st \store$ be an IDHS with event store. An \emph{assignment function} for $G\st \store$ is a function $f_{\store}: \mathrm{x}\mapsto \mathrm{e^M}$ mapping each event variable $\mathrm{x}\in \store$ to a pointed event model $\mathrm{e^M} \in \lAct$ and preserving the relations in $\store$ \emph{i.e.} if $\mathrm{e}\sim_a \mathrm{x} \in \store$ -- for some $\mathrm{e^M\in \lAct}$ and $a\in \mathcal{A}$ -- then $f_\store(\mathrm{x})\in \mathrm{S^M}$ and $\mathrm{e}\sim^{\mathrm{M}}_a f_{\store}(\mathrm{x})$, and if $\mathrm{x}\sim_a \mathrm{y}\in \store$ -- for some $\mathrm{y\in Var}$ -- then $f_\store(\mathrm{x})\sim^{\mathrm{M}}_a f_\store(\mathrm{y})$ for some event model $\mathrm{M}$ such that $f_\store(\mathrm{x}),f_\store(\mathrm{y})\in \mathrm{S^M}$.
    
    By $f_{\store}(G)$ we denote $G$ wherein every variable $\mathrm{x}$ in $G$ is replaced by its assigned pointed event model $f_{\store}(\mathrm{x})$.
\end{definition}

 \noindent The interpretation of an IDHS with event store $G \st \store$ therefore extends that of IDHS by quantifying over assignments functions $f_\store$ for $G \st \store$.
 Notice since there are finitely many event variables stored in $\store$, which are necessarily bound by some pointed event model $\evt^{\mathrm{M}}$, and each event model $\mathrm{M}$ is finite, there are only finitely many different assignment functions $f_\store$, so we can safely take the conjunction of all such functions.

\medskip\begin{definition}[Interpretation of indexed dynamic hypersequents with event store]
    Let $G\st \store$ be an IDHS with event store. Its interpretation is defined as
    $$(G \! \st \store)^\tau := \bigwedge_{f_\store} G^\tau$$
\end{definition}

Hence, a \emph{dynamic indexed hypersequent with event store} $G \st \store$ is valid if, and only if, $f_{\store}(G)^\tau$ is valid, namely with respect to any assignment function $f_\store$ is used to specify what pointed event models the event letters stand for.

\section{The calculus \calculus}\label{sectionCalculus}


\noindent In this section we introduce the calculus $\calculus$. In the following, let $\alpha = \mathrm{e_1}\cdots \mathrm{e_k}$, $\beta = \mathrm{f_1} \cdots \mathrm{f_k}, \ \chi=\mathrm{x_1\cdots x_k}$, where $\evt_i, \mathrm{f}_i \in \Evt$ and $\mathrm{x}_i \in \Var$ for all $1 \leq i \leq k$.  
For any formula $A$ and events $\evt, \mathrm{f}\in \Evt$, by $A[\evt/\mathrm{f}]$ we denote  the formula obtained by substituting every occurrence of $\mathrm{f}$ in $A$ by $\evt$. Finally, $\eta\hat \; na$ denotes the set of indices obtained by taking the union of $ \eta$ and $\lbrace na \rbrace$ and $\store, \evt \sim_a \mathrm{f}$ stands for $\store\union \lbrace \evt \sim_a \mathrm{f}\rbrace$ (for any events $\mathrm{e,f} \in \Evt$ and agent $a\in \mathcal{A}$).\footnote{Consequently, following the notation introduce in Definition \ref{defActionStore}, for $\alpha = \evt_1,\dots,\evt_k$ and $\beta= \mathrm{f}_1,\dots,\mathrm{f}_k$, we denote by $\store, \alpha \approx \beta$ the fact that, for all $1 \leq i \leq k$, there are $\mathrm{g}_{i_1}, \dots, \mathrm{g}_{i_p} \in \Evt$ and $a_{i_1},\dots,a_{i_{p+1}}$ such that $\evt_i \sim_{a_{i_1}} \mathrm{g}_{i_1}\sim_{a_{i_2}} \cdots \sim_{a_{i_p}} g_{a_{i_p}} \sim_{a_{i_{p+1}}} \mathrm{f}_i \in \store.$}

 
 The calculus $\calculus$ is composed by the following axioms and inference rules.

\begin{footnotesize}
\begin{flalign*}
&\textbf{Axioms:} & \\\
    &\begin{prooftree}[small]
    	\hypo{G \SepC \eta: X \update_\alpha \ \pi,\Gamma \Imp \Delta, \pi  \st \store}
    \end{prooftree} &
    &\begin{prooftree}[small]
    	\hypo{G \SepC \eta: X \update_\alpha \ \mathrm{pre(x)},\Gamma \Imp \Delta \update_{\alpha'} \ \Gamma' \Imp \Delta', \mathrm{pre(x)}  \st \store}
    \end{prooftree} & \\
    &\text{where } \pi \text{ is either $p$ or $\mathrm{pre(x})$} &
    &\text{where } \alpha = \alpha_1 \cdot \mathrm{e\cdot f} \cdot \alpha_2 \text{ and } \alpha' = \alpha_1 \cdot \mathrm{(e;f)}\cdot \alpha_2 \text{ (or vice versa)} &
    \hfill \\
&\textbf{Propositional rules:} &  \\
    &\begin{prooftree}[small]
		\hypo{ G \SepC \eta: X \update_\alpha \ \Gamma \Imp \Delta, A \st \store }
		\infer1[(L$\lnot$)]{G\SepC \eta: X \update_\alpha \ \lnot A, \Gamma \Imp \Delta \st \store}
	\end{prooftree}
	&
	&\begin{prooftree}[small]
		\hypo{ G \SepC \eta: X \update_\alpha \ A,\Gamma \Imp \Delta \st \store}
		\infer1[(R$\lnot$)]{G\SepC \eta: X \update_\alpha \ \Gamma \Imp \Delta,\lnot A \st \store}
	\end{prooftree} \\
	\hfill \\
	&\begin{prooftree}[small]
		\hypo{ G \SepC \eta: X  \update_\alpha \ A,B, \Gamma \Imp \Delta \st \store}
		\infer1[(L$\land$)]{G\SepC \eta: X  \update_\alpha \ A\land B, \Gamma \Imp \Delta \st \store}
	\end{prooftree} 
	& 
	&\begin{prooftree}[small]
		\hypo{ G \SepC \eta: X \update_\alpha \ \Gamma \Imp \Delta,A  \st \store}\quad
		\hypo{ G \SepC \eta: X \update_\alpha \ \Gamma \Imp \Delta,B \st \store}
		\infer[separation=2em]2[(R$\land$)]{G\SepC X \update_\alpha \Gamma \Imp \Delta,A \land B \st \store}
	\end{prooftree} & \\
\hfill \\
    &\textbf{Propositional event rules:} & \\
    &\begin{prooftree}[small]
        \hypo{ G \SepC \eta: X \update_\alpha \bigwedge_{\evt\sim_a \mathrm{x}} \, A, A[\mathrm{f}/\mathrm{x}], \Gamma \Imp \Delta \st \store, \evt\approx_a \mathrm{f}}
        \infer1[(L$\bigwedge\evt$)]{G \SepC \eta: X  \update_\alpha \bigwedge_{\evt\sim_a \mathrm{x}} \, A, \Gamma \Imp \Delta \st \store, \evt\approx_a \mathrm{f} }
    \end{prooftree} & 
    &\begin{prooftree}[small]
        \hypo{ G \SepC \eta: X \update_\alpha \ \Gamma \Imp \Delta, A[\mathrm{y}/\mathrm{x}] \st \store, \evt\sim_a \mathrm{y}}
        \infer1[(R$\bigwedge\evt$)]{G \SepC \eta: X  \update_\alpha \ \Gamma \Imp \Delta, \bigwedge_{\evt\sim_a \mathrm{x}} A \st \store }
    \end{prooftree} & \\
    & & &\text{where } \mathrm{y} \text{ does not occur in the conclusion.} &
\end{flalign*}

\noindent \textbf{Modal rules:}
\begin{flalign*}
    &\begin{prooftree}[small]
        \hypo{ G \SepC \eta: X \update_\alpha \ K_a A, \Gamma \Imp \Delta \update_\beta \ A, \Gamma'\Imp \Delta' \st \store, \alpha \approx_a \beta}
		\infer1[($\evt$LK$_1$)]{G\SepC \eta: X \update_\alpha \ K_a A, \Gamma \Imp \Delta \update_\beta \Gamma'\Imp \Delta' \st \store, \alpha \approx_a \beta}
    \end{prooftree} & \\
    \hfill \\
    &\begin{prooftree}[small]
        \hypo{G \SepC \eta: X \update_{\alpha\cdot\gamma} \ K_a A, \Gamma \Imp \Delta \update_{\beta} \Gamma' \Imp \Delta', \mathrm{pre(d_1})  \ \st \store}
        \hypo{G \Sep \overline{X} \update_{\beta\cdot\mathrm{d_1}} \Imp \mathrm{pre(d_2})  \ \st \store}
        \hypo{\cdots}
    	\hypo{G \Sep \overline{X} \update_{\beta\cdot\delta} \ A \Imp  \ \st \store}
   		\infer[separation=1.8em]4[($\evt$LK$_1)'$]{G \SepC \eta: X \update_{\alpha\cdot \gamma} \ K_a A, \Gamma \Imp \Delta \update_{\beta} \ \Gamma' \Imp \Delta' \ \st \store}
    \end{prooftree} & \\
    &\text{where } \overline{X}= \eta: X \update_{\alpha\cdot\gamma} K_a A, \Gamma \Imp \Delta \update_\beta \Gamma' \Imp \Delta' \text{ and } \alpha\cdot \gamma \approx_a \beta \cdot \delta\in \store \text{ with } \delta=\mathrm{d_1\cdots d_k} &\\
    &\text{ and } \beta\cdot \delta' \notin l_\evt(X) \text{ for any prefix } \delta'\prefix\delta . & \\
    \hfill \\
	&\begin{prooftree}[small]
    	\hypo{ G \SepC \eta\hat\;na : X \update_\alpha \ K_a A, \Gamma \Imp \Delta \SepC \theta\hat\;na: Y \update_\beta \ A, \Gamma' \Imp \Delta' \st \store, \alpha \approx_a  \beta}
    	\infer1[($\evt$LK$_2$)]{ G \SepC \eta\hat\;na: X \update_\alpha \ K_a A, \Gamma \Imp \Delta \SepC \theta\hat\;na: Y \update_\beta \ \Gamma' \Imp \Delta' \st \store, \alpha \approx_a  \beta}
	\end{prooftree} 
    \end{flalign*}
    \begin{flalign*}
	&\begin{prooftree}[small]
    	\hypo{G \SepC \overline{X} \SepC \theta\hat\;na: Y \update_{\beta} \, \Gamma' \Imp \Delta', \mathrm{pre(d_1})  \ \st \store }
        \hypo{G \SepC \overline{X} \SepC \overline{Y} \update_{\beta\cdot \mathrm{d_1}} \Imp \mathrm{pre(d_2}) \st \store}
    	\hypo{\cdots}
    	\hypo{G \SepC \overline{X} \SepC \overline{Y} \update_{\beta\cdot\delta} \, A \Imp \st \store}
   		\infer[separation=2em]4[($\evt$LK$_2)'$]{G \SepC \eta\hat\;na: X \update_{\alpha\cdot\gamma} K_a A, \Gamma \Imp \Delta \SepC \theta\hat\;na: Y \update_{\beta} \Gamma' \Imp \Delta'  \st \store}
	\end{prooftree} & \\
	&\text{where } \overline{X}= \eta\hat\;na: X \update_{\alpha\cdot\gamma} K_a A, \Gamma \Imp \Delta, \ \overline{Y}= \theta\hat\;na: Y \update_\beta \Gamma' \Imp \Delta' \text{ and } \alpha\cdot \gamma \approx_a \beta \cdot \delta \in \store \text{, with } \delta = \mathrm{d_1, \dots, d_k} & \\
    &\text{ and } \beta\cdot \delta' \notin l_\evt(Y) \text{ for any prefix } \delta'\prefix\delta , \text{ (see Definition \ref{def:ActionLabelsDS}).} &
    \end{flalign*}
    
    \begin{flalign*}
    &\begin{prooftree}[small]
		\hypo{ G \SepC \eta\hat\;na: X \update_\alpha \ \Gamma \Imp \Delta \Sep na: \update_\chi \Imp A \st \store, \alpha \sim_a \chi}
		\infer1[($\evt$RK)]{G\SepC \eta': X \update_\alpha \ \Gamma \Imp \Delta,K_a A \st \store}
	\end{prooftree} & \\
    &\text{where $\chi$ is not in the conclusion and } \eta'= \eta \text{ if } na \notin \Vert G\Vert \text{ and } \eta'= \eta\hat\;na \text{ if } na \in \Vert G\Vert .
\end{flalign*}

\noindent \textbf{Event rules:}
\begin{flalign*}
	&\begin{prooftree}[small]
    	\hypo{G \SepC \eta: X \update_\alpha \Gamma \Imp \Delta, \mathrm{pre(\evt)} \update_{\alpha \cdot \evt} \Gamma' \Imp \Delta '\st \store}
    	\hypo{G \SepC \eta: X \update_\alpha \Gamma \Imp \Delta \update_{\alpha\cdot \evt} \, A,\Gamma' \Imp \Delta' \st \store}
    	\infer[separation=2em]2[(L$[\cdot]$)]{ G \SepC \eta: X \update_\alpha [\evt]A, \Gamma \Imp \Delta \update_{\alpha\cdot \evt} \Gamma' \Imp \Delta'  \st \store} 
	\end{prooftree} & \\
    \hfill \\
	&\begin{prooftree}[small]
    	\hypo{G \SepC \eta: X \update_\alpha \, \mathrm{pre(\evt)}, \Gamma \Imp \Delta \update_{\alpha\cdot \evt} \Gamma' \Imp \Delta', A  \st \store}
    	\infer1[(R$[\cdot]$)]{ G \SepC \eta: X \update_\alpha \Gamma \Imp \Delta, [\evt]A \update_{\alpha\cdot \evt} \Gamma' \Imp \Delta' \st \store}
	\end{prooftree}  &
\end{flalign*}

\noindent \textbf{Dynamic rules:}
\begin{flalign*}
	&\begin{prooftree}[small]
    	\hypo{G \SepC \eta: X \update_\alpha \ p, \Gamma \Imp \Delta \update_{\alpha \cdot \evt} \ p, \Gamma' \Imp \Delta' \st \store}
    	\infer1[(Lat)]{G \SepC \eta: X \update_\alpha \, \Gamma \Imp \Delta \update_{\alpha \cdot \evt} \ p, \Gamma' \Imp \Delta' \st \store}
	\end{prooftree}
	&
	&\begin{prooftree}[small]
    	\hypo{G \SepC \eta: X \update_\alpha \, \Gamma \Imp \Delta, p \update_{\alpha \cdot \evt} \, \Gamma' \Imp \Delta', p \st \store}
    	\infer1[(Rat)]{G \SepC \eta: X \update_\alpha\, \Gamma \Imp \Delta \update_{\alpha \cdot \evt}\, \Gamma' \Imp \Delta',p \st \store}
	\end{prooftree} & \\
	\hfill & \\
	&\begin{prooftree}[small]
    	\hypo{G \SepC \eta: X \update_\alpha \ \Imp \update_{\alpha \cdot \evt} \, \Gamma \Imp \Delta \st \store}
    	\infer1[(New)]{G \SepC \eta: X \update_{\alpha \cdot \evt} \, \Gamma \Imp \Delta \st \store}
	\end{prooftree}
	&
	&\begin{prooftree}[small]
    	\hypo{G \SepC \eta: X \update_\alpha \, \mathrm{pre}(\evt),\Gamma \Imp \Delta \update_{\alpha \cdot \evt} \,\Gamma' \Imp \Delta' \st \store}
    	\infer1[(Recall)]{G \SepC \eta: X \update_\alpha \, \Gamma \Imp \Delta \update_{\alpha \cdot \evt} \, \Gamma' \Imp \Delta' \st \store}
	\end{prooftree} & \\
    &\text{where $\alpha \notin l_\evt(X)$}\\
\end{flalign*}
\end{footnotesize}

    Certain rules of the calculus \calculus display particular features that need to be stressed.
    First, the modal rules ($\evt$LK$_1$), ($\evt$LK$_1)'$, ($\evt$LK$_2$) and ($\evt$LK$_2)'$ may also be applied when $\alpha = \beta$ -- in which case $\alpha \approx_a \beta$ is not required. Similarly, the propositional event rule (L$\bigwedge\evt$) may also be applied when $\evt = \mathrm{f}$.
    Second, in both event rules (L$[\cdot]$) and (R$[\cdot]$), if $\Gamma' \Imp \Delta'$ is empty -- namely $\Gamma'$ and $\Delta'$ are both empty multisets -- then the dynamic sequent $\update_{\alpha\cdot\evt}\ \Gamma' \Imp \Delta'$ is not present in the conclusion (neither is it in the left-hand premise in case of (L$[\cdot]$)). We will refer to these particular cases as (L$[\cdot])'$ and (R$[\cdot])'$. For example, (L$[\cdot])'$ has the following form:
\begin{center}
	\begin{prooftree}[regular]
    	\hypo{G \SepC \eta: X \update_\alpha \Gamma \Imp \Delta, \mathrm{pre}(\evt) \st \store }
    	\hypo{G \SepC \eta: X \update_\alpha \Gamma \Imp \Delta \update_{\alpha\cdot \evt} A\Imp \st\store}
    	\infer[separation=2.5em]2[(L$[\cdot])'$]{ G \SepC \eta: X \update_\alpha [\evt]A, \Gamma \Imp \Delta \st \store} 
	\end{prooftree}
\end{center}

    \noindent Finally, in rule ($\evt$RK), when $\alpha = \epsilon$, no event is added in $\store$ because $\epsilon$ is only equivalent to itself. Thus the rule gets the following form:

    \begin{center}
    \begin{prooftree}[regular]
		\hypo{ G \SepC \eta\hat\;na: X \update_\epsilon \ \Gamma \Imp \Delta \Sep na: \update_\epsilon \ \Imp A \st \store}
		\infer1[(RK)]{G\SepC \eta': X \update_\epsilon \ \Gamma \Imp \Delta,K_a A \st \store}
	\end{prooftree} 
    \begin{small}
    where $\eta' =\begin{cases}
        \eta &\text{ if } na \notin \Vert G\Vert \\
        \eta\hat\;na &\text{ otherwise}
    \end{cases}$\\
    \end{small}
    \end{center}

 We now comment on the rules that compose the calculus \calculus. Axioms, propositional rules, and modal rules $(\evt$LK$_{1})$, $(\evt$LK$_{2})$ and $(\evt$RK) are the ones that can be found in \cite{Poggiolesi2008} and \cite{Poggiolesi2010}, but adapted to the indexed dynamic setting. Consider for example the rule $(\evt$RK). Read bottom-up, the rule says that if a formula of the form $K_{a}A$ is false at a world-sequent ($w, \alpha$) in some updated model $\M^{\mathrm{M}(\alpha)}$, then there exists a world-sequent ($v ,\beta$) indistinguishable from ($w, \alpha$) for agent $a$ that falsifies $A$. More precisely, there is $(v ,\beta)$ in that updated model such that both $w$ and $v$ lie in the same $a$-cell \emph{i.e.} those worlds are indistinguishable for agent $a$, and  $\alpha \sim^{\mathrm{M}(\alpha)}_{a} \beta$ \emph{i.e.} the sequences of events are indistinguishable from $a$'s perspective,\footnote{Note that if $\alpha = \evt_1 \cdots \evt_k$ and $\beta = \mathrm{f}_1 \cdots \mathrm{f}_k$ then $\alpha \sim_a \beta$ implies $\evt_i$ and $\mathrm{f}_i$ belong to the same event model, for each $1 \leq i \leq k$, so that ($v ,\beta$) belongs to the same updated model $\M^{\mathrm{M}(\alpha)}$ than $(w,\alpha)$. In other words, $\M^{\mathrm{M}(\beta)}=\M^{\mathrm{M}(\alpha)}$: updating $\M$ by $\mathrm{M(\beta)}$ yields the same model as the update by $\mathrm{M(\alpha)}$.} and in the updated model $\M^{\mathrm{M}(\alpha)}$ the formula $A$ is false at ($v ,\beta$). This naturally reflects the semantics of epistemic modalities in DEL.

Rules ($\evt$LK$_1)'$ and ($\evt$LK$_2)'$, on the other hand, should be seen as rules peculiar to the proposed calculus: they are dynamic modal rules in that they display the relationship between knowledge and events. To understand them better, let us consider one of them, namely ($\evt$LK$_2)'$. If read bottom-up, the rule says the following. For a formula of the form $K_a A$ to be true at some world-sequent $(w,\alpha\cdot\gamma)$ in some updated model $\M' := \M^{\mathrm{M}(\alpha\cdot \gamma)} = (W', \lbrace \sim'_a \rbrace_{a \in \mathcal{A}}, V')$, every world-sequent $(v,\beta\cdot\delta)\in W'$ indistinguishable from $(w,\alpha \cdot \gamma)$ for agent $a$ must satisfy $A$.
Note that $(w,\alpha \cdot \gamma)\sim'_a (v,\beta\cdot \delta)$ if, and only if, $w\sim_a v$, $\alpha \sim_a \beta$ and $\gamma \sim_a \delta$.
Hence, for any world-sequent $(v,\beta)\in \M^{\mathrm{M}(\alpha)}$ and any event sequence $\delta$ such that $\gamma \sim_a \delta$ either $(v,\beta\cdot\delta) \notin \M'$, or $A$ is true at $(v,\beta\cdot \delta)$ in $\M'$.\footnote{Notice $\alpha\cdot \gamma \sim_a \beta \cdot \delta$ implies $\mathrm{M}(\alpha\cdot \gamma)=\mathrm{M}(\beta\cdot\delta)$ and therefore $\M^{\mathrm{M}(\beta\cdot\delta)}=\M'$.} The latter case is represented by $\update_{\beta\cdot \delta} \ A \Imp$ in the rightmost premise of the rule. The former case, on the other hand, is represented by all the other premises: indeed, for $(v,\beta\cdot\delta)$ not to belong to $\M'$, one of the following must hold (where $\delta = \mathrm{d_1 \cdots d_k}$): $\mathrm{pre^{M(\beta)}(d_1)}$ is false at $(v,\beta)$ so $(v,\beta)$ does not belong to $\M^{\mathrm{M}(\beta\cdot \mathrm{d_1})}$ -- and \textit{a fortiori} not to $\M'$ either -- or $(v,\beta) \in \M^{\mathrm{M}(\beta\cdot \mathrm{d_1})}$ but $\mathrm{pre^{M(\beta\cdot d_1)}(d_2)}$ is false at $(v,\beta\cdot\mathrm{d_1})$ in that updated model, etc. This is why ($\evt$LK$_2)'$ has $k+1$ premises, when the sequence $\gamma$ is composed of $k$ events. 
As an example, consider the following instance of the rule, where $\alpha = \epsilon$, $\gamma = \mathrm{c_1\cdot c_2},\ \delta = \mathrm{d_1 \cdot d_2}$ and $\mathrm{\alpha \cdot c_1\cdot c_2}\approx_a \beta \cdot \mathrm{d_1\cdot d_2} \in \store$. Let $\overline{G}:= G \ \vert \ \eta: X \update_{\alpha \cdot \gamma} K_a A, \Gamma  \Imp \Delta$.

\begin{center}
\begin{prooftree}[small]
    \hypo{\overline{G} \SepC \theta: Y \update_\beta \Gamma' \Imp \Delta', \mathrm{pre(d_1}) \st \store }
    \hypo{ \overline{G} \SepC \theta: Y \update_\beta \Gamma' \Imp \Delta' \update_{\beta \cdot \mathrm{d_1}} \Imp \mathrm{pre(d_2}) \st \store }
    \hypo{\overline{G} \SepC \theta: Y \update_\beta \Gamma' \Imp \Delta' \update_{\beta \cdot \delta} A \Imp  \st \store }
    \infer[separation=1.5em]3[($\evt$LK$_2)'$]{G \Sep \eta: X \update_{\alpha \cdot \gamma} K_a A, \Gamma \Imp \Delta \Sep \theta: Y \update_\beta  \Gamma' \Imp \Delta' \st \store  }
\end{prooftree}
\end{center}

Let $\ModelM$ be a model where $W = \lbrace w,v \rbrace$. Worlds $w, v$ are represented by $\eta: X \update_{\alpha \cdot \gamma} K_a A, \Gamma \Imp \Delta$ and $\theta: Y \update_\beta \Gamma' \Imp \Delta'$ respectively. Then, we have $w \in \M' = \M^{\mathrm{M}(\alpha\cdot \gamma)}$ and $(v,\beta) \in \M^{\mathrm{M(\beta)}}=\M^{\mathrm{M}(\alpha)}$ (since $\alpha \approx_a \beta$). For $\M',(w,\alpha\cdot \gamma) \satisfies K_a A$ to hold, either $(v,\beta\cdot\delta) \notin \M'$ or $\M', (v,\beta\cdot \delta) \satisfies A$. The first case arises when $\M^{\mathbf{M}(\beta)}, (v,\beta) \not \satisfies \mathrm{pre(d_1})$ or $\M^{\mathbf{M}(\beta\cdot d_1)}, (v,\beta\cdot \mathrm{d_1}) \not \satisfies \mathrm{pre(d_2})$, which corresponds, in the first and second premise, to sequents $\update_{\beta} \ \Gamma' \Imp \Delta', \mathrm{pre(d_1})$ and $\update_{\beta \cdot \mathrm{d_1}} \ \Imp \mathrm{pre(d_2})$ respectively. The second case is captured, in the third premise, by the sequent $\update_{\beta \cdot \delta} \ A \Imp$.

Whilst propositional event rules account for the link between conjunction, events and variables, and are unique to our calculus\footnote{E.g. the calculus  proposed in \cite{Nomuraandothers2} seem to need analogous rules, which are however not explicitly stated.},
event rules introduce the dynamic event modality $[\cdot]$ on the right and on the left of the sequent and reflect the semantic interpretation of the event modality introduced in Definition \ref{defSemDEL}. Finally, dynamic rules account for dynamic features of the calculus. Rules (Lat) and (Rat) state that whatever the truth-value of an atom is, it remains as such in any updated model, thereby representing so-called atomic permanence. Rule (New) states that (reading it bottom-up) whenever the world-sequent $(w, \alpha\cdot \evt)$ belongs to an updated model $\M^{\mathrm{M(\alpha\cdot a)}}$, then world $(w, \alpha)$ occurs in the previous updated model $\M^{\mathrm{M}(\alpha)}$. Finally, rule (Recall) (still reading it bottom-up) states that whenever the world-sequent $(w, \alpha\cdot \evt$) belongs to an updated model $\M^{\mathrm{M(\alpha \cdot \evt)}}$, necessarily the precondition formula $\mathrm{pre^{M(\alpha)}(\evt)}$ is true at the world $(w, \alpha)$ of the previous updated model $\M^{\mathrm{M}(\alpha)}$

\begin{definition} \emph{Derivations} in $\calculus$ are defined in the standard way and are denoted by $d, d', etc$. The \emph{height} of a derivation $d$, denoted $h(d)$, is also inductively defined in the standard way (see \cite{BasicProofTheoryTroelstra}). As usual, we write $\vdash_{\calculus} G \st \store$ to denote that the IDHS with event store $G \st \store$ is derivable in $\calculus$.
\end{definition}

To end this section, we provide an example of a proof in the calculus \calculus that shows one instance of the axiom of atomic permanence is derivable, for a given pointed event model $\evt^{\mathrm{M}}$.

\begin{center}
    \begin{prooftree}[regular]
        \hypo{\mathrm{pre^M(e)} \Imp p, \mathrm{pre^M(e)}}
        \hypo{p \Imp p \update_{\evt^{\mathrm{M}}} p \Imp}
        \infer1[(Lat)]{ \Imp p \update_{\evt^{\mathrm{M}}} p \Imp }
        \infer2[(L$[\cdot]$)]{[\mathrm{e^M}] p,  \mathrm{pre^M(e)} \Imp p}
        \infer1[(R$\imp$)]{[\mathrm{e^M}] p \Imp  \mathrm{pre^M(e)} \imp p}
        \infer1[(R$\imp$)]{\Imp [\mathrm{e^M}] p \imp (\mathrm{pre^M(e)} \imp p)}
    \end{prooftree}
\end{center}



\section{Admissibility results} \label{sectionAdmissibility}

In this section, we show that all logical rules are hp-invertible, and several structural rules, e.g. weakening, contraction and merge, are hp-admissible.

We start by introducing a complexity measure which takes into account both formulas and events, and will be used in the proof of the next lemma, as well as in the proof of Lemma \ref{lemmaComposition} and in  Section \ref{sectionCut}.

\begin{definition}[DHS-Complexity relatively to event store]\label{defComplexity}
    The \emph{$\store$-complexity} of a formula $c_\store(A)$, of an event $c_\store(\evt)$ and of an event model $c_\store(\mathrm{M})$ is defined by induction as follows:
\begin{small}
    \begin{align*}
        c_\store(p)&:= 1  \quad \text{ for } p\in P & c_\store(K_a A) &:= c_\store(A) + 1 \\
        c_\store(\lnot A) &:= c_\store(A) + 1 & c_\store([\mathrm{e^M}]A) &:= c_\store(\mathrm{e^M})+ c_\store(A) +1 \\
        c_\store(A \land B) &:= c_\store(A) + c_\store(B) +1 &  c_\store(\bigwedge_{\mathrm{e^M}\sim_a \mathrm{x}}A)&:= 2|\mathrm{M}|c_\store (A) + 1 \\
        c_\store (\mathrm{M}) &:= max\lbrace c_\store(\mathrm{pre^M(e})) \ |\ \evt\in \mathrm{E^M} \rbrace  & c_\store((\mathrm{M;N})) &:=c_\store(\mathrm{M}) + c_\store(\mathrm{N})\\
        c_\store(\mathrm{e^M}) &:= c_\store(\mathrm{M}) & c_\store(\mathrm{(e^M;f^N)})&:= c_\store(\mathrm{e^M}) + c_\store (\mathrm{f^N})
    \end{align*}
\end{small}
    and the $\store$-complexity for event variables $\mathrm{x} \in \store$ is defined as below\footnote{Note that by Definition \ref{defActionStore}, every variable $\mathrm{x}\in \store$ is bound. Hence, a pointed event model such that $\mathrm{\evt^M} \approx \mathrm{x} \in \store$ is known to exist.}
\begin{small}
        $$c_\store(\mathrm{x}) := c_\store(pre(\mathrm{x})) := c_\store(\mathrm{e^M}) \text{ for some } \mathrm{e^M}\approx \mathrm{x}\in \store$$
\end{small}

    \noindent whilst the $\store$-complexity of a sequence of events $\alpha=\evt_1\cdots \evt_k$ is defined as the sum of the complexity of the events it is composed of: $c_\store(\alpha):=c_\store(\evt_1)+\cdots+c_\store(\evt_k)$.
    
    Finally, the DHS-complexity of a formula $A$ \emph{relatively to a sequence of events} $\alpha$, \emph{and to a store} $\store$, denoted $c_\store(A,\alpha)$, is defined as:
    \begin{align*}
        c_\store(A,\alpha) := c_\store(A) + c_\store(\alpha).
    \end{align*}
\end{definition}


\begin{lemma}\label{lemmaDerivabilityTrivialSequents}
All IDHS with event store of the form $G \Sep \eta: X \update_\alpha A, \Gamma \Imp \Delta, A \st \store$ are derivable in \calculus.
\end{lemma}

\begin{proof} By standard induction on the complexity of $A$.
\end{proof}



\begin{lemma}\label{lemmaAdmissibilityWeakening}
The rules of \emph{internal, external} and \emph{dynamic weakening} are all hp-admissible in \calculus, where the only variables that may be introduced already occur in $\store$.
\begin{small}
\begin{align*}
\begin{prooftree}[small]
    \hypo{G \SepC \eta: X \update_\alpha \ \Gamma \Imp \Delta \st \store}
    \infer1[(IW)]{G \SepC \eta: X \update_\alpha \ \Gamma ,\Gamma' \Imp \Delta, \Delta' \st \store}
\end{prooftree} \qquad
&
\begin{prooftree}[small]
    \hypo{G \SepC \eta: X \st \store}
    \infer1[(EW)]{G \SepC \eta': X \Sep \theta\hat\; na: Y \st \store}
\end{prooftree} \ \dagger
&
\begin{prooftree}[small]
        \hypo{G \SepC \eta: X  \st \store}
        \infer1[(DW)]{G \SepC \eta: X \update_\alpha \ \Gamma \Imp \Delta \st \store}
    \end{prooftree} \ \ddagger \\
&
\dagger \text{where } \eta'= \begin{cases}
    \eta &\text{ if } na \in \Vert G \Vert \\
    \eta\hat\; na & \text{ otherwise}
\end{cases} \qquad
&
\ddagger \text{ where } \alpha \notin l_{\evt}(X) \qquad \qquad \qquad
\end{align*}
\end{small}
\end{lemma}

\begin{proof}
    For all three rules, the proof proceeds straightforwardly by induction on the height of the derivation of the premise. \end{proof}


\noindent

As is standard in a hypersequent framework, there is a rule of (Merge), which can be shown to be hp-admissible and plays a crucial role for the admissibility of the Cut-rule. We first introduce it informally and then move to the formal definition.

 In the setting of standard Epistemic Logic and observed from a semantic perspective, the rule of (Merge)  captures the operation of modifying a model $\ModelM$ by identifying two worlds $w,v \in W$, represented by two sequents, that are indistinguishable for an agent $a$, that is, worlds belonging to the same equivalence class $W_a$. We extend this idea to Dynamic Epistemic Logic by allowing not only worlds but also events to be identified whenever they are indistinguishable for the same agent $a$. Consider an (updated) model $\M':=\M^{\mathrm{M}(\alpha)}$, where $\alpha=\evt_1,\ldots,\evt_k$ is a sequence of events. The Merge rule identifies two worlds $(w,\alpha)$ and $(v,\beta)$ whenever they are indistinguishable for agent $a$, that is, whenever $w\sim_a v$ and $\alpha\sim_a\beta$. Operationally, this requires simultaneously identifying the worlds $w$ and $v$ and the event sequences $\alpha$ and $\beta$. More precisely, if $\alpha=\evt_1,\ldots,\evt_k$ and $\beta=\mathrm{f}_1,\ldots,\mathrm{f}_k$, then the identification is performed componentwise by merging each pair of corresponding events $\evt_i$ and $\mathrm{f}_i$, for $1\leq i\leq k$. As a consequence, every intermediate updated model along the execution of the event sequences must be modified accordingly. For every pair of corresponding prefixes $\alpha'\prefix\alpha$ and $\beta'\prefix\beta$, the worlds $(w,\alpha')$ and $(v,\beta')$ are merged in the intermediate updated model generated by these prefixes. At the proof-theoretic level, this requires combining the corresponding dynamic sequents $\update_{\alpha'}\Gamma\Imp\Delta$ and $\update_{\beta'}\Gamma'\Imp\Delta'$ into the sequent, say $\update_{\alpha'}\Gamma , \Gamma '\Imp\Delta , \Delta '$.
Moreover, every occurrence of events $\mathrm{f}_i\in \beta$ must be replaced by $\evt_i \in\alpha$ both in the indexed dynamic hypersequent $G$ and in the event store $\store$. Finally, the sets of indices must also be merged in such a way that conditions (1), (2), and (3) remain satisfied. In particular, if $w$ and $v$ belong to different partition cells, \emph{e.g.} in addition to belonging to, say,$1a$, they also belong to $1b$ and $2b$, respectively, the resulting merged world-sequent can belong to only one of these two cells, say $1b$, in order for the resulting index set to remain well-formed. Therefore, all occurrences of the deleted index $2b$ must be replaced by the surviving index $1b$.

Let us now move to the formal definition of (Merge). For this, we need some preliminary notions that are introduced in the following. First, we define $X\xMerge Y$ as the dynamic sequent that results from merging $X$ and $Y$.

\begin{definition} Let $X$ and $Y$ be two dynamic sequents, then  $X \xMerge Y$ is inductively defined as follows:

    $-$ if $X= \Gamma \Imp \Delta$ and $Y= \Gamma \Imp \Delta'$ then $X \xMerge Y := \Gamma, \Gamma' \Imp \Delta, \Delta'$;
       
    $-$ if $X= X' \update_{\alpha}\,  \Gamma \Imp \Delta$ and $Y=  Y' \update_{\alpha} \, \Gamma' \Imp \Delta'$ then $X \xMerge Y := X'\xMerge Y' \update_{\alpha} \Gamma, \Gamma' \Imp \Delta, \Delta'$;
       
     $-$ if $X=  X' \update_{\alpha} \Gamma \Imp \Delta$ and $Y=  Y' \update_{\beta} \Gamma' \Imp \Delta'$, where $\alpha \notin \ l_\evt(Y')$ and $\beta \notin \ l_\evt(X')$, then $X \xMerge Y :=  X'\xMerge Y' \update_{\alpha} \Gamma \Imp \Delta \update_{\beta} \Gamma' \Imp \Delta'$.
\end{definition}

\noindent For instance, if $X = \Gamma_1 \Imp \Delta_1 \update_{\evt} \Gamma_2 \Imp \Delta_2 \update_{\mathrm{e\cdot f}} \Gamma_3 \Imp \Delta_3$ and $Y= \Gamma_1' \Imp \Delta_1' \update_{\evt} \Gamma_2' \Imp \Delta_2' \update_{\mathrm{f}} \Gamma_3' \Imp \Delta_3' \update_{\mathrm{f\cdot e}} \Gamma_4' \Imp \Delta_4'$, then $X \xMerge Y = \Gamma_1, \Gamma'_1 \Imp \Delta_1, \Delta_1' \update_{\evt} \Gamma_2, \Gamma'_2 \Imp \Delta_2,\Delta'_2 \update_{\mathrm{e \cdot f}} \Gamma_3 \Imp \Delta_3 \update_{\mathrm{f}} \Gamma'_3 \Imp \Delta_3' \update_{\mathrm{f\cdot e}} \Gamma'_4 \Imp \Delta'_4$.

\medskip

We now introduce further notions to capture the merging of two sequences of events in an indexed dynamic hypersequent with event store. For any events $\mathrm{e,f}$ and any sequences of events $\gamma$, $\gamma[\mathrm{e/f}]$ is obtained by standard substitution of all occurrences of $\mathrm{f}$ by $\evt$. This extends to substitution of whole sequences of events: if $\alpha = \mathrm{e_1}\cdots\mathrm{e_k}$ and $\beta = \mathrm{f_1}\cdots\mathrm{f_k}$, then $\gamma[\alpha/\beta]$ denotes the sequence wherein all $\mathrm{f_i}$ have been replaced by $\mathrm{e_i}$ for $1 \leq i \leq k$. Similarly, for any sequent $\Gamma \Imp \Delta$, the sequent $(\Gamma \Imp \Delta)[\alpha/\beta]$ is obtained by replacing all $\mathrm{f_i}$ occurring in either $\Gamma$ or $\Delta$ by $\mathrm{e_i}$ (for $1 \leq i \leq k$).

\begin{definition}
    Let $X$ be a dynamic sequent and $\alpha,\beta$ be two sequences of events. $X[\alpha/\beta]$ is inductively defined as follows:
    
        $-$  If $X = \Gamma \Imp \Delta$, then $X[\alpha/\beta] := (\Gamma \Imp \Delta)[\alpha/\beta]$;

        $-$ If $X = X' \update_\gamma \Gamma \Imp \Delta \update_{\gamma[\alpha / \beta]} \Gamma' \Imp \Delta'$ then $X[\alpha / \beta]:= X'[\alpha / \beta] \update_{\gamma[\alpha / \beta]} (\Gamma, \Gamma' \Imp \Delta, \Delta')[\alpha / \beta]$;

        $-$ If $X = X' \update_\gamma \Gamma \Imp \Delta$ and $\gamma[\alpha / \beta] \notin l_\evt(X')$ then $X[\alpha / \beta]:= X'[\alpha / \beta] \update_{\gamma[\alpha / \beta]} (\Gamma \Imp \Delta)[\alpha / \beta]$
      
        
    
    Let $G = \eta_1: X_1 \Sep \cdots \Sep \eta_k: X_k$ be an IDHS. Then $G[\alpha/\beta] = \eta_1: X_1[\alpha/\beta] \Sep \cdots \Sep \eta_k: X_k[\alpha/\beta]$.
    \end{definition}

    \begin{definition}  Let $\store$ be an event store and $\alpha = \mathrm{e_1}\cdots\mathrm{e_k}$ and $\beta = \mathrm{f_1}\cdots\mathrm{f_k}$ be two sequences of events. Then $\store[\alpha/\beta]$ is defined inductively as follows, where $a\in \mathcal{A}$:

    $-$ If $\store = \emptyset$, then $\store[\alpha/\beta]=\store = \emptyset$;

    $-$ If $\store = \store', \mathrm{f_i}\sim_a \mathrm{f_j}$, then $\store[\alpha/\beta]= \begin{cases}
            \store'[\alpha/\beta],\mathrm{e_i}\sim_a \mathrm{e_j} &\text{if } i \neq j \text{ and } \mathrm{e_i}\approx_a \mathrm{e_j} \notin \store' \\
            \store'[\alpha/\beta] &\text{otherwise}
        \end{cases}$;

    $-$ If $\store= \store', \mathrm{g}\sim_a \mathrm{f}_i$, where $\mathrm{g}\notin \beta$, then $\store[\alpha/\beta]= \begin{cases}
            \store'[\alpha/\beta],\mathrm{g}\sim_a \mathrm{e}_i  &\text{if } \mathrm{g} \neq \mathrm{e}_i \text{ and } \mathrm{g}\approx_a \mathrm{e}_i \notin \store'\\
            \store'[\alpha/\beta] &\text{otherwise}
        \end{cases}$;

    $-$ If $\store= \store',\mathrm{f}_i \sim_a \mathrm{g}$, where $\mathrm{g}\notin \beta$, then $\store[\alpha/\beta]= \begin{cases}
            \store'[\alpha/\beta],\mathrm{e}_i\sim_a \mathrm{g} &\text{if } \mathrm{g} \neq \mathrm{e}_i \text{ and } \mathrm{g}\approx_a \mathrm{e}_i \notin \store'\\
            \store'[\alpha/\beta] &\text{otherwise}
        \end{cases}$;

    $-$ If $\store = \store', \mathrm{g}\sim_a \mathrm{g}'$ where $\mathrm{g},\mathrm{g}' \notin \beta$, then $\store[\alpha/\beta] = \store'[\alpha/\beta],\mathrm{g}\sim_a \mathrm{g}'$.

\noindent Finally, $(G \st \store)[\alpha/\beta] := G[\alpha/\beta] \st \store [\alpha/\beta] $.

\end{definition}

\noindent Note that when substituting $\alpha$ for $\beta$ in $\store$, we remove occurrences of the form $\evt \sim_a \evt$ possibly obtained through the substitution. Finally, the following definition shows how sets of indices will be merged through an application of (Merge) -- and occurrences of indices modified accordingly.

\begin{definition}
    Let $\eta,\theta$ be two sets of indices. Then, $\eta \indexMerge \theta$ is inductively defined as follows.
    
        $-$ if $\eta=\theta = \emptyset$, then $\eta \indexMerge \theta = \emptyset$;
        
        $-$ if $\eta = \eta'\hat\;na$ and $a\notin \theta$ (resp. $\theta = \theta'\hat\;na$ and $a\notin \eta$), then $\eta \indexMerge \theta = \eta'\indexMerge\theta\hat\;  na$ (resp. $\eta \indexMerge \theta = \eta\indexMerge\theta'\hat\;na$);
        
        $-$ if $\eta = \eta'\hat\;na$ and $\theta = \theta'\hat\;ka$ (where we might have either $n\neq k$ or $n = k$), then $\eta \indexMerge \theta = \eta' \indexMerge \theta'\hat\; na$.
\end{definition}
\noindent For any IDHS $G$ and indices $na,ka$, $G[na/ka]$ is obtained by standard substitution of all occurrences of $ka$ by $na$ in all sets of indices present in $G$.
$G[\eta / \theta]$ is further defined by substituting $na$ to $ka$ in $G$ for all $a$ occurring both in $\eta$ and $\theta$, \emph{i.e.} such that $na \in \eta$ and $ka \in \theta$.

\medskip
We now have all ingredients to formulate the rule of (Merge) and show its admissibility.

\begin{lemma}
    The rule of (\emph{Merge}) is hp-admissible in \calculus.
    \vspace{0.3cm}

        \begin{prooftree}[regular]
            \hypo{G \Sep \eta\hat\; na : X \update_\alpha \Gamma \Imp \Delta \Sep \theta\hat\; na: Y \update_\beta \Gamma' \Imp \Delta' \st \store, \alpha \approx_a \beta  }
            \infer1[(Merge)]{\left(G[\eta/\theta] \Sep (\eta \indexMerge \theta)' : X \xMerge Y \update_\alpha \Gamma, \Gamma' \Imp \Delta, \Delta' \st \store \right)[\alpha/\beta]  }
         \end{prooftree}
\begin{small}
         where  $(\eta \indexMerge\theta)' = \begin{cases}
             \eta \indexMerge \theta \hat\; na &\text{if } na \in \Vert G \Vert \\
             \eta \indexMerge \theta           &\text{if } na \notin \Vert G \Vert
         \end{cases} $

\end{small}

\end{lemma}

\begin{proof} 
     The proof proceeds by induction on the height of the derivation of the premise. Let $\eta,\theta$ be sets of indices such that $\eta \inter \theta = \lbrace na\rbrace$ and let $\alpha \approx_a \beta \in \store$.
     If $G \ \vert \ \eta: X \update_\alpha \Gamma \Imp \Delta \ \vert \ \theta: Y \update_\beta \Gamma' \Imp \Delta' \st \store$ is an axiom, so is $(G[\eta/\theta] \ \vert \ (\eta\indexMerge\theta)': X \xMerge Y \update_\alpha \Gamma, \Gamma' \Imp \Delta, \Delta' \st \store)[\alpha/\beta] $.
     Otherwise, we consider the last applied rule $ \Rule$: we apply the inductive hypothesis to the premise(s) of $ \Rule$ and then use $ \Rule$ again. We only show two paradigmatic example.
     
     If $\Rule$ is an occurrence of ($\evt$RK) as below, where $\chi$ is not in the conclusion and $kb \in \theta$
\begin{center}
        \begin{prooftree}[regular]
            \hypo{G \Sep \eta: X \update_\alpha \ \Gamma \Imp \Delta \Sep \theta: \update_\beta \ \Gamma' \Imp \Delta' \Sep kb: \update_\chi \Imp A \st \store, \beta \sim_b \chi}
            \infer1[($\evt$RK)]{G \Sep \eta: X \update_\alpha \ \Gamma \Imp \Delta \Sep \theta': \update_\beta \ \Gamma' \Imp \Delta', K_b A \st \store}
        \end{prooftree}
\begin{small}
         where $\theta'=\begin{cases}
         \theta                                  &\text{if } kb \in \Vert G \Vert \\
         \theta \setminus \lbrace kb \rbrace     &\text{otherwise}
     \end{cases}$
\end{small}
\end{center}
    \noindent we apply the inductive hypothesis and replace $\beta$ by $\alpha$ in both the IDHS and the event store $\store$, thus obtaining $\alpha \sim_b \chi$\footnote{Note that since $\chi$ is new, $\alpha \sim_b \chi \notin \store$ so we should keep it -- obtained by substituting $\beta$ by $\alpha$ -- while applying the inductive hypothesis.} that allows us to apply ($\evt$RK) again. Below, either $l\neq k$, if $lb\in \eta$,\footnote{Since $\eta \cap \theta = na$, if $b\in \eta$ and $b\in \theta$, necessarily $k\neq l$. This is linked to condition (1) in Definition \ref{defIDHS}.} or $l=k$, if $b\notin\eta$.
    \begin{center}
        \begin{prooftree}[regular]
            \hypo{\left( G[\eta/\theta] \Sep (\eta\indexMerge \theta)'\hat\;lb: X \xMerge Y \update_\alpha \ \Gamma, \Gamma' \Imp \Delta, \Delta' \Sep lb: \update_\chi \Imp A \st \store, \alpha \sim_b \chi \right) [\alpha / \beta]}
            \infer1[($\evt$RK)]{\left(G[\eta/\theta] \Sep (\eta\indexMerge\theta)'': X \xMerge Y \update_\alpha \ \Gamma, \Gamma' \Imp \Delta, \Delta', K_b A \st \store \right)[\alpha/\beta]}
        \end{prooftree}
    \end{center}
\begin{small}
    where $(\eta\indexMerge\theta)' = \begin{cases}
        \eta\indexMerge \theta                                &\text{if } na\in \Vert G \Vert\\
        \eta\indexMerge\theta \setminus \lbrace na \rbrace    &\text{otherwise}
    \end{cases}$ and $(\eta\indexMerge\theta)''= \begin{cases}
        (\eta\indexMerge\theta)'                              &\text{if } lb \in \Vert G \Vert \\
        (\eta\indexMerge\theta)'\setminus \lbrace lb \rbrace   &\text{otherwise}.
    \end{cases}$
\end{small}

\medskip

    If now $\Rule$ is an occurrence of (R$[\cdot]$) where $\evt\in \alpha$ and $\mathrm{f} \in \beta$, consider the derivation below.

    \begin{center}
        \begin{prooftree}
            \hypo{G \Sep \eta: X \update_\alpha \Gamma \Imp \Delta  \Sep \theta: Y \update_\beta \mathrm{pre(f)}, \Gamma' \Imp \Delta' \update_{\beta \cdot \mathrm{f}} \Imp A \st \store }
            \infer1[(R$[\cdot]$)]{G \Sep \eta: X \update_\alpha  \Gamma \Imp \Delta \Sep \theta: Y \update_\beta \Gamma' \Imp \Delta', [\mathrm{f}]A \st \store }
        \end{prooftree}
    \end{center}

\noindent We again apply the inductive hypothesis to the premise, thereby replacing $\beta$ by $\alpha$ and in particular $\mathrm{f}$ by $\mathrm{e}$. We then apply (R$[\cdot]$) again.

    \begin{center}
        \begin{prooftree}
            \hypo{\left( G[\eta / \theta] \Sep \eta\indexMerge \theta: X\xMerge Y \update_\alpha \mathrm{pre(e)}, \Gamma, \Gamma' \Imp \Delta , \Delta' \update_{\alpha \cdot \evt} \Imp A \st \store \right) [\alpha / \beta] }
            \infer1[(R$[\cdot]$)]{ \left( G[\eta / \theta] \Sep \eta\indexMerge \theta: X\xMerge Y \update_\alpha  \Gamma, \Gamma' \Imp \Delta , \Delta', [\evt]A \st \store \right) [\alpha / \beta] }
        \end{prooftree}
    \end{center}

\end{proof}

The rule of (Merge) is the most general that can be formulated for a dynamic setting with an event store. However, it is also possible to formulate two more  restricted versions of (Merge), namely one called ($w$-Merge) where the sequences $\alpha$ and $\beta$ are one and the same so only the worlds $w$ an $v$ are merged, and another called $(\evt$-Merge) where the dynamic sequents $X$ and $Y$ -- the worlds so to say -- are one and the same so only the sequences of events $\alpha$ and $\beta$ are merged. The proofs of hp-admissibility of both ($w$-Merge) and ($\evt$-Merge) are very similar to that of (Merge) and thus omitted.

\begin{lemma}
    The rules of (\emph{$w$-Merge}) and (\emph{$\evt$-Merge}) are hp-admissible in \calculus.

 \vspace{0.3cm}
    \begin{prooftree}[regular]
        \hypo{G \Sep \eta\hat\; na : X \update_\alpha \ \Gamma \Imp \Delta \Sep \theta\hat\; na: Y \update_\alpha  \ \Gamma' \Imp \Delta' \st \store}
        \infer1[($w$-Merge)]{G[\eta/\theta] \Sep (\eta \indexMerge \theta)' : X \xMerge Y \update_\alpha \ \Gamma, \Gamma' \Imp \Delta, \Delta' \st \store  }
    \end{prooftree} 
\begin{small}
    where $(\eta \indexMerge\theta)' = \begin{cases}
        \eta \indexMerge \theta \hat\; na & \text{if } na \in \Vert G \Vert \\
        \eta \indexMerge \theta & \text{if } na \notin \Vert G \Vert
    \end{cases}$,\\
\end{small}

\hfill \\

    \begin{prooftree}[regular]
        \hypo{G \Sep \eta : X \update_\alpha \ \Gamma \Imp \Delta \update_\beta \ \Gamma' \Imp \Delta' \st \store, \alpha \approx_a \beta  }
        \infer1[($\evt$-Merge)]{\left(G \Sep \eta: X \update_\alpha \ \Gamma,\Gamma' \Imp \Delta, \Delta' \st \store \right)[\alpha/\beta]  }
    \end{prooftree}

\end{lemma}

\begin{proof} The proof proceeds by induction on the height of the derivation of the premise and it is analogous to the one of the previous lemma.
\end{proof}

\medskip

Beyond (Merge), there is another structural rule, (New$^{\evt}$), which is both admissible in the calculus \calculus and  crucial to prove the admissibility of the (Cut-rule). 

\begin{lemma}\label{lemmaAdmissibilityNewA}
The rule of (\emph{New}$^{\evt}$) is admissible in \calculus.
    \begin{displaymath}
    \begin{prooftree}[regular]
        \hypo{G \Sep \eta: X \update_\alpha \ \mathrm{pre}(\evt), \Gamma \Imp \Delta \ \update_{\alpha \cdot \evt} \ \Imp \st \store}
        \infer1[(New$^{\evt}$)]{G \Sep \eta: X \update_\alpha \ \mathrm{pre}(\evt), \Gamma \Imp \Delta \st \store}
    \end{prooftree} \quad \text{where } \alpha\cdot \evt \notin l_\evt(X)
    \end{displaymath}
\end{lemma}

\begin{proof}
We proceed by induction on the height of the derivation of the premise.
If the premise $G\ \vert \ \eta: X \update_\alpha \mathrm{pre}(\evt), \Gamma \Imp \Delta \update_{\alpha \cdot \evt} \Imp \st \store$ is an axiom then so is $G \ \vert \ \eta: X \update_\alpha \mathrm{pre}(\evt), \Gamma \Imp \Delta \st \store$. If $G\ \vert \ \eta: X \update_\alpha \mathrm{pre}(\evt), \Gamma  \Imp \Delta\update_{\alpha \cdot \evt} \Imp$ is of the form $G \ \vert \ \eta: X \update_\alpha \mathrm{pre}(\evt),\Gamma \Imp \Delta \update_{\alpha \cdot \evt} \Imp \update_{\alpha \cdot \evt \cdot \mathrm{f}} \, \Gamma' \Imp \Delta'$ -- namely the dynamic sequent $X$ contains a sequent of the form $\update_{\alpha \cdot \evt \cdot \mathrm{f}} \, \Gamma' \Imp \Delta'$ -- we simply apply (New) to obtain the desired conclusion.

If $G\ \vert \ \eta: X \update_\alpha \mathrm{pre}(\evt), \Gamma \Imp \Delta \update_{\alpha \cdot \evt}\Imp \st \store$ has been obtained by a rule $\Rule$, which is a left-rule and does not involve the dynamic sequent $\update_{\alpha \cdot \evt} \Imp$, we apply the inductive hypothesis on the premise of $\Rule$, and then apply $\Rule$ again. Here is an example with (L$[\cdot]$):
    \begin{align*}
        &\begin{prooftree}[regular]
            \hypo{G \Sep \eta: X \update_\alpha \mathrm{pre}(\evt), \Gamma \Imp \Delta, \mathrm{pre(f)}  \ \update_{\alpha \cdot \evt} \Imp}
            \hypo{G \Sep \eta: X \update_\alpha \mathrm{pre}(\evt),\Gamma \Imp \Delta \update_{\alpha \cdot \evt} \Imp \update_{\alpha \cdot \mathrm{f}} \ A \Imp}
            \infer2[(L$[\cdot]$)]{G \Sep \eta: X \update_\alpha \mathrm{pre}(\evt), [\mathrm{f}]A,\Gamma \Imp \Delta \ \update_{\alpha \cdot \evt} \Imp }
        \end{prooftree}\\
    \\
    \byIH
    &\qquad \qquad
        \begin{prooftree}[regular]
            \hypo{G \SepC \eta: X \update_\alpha \mathrm{pre}(\evt), \Gamma \Imp \Delta, \mathrm{pre(f)}}
            \hypo{G \SepC \eta: X \update_\alpha \mathrm{pre}(\evt), \Gamma \Imp \Delta \update_{\alpha \cdot \mathrm{f}} A \Imp}
            \infer2[(L$[\cdot]$)]{G \SepC \eta: X \update_\alpha \mathrm{pre}(\evt), [{\mathrm{f}}]A, \Gamma \Imp \Delta}
        \end{prooftree}
    \end{align*}

If, on the other hand, $\Rule$ is an application of a left-rule that does involve the dynamic sequent $\update_{\alpha \cdot \evt} \Imp$, then we obtain the desired result using Lemma \ref{lemmaDerivabilityTrivialSequents}, applying the inductive hypothesis or directly using $\Rule$ again. In case of ($\evt$LK$_1$) (resp.($\evt$LK$_2$)), we might need to use ($\evt$LK$_1)'$ (resp. ($\evt$LK$_2)'$) instead.

We only show one of the most significant cases, namely the case where $\Rule$ is an application of ($\evt$LK$_1$) involving $\update_{\alpha\cdot \evt} \Imp$, where $\alpha\cdot\evt\approx_a \beta\cdot \mathrm{f} \in \store$:
\begin{center}
    \begin{prooftree}[regular]
        \hypo{G \Sep \eta: X \update_{\alpha} \ \mathrm{pre}(\evt),\Gamma \Imp \Delta \update_{\alpha\cdot\evt} \ A \Imp \update_{\beta\cdot\mathrm{f}} \ K_a A, \Gamma' \Imp \Delta' \st \store}
        \infer1[($\evt$LK$_1$)]{G \Sep \eta: X \update_{\alpha} \ \mathrm{pre}(\evt),\Gamma \Imp \Delta \update_{\alpha\cdot\evt} \Imp \update_{\beta\cdot\mathrm{f}}\  K_a A, \Gamma' \Imp \Delta' \st \store}
    \end{prooftree}
\end{center}

\noindent We use $G \ \vert \ \eta: X \update_{\alpha} \mathrm{pre}(\evt),\Gamma \Imp \Delta, \mathrm{pre}(\evt) \update_{\beta\cdot\mathrm{b}} K_a A, \Gamma' \Imp \Delta' \st \store$, that is derivable by Lemma \ref{lemmaDerivabilityTrivialSequents}, as an additional premise to apply ($\evt$LK$_1)'$, thus directly obtaining the desired result: 
\begin{center}
    \begin{prooftree}[small]
        \hypo{G \SepC \eta: X \update_{\alpha} \ \mathrm{pre}(\evt),\Gamma \Imp \Delta, \mathrm{pre}(\evt)  \update_{\beta\cdot\mathrm{f}}\  K_a A, \Gamma' \Imp \Delta' \st \store}
        \hypo{G \SepC \eta: X \update_{\alpha} \ \mathrm{pre}(\evt),\Gamma  \Imp \Delta \update_{\alpha\cdot\evt} \ A \Imp \update_{\beta\cdot\mathrm{f}}\  K_a A, \Gamma' \Imp \Delta' \st \store}
        \infer[separation= 1.5em]2[($\evt$LK$_1)'$]{G \Sep \eta: X \update_{\alpha} \ \mathrm{pre}(\evt),\Gamma \Imp \Delta \update_{\beta\cdot\mathrm{f}} \ K_a A, \Gamma' \Imp \Delta' \st \store}
    \end{prooftree}
\end{center}

\medskip

If $G\ \vert \ \eta: X \update_\alpha \mathrm{pre}(\evt), \Gamma \Imp \Delta \update_{\alpha \cdot \evt} \Imp \st \store$ has been obtained by a rule $ \Rule$ that does not involve any of the previous cases, then we apply the inductive hypothesis to the premise(s) of $ \Rule$ and then use $ \Rule$ again. As a way of example, we consider the case where $\Rule$ is ($\evt$RK).
\begin{center}
        \begin{prooftree}[regular]
            \hypo{G \Sep \eta\hat\;na: X \update_\alpha \mathrm{pre}(\evt), \Gamma \Imp \Delta \update_{\alpha \cdot \evt} \Imp \Sep na: \update_\chi \Imp A \st \store, \alpha \sim_a \chi}
            \infer1[($\evt$RK)]{G \Sep \eta': X \update_\alpha \mathrm{pre}(\evt),\Gamma \Imp \Delta, K_a A \ \update_{\alpha \cdot \evt} \Imp \st \store }
        \end{prooftree}
\begin{small}
    where $\eta'= \begin{cases} \eta\hat\;na &\text{ if } na \in \Vert G\Vert \\ \eta &\text{ otherwise} \end{cases}$
\end{small}
\end{center}

\noindent We simply apply the induction hypothesis to the premise and then apply ($\evt$RK) again:
\begin{center}
        \begin{prooftree}[regular]
            \hypo{G \Sep \eta\hat\;na:X \update_\alpha \mathrm{pre}(\evt), \Gamma \Imp \Delta \Sep na: \update_\chi \Imp A \st \store,  \alpha \sim_a \chi}
            \infer1[($\evt$RK)]{G \Sep \eta': X \update_\alpha \mathrm{pre}(\evt),\Gamma \Imp \Delta, K_a A \st \store}
        \end{prooftree}
\end{center}
\end{proof}

\medskip


We end this section by showing that all logical rules are height-preserving invertible and that the contraction rules are hp-admissible in \calculus.

\begin{lemma}\label{lemmaInvertible}
    All logical rules are hp-invertible.
\end{lemma}

\begin{proof} 
 Rules (L$\bigwedge\evt$), ($\evt$LK$_1$), ($\evt$LK$_2$) are hp-invertible by hp-admissibility of (IW).
 Rules ($\evt$LK$_1)'$ and ($\evt$LK$_2)'$ are hp-invertible by hp-admissibility of (IW) and (DW).
 As for the other rules, we reason by induction on the height of the derivation. We only show the cases of the rules ($\evt$RK) and (L$[\cdot]$). Invertibility of the remaining rules can be established analogously.
 
 Let us first consider ($\evt$RK). If $G\ \vert \ \eta': X \update_\alpha \Gamma \Imp \Delta, K_a A \st \store$ is an axiom, so is $G\ \vert \ \eta\hat\;na: X \update_\alpha \Gamma \Imp \Delta \ \vert \ na: \update_\chi \Imp A \st \store, \alpha \sim_a \chi$. Otherwise, we consider the last occurrence of a rule $\Rule$ in the derivation. If $ \Rule$ has $K_a A$ as its principal formula, then it must be an occurrence of ($\evt$RK) and the premise is precisely what we are looking for. If $K_a A$ is not principal in $\Rule$, we apply the inductive hypothesis to the premise(s) of $\Rule$ and then use $\Rule$ again as below.

\begin{align*}
\begin{prooftree}[regular]
	\hypo{G' \Sep \eta'': X' \update_\alpha \Gamma' \Imp \Delta', K_a A  \st \store'}
	\infer1[$\Rule$]{G \Sep \eta': X \update_\alpha \Gamma \Imp \Delta, K_a A  \st \store}
\end{prooftree}
\byIH
\begin{prooftree}[regular]
	\hypo{G' \Sep \eta''\;\hat\;na: X' \update_\alpha \Gamma' \Imp \Delta' \Sep na: \update_\chi \Imp A  \st \store', \alpha \sim_a \chi}
	\infer1[$\Rule$]{G \Sep \eta\;\hat\;na: X \update_\alpha \Gamma \Imp \Delta \Sep na: \update_\chi \Imp A \st \store, \alpha\sim_a \chi}
\end{prooftree}
\end{align*}

Let us now move to (L$[\cdot])'$.
If $G \ \vert \ \eta: X \update_\alpha [\evt]A, \Gamma \Imp \Delta \st \store$ is an axiom, so are $G \ \vert \ \eta: X \update_\alpha \Gamma \Imp \Delta, \mathrm{pre}(\evt) \st \store$ and $G \ \vert \ \eta: X \update_\alpha \Gamma \Imp \Delta \update_{\alpha \cdot \evt} A \Imp \st \store$.
If $ \Rule$ has $[\evt] A$ as its principal formula, then it must be an occurrence of (L$[\cdot])$ and the premises are precisely what we were looking for.
Otherwise, we consider the last occurrence of a rule $\Rule$ in the derivation.
If $[\evt] A$ is not principal in $\Rule$, we apply the inductive hypothesis to the premise(s) of $\Rule$ and then use $\Rule$ again.
\end{proof}

From now on, for any given rule $\Rule$, we denote by $\Rule^\star$ the inverted rule\footnote{In case rule $\Rule$ has several premises then $\Rule^\star$ denotes any of the corresponding inverted rules, e.g. for (R$\land$), we denote by (R$\land)^\star$ either $\begin{prooftree}[small]
    \hypo{G\Sep \eta: X \update_\alpha \Gamma \Imp \Delta,A\land B \st \store}
    \infer1{G\Sep \eta: X\update_\alpha \Gamma \Imp \Delta, A \st \store}
\end{prooftree}$ or
$\begin{prooftree}[small]
    \hypo{G\Sep \eta: X \update_\alpha \Gamma \Imp \Delta,A\land B \st \store}
    \infer1{G\Sep \eta: X\update_\alpha \Gamma \Imp \Delta, B \st \store}
\end{prooftree}$.}, that is derivable, by Lemma \ref{lemmaInvertible}.


\begin{lemma}\label{lemmaAdmissibilityC}
    The rules of \emph{contraction} are hp-admissible.
    \begin{align*}
    &\begin{prooftree}[regular]
        \hypo{G \Sep \eta: X \update_\alpha A, A, \Gamma \Imp \Delta \st \store}
        \infer1[(LC)]{G \Sep \eta: X \update_\alpha A,\Gamma \Imp \Delta \st \store}
    \end{prooftree}
    &
    &\begin{prooftree}[regular]
        \hypo{G \Sep \eta: X \update_\alpha \Gamma \Imp \Delta,A,A \st \store}
        \infer1[(RC)]{G \Sep \eta:  X \update_\alpha \Gamma \Imp \Delta, A \st \store}
    \end{prooftree}
    \end{align*}

\end{lemma}

\begin{proof} 
We simultaneously prove the hp-admissibility of (LC) and (RC) by induction on the height of the derivation of the premise of each rule. If $G \ \vert \ X \update_\alpha A, A, \Gamma \Imp \Delta \st \store$ and $G \ \vert \ X \update_\alpha \Gamma \Imp \Delta, A,A \st \store$ are axioms, then so are $G \ \vert \ X \update_\alpha A, \Gamma \Imp \Delta \st \store$ and $G \ \vert \ X \update_\alpha \Gamma \Imp \Delta, A \st \store$. Otherwise, we distinguish cases on the last applied rule $ \Rule$. If neither of the occurrences of $A$ is principal, we apply the inductive hypothesis to the premise(s) of $ \Rule$, and then use $\Rule$ again. Consider the following example where $\Rule$ is an occurrence of (R$[\cdot])'$.

\medskip\noindent \begin{prooftree}[regular]
    \hypo{G \Sep \eta: X \update_\alpha \mathrm{pre(e)}, A, A, \Gamma \Imp \Delta \update_{\alpha \cdot \evt} \Imp B \st \store}
    \infer1[(R$[\cdot])'$]{G \Sep \eta: X \update_\alpha A, A, \Gamma \Imp \Delta, [\evt] B \st \store}
\end{prooftree}
$\byIH$
\begin{prooftree}[regular]
    \hypo{G \Sep \eta: X \update_\alpha \mathrm{pre(e)}, A,  \Gamma \Imp \Delta \update_{\alpha \cdot \evt} \Imp B \st \store}
    \infer1[(R$[\cdot])'$]{G \Sep \eta: X \update_\alpha A,  \Gamma \Imp \Delta, [\evt] B \st \store}
\end{prooftree}

\hfill \\

If one of the occurrences of $A$ is principal, we apply $\Rule^\star$ -- that is hp-admissible by Lemma \ref{lemmaInvertible} -- to the premise(s) of $ \Rule$, apply the inductive hypothesis and then use $ \Rule$ again. In case of (R$\bigwedge\evt$) or ($\evt$RK), we also need to use hp-admissibility of ($\evt$-Merge), and hp-admissibility of (Merge), respectively, before applying the inductive hypothesis. For the sake of clarity, let us examine those two cases. First consider the case where $\Rule$ is an occurrence of ($\evt$RK).

\begin{center}
\begin{prooftree}[regular]
    \hypo{G \Sep \eta\;\hat\;na : X \update_\alpha\Gamma \Imp \Delta, K_a A \Sep na: \update_{\chi} \Imp A \stSpace \store, \alpha \sim_a \chi}
    \infer1[($\evt$RK)]{G \Sep \eta' : X \update_\alpha\Gamma \Imp \Delta, K_a A, K_a A \stSpace \store}
\end{prooftree}
\begin{small}
    where $\eta'= \begin{cases} \eta\hat\;na &\text{ if } na \in \Vert G\Vert \\ \eta &\text{ otherwise} \end{cases}$
\end{small}
\end{center}

\noindent We use hp-invertibility of ($\evt$RK) and then hp-admissibility of (Merge) to obtain a single new world-component $na: \update_\chi \Imp A, A$\footnote{Note that in the derivation, since $\chi$ and $\zeta$ are new, they do not appear anywhere else in the indexed dynamic hypersequent or in the event store, so the substitution $[\chi / \zeta]$ is fairly simple.}. This allows us to apply the inductive hypothesis, and finally use ($\evt$RK) again to get the desired conclusion:

\begin{center}
\begin{prooftree}[regular]
    \hypo{G \Sep \eta\;\hat\;na : X \update_\alpha\Gamma \Imp \Delta \Sep na: \update_{\chi} \Imp A \Sep na: \update_{\zeta} \Imp A \stSpace \store, \alpha \sim_a \chi, \alpha \sim_a \zeta}
    \infer1[(Merge)]{G \Sep \eta\;\hat\;na : X \update_\alpha\Gamma \Imp \Delta \Sep na: \update_{\chi} \Imp A, A \stSpace \store, \alpha \sim_a \chi}
    \infer1[(IH)]{G \Sep \eta\;\hat\;na : X \update_\alpha\Gamma \Imp \Delta \Sep na: \update_{\chi} \Imp A \stSpace \store, \alpha \sim_a \chi}
    \infer1[($\evt$RK)]{G \Sep \eta' : X \update_\alpha\Gamma \Imp \Delta, K_a A \st \store}
\end{prooftree}
\end{center}

Now consider the case where $\Rule$ is the following occurrence of (R$\bigwedge\evt$).

\begin{center}
\begin{prooftree}[regular]
    \hypo{G \Sep \eta : X \update_\alpha\Gamma \Imp \Delta, \bigwedge_{\evt \sim_a \mathrm{x}} A, A[\mathrm{y} / \mathrm{x}] \stSpace \store, \evt \sim_a \mathrm{y}}
    \infer1[(R$\bigwedge\evt$)]{G \Sep \eta : X \update_\alpha\Gamma \Imp \Delta, \bigwedge_{\evt \sim_a \mathrm{x}} A, \bigwedge_{\evt \sim_a \mathrm{x}} A \stSpace \store}
\end{prooftree}
\end{center}

\noindent We use hp-invertibility of (R$\bigwedge\evt$) and then hp-admissibility of ($\evt$-Merge) to combine both event variables $\mathrm{y,z}$ thereby obtained. We then apply the inductive hypothesis and get the desired conclusion by applying (R$\bigwedge\evt$) again.

\begin{center}
\begin{prooftree}[regular]
    \hypo{G \Sep \eta : X \update_\alpha\Gamma \Imp \Delta, A[\mathrm{y} / \mathrm{x}], A[\mathrm{z}/\mathrm{x}] \stSpace \store, \evt \sim_a \mathrm{y}, \evt \sim_a \mathrm{z}}
    \infer1[($\evt$-Merge)]{G \Sep \eta : X \update_\alpha\Gamma \Imp \Delta,  A[\mathrm{y} / \mathrm{x}], A[\mathrm{y}/\mathrm{x}] \stSpace \store, \evt \sim_a \mathrm{y}}
    \infer1[(IH)]{{G \Sep \eta : X \update_\alpha\Gamma \Imp \Delta,  A[\mathrm{y} / \mathrm{x}] \stSpace \store, \evt \sim_a \mathrm{y}}}
    \infer1[(R$\bigwedge\evt$)]{G \Sep \eta : X \update_\alpha\Gamma \Imp \Delta, \bigwedge_{\evt \sim_a \mathrm{x}} A[\mathrm{y} / \mathrm{x}] \stSpace \store}
\end{prooftree}
\end{center}
\end{proof}



\section{Soundness and Completeness} \label{sectionCompleteness}

In this section, we show that $\calculus$ is sound and complete.


\begin{theorem}[Soundness]\label{theoremSoundness}
     For all formulas $A \in \lDEL$, if $\vdash_{\calculus} \Imp A$ then $\satisfies_{DEL} A$.
\end{theorem}

\begin{proof}

We need to show that axioms are valid and that all rules are correct, \emph{i.e.} that they preserve validity. In each case, conditions $(2)$ and $(3)$ of Definition \ref{defIDHS} are crucial to carry out the proof.
In particular, for any sets of indices $\eta,\theta$ in an IDHS, condition $(2)$ ensures that there is a path from $\eta$ to $\theta$, namely sets of indices $\iota_1, \dots, \iota_{k+1}$ and agents $a_1, \dots, a_{k}$ such that $\iota_1 = \eta, \iota_k = \theta$ and for all $1 \leq i \leq k$, $a_i = f(\iota_{i}, \iota_{i+1})$ (see Definition \ref{defn:setofsameindexes}).
We denote such a path by $ \eta \nearrow \theta$ and $na \in \eta \nearrow \theta$ means that there is $1\leq i \leq k$ such that $\iota_i \inter \iota_{i+1}=\lbrace na \rbrace$.
Then, in order to show that all axioms are valid, one proceeds in a standard way.
   
As for the correctness of the rules, for the sake of clarity, we provide details for the case of rule ($\evt$RK), while the other rules can be treated similarly.

Let $\alpha = \mathrm{e_1\cdots e_k}$ and $\chi = \mathrm{x_1 \cdots x_k}$.
Suppose the premise $I_P:= G \ \vert \ \eta\hat\;na : X \update_{\alpha} \Gamma \Imp \Delta \ \vert \ na: \update_\chi \Rightarrow A\st \store, \alpha \sim_a \chi$ is valid but not the conclusion $I_C:= G \ \vert \ \eta': X \update_\alpha \Gamma \Imp \Delta, K_a A \st \store$, where $\chi$ does not occur and $\eta' = \eta \hat \; na$ if $na \in \Vert G \Vert$ and $\eta'=\eta$ otherwise.
    Then, there is an assignment function $f_{\store}$ such that $\not\satisfies f_\store(I_C)^\tau$. Let $I_C':=f_\store(I_C)$ be the IDHS obtained by replacing all the variables in $I_C$ with their respective assigned events -- thus we now have $G' \ \vert \ \eta': X' \update_{\alpha'} \Gamma' \Imp \Delta', K_a A'$.
    Since $\not\satisfies (I_C')^\tau$, there is a model $\ModelM$, a world $w\in W$ and a dynamic sequent $X'_i \in I_C'$ such that $\M,w\not\satisfies (I_C')_{X'_i}^\tau$. We now distinguish two cases.

     $(1)$ If $X'_i = X' \update_{\alpha'} \Gamma' \Imp \Delta',K_a A'$, we have
     \begin{align*}
         \M,w \not \satisfies X'^\tau \lor [\alpha'](\bigwedge \Gamma' \imp \bigvee \Delta' \lor K_a A') \lor \phi
     \end{align*}
     for some formula $\phi \in \lDEL$.\footnote{The formula $\phi$ corresponds to the rest of the interpretation $(I_C')^\tau_{X'}$ of the disconnected indexed dynamic hypersequent rooted at $X'$.} In particular, then, $\M^{\mathrm{M}(\alpha')},(w,\alpha') \not \satisfies K_a A'$. Hence, there is $(v,\gamma)\in \M^{\mathrm{M}(\alpha')}$ such that $(w,\alpha')\sim_a (v,\gamma)$, $\M,v\satisfies \mathrm{pre}(\gamma)$ and $\M^{\mathrm{M}(\alpha')},(v,\gamma) \not \satisfies A'$ -- where $\gamma = \mathrm{g_1\cdots g_k}$ and $\mathrm{e_i}\sim_a\mathrm{g_i}$ for all $1 \leq i \leq k$.

     However, since the premise is valid, in particular $\satisfies f'_{\store}(I_P)$ for the assignment $f'_{\store}$ defined by $f'_{\store}(\mathrm{x}):=f_{\store}(\mathrm{x})$ for all $\mathrm{x}\in \store$ and $f'_{\store}(\mathrm{x_i}):= \mathrm{g_i}$ for all $\mathrm{x_i} \in \chi$. Hence
     \begin{align*}
        \M,w \satisfies X'^\tau \lor [\alpha'](\bigwedge \Gamma' \imp \bigvee \Delta') \lor K_a [\gamma]A' \lor \phi.
     \end{align*}

    \noindent Necessarily, then, $\M,w \satisfies K_a[\gamma]A'$. Since $(w,\alpha)\sim_a (v,\gamma)$, also $w \sim_a v$ so $\M,v \satisfies [\gamma]A'$. But then, since $\M^{\mathrm{M}(\alpha')},(v,\gamma) \not \satisfies A'$ and $\M,v \satisfies \mathrm{pre}(\gamma)$, we get $\M^{\mathrm{M}(\alpha')},(v,\gamma) \satisfies A'$ and reach a contradiction.
    
    \medskip
    
     $(2)$ Otherwise, by the second condition of Definition \ref{defIDHS}, there is a path $\eta_i \nearrow \eta\hat\;na$ so we have
    \begin{align*}
        \M,w \not \satisfies X_i'^\tau \lor K_{a_1}((\cdots K_{a_k}(X'^\tau \lor [\alpha'](\bigwedge \Gamma' \imp \bigvee \Delta' \lor K_a A') \lor \psi_k) \cdots) \lor \psi_1) \lor \phi
    \end{align*}
     for some formulas $\psi_1, \dots, \psi_k, \phi \in \lDEL$. 
     In particular then, there are $u_1,\dots,u_k \in W$ such that $w\sim_{a_1} u_1 \sim_{a_2} \cdots \sim_{a_k} u_k$ and $\M,u_k \not \satisfies [\alpha'](\bigwedge \Gamma' \imp \bigvee \Delta' \lor K_a A')$. Hence, $\M^{\mathrm{M}(\alpha')}, (u_k,\alpha') \not\satisfies K_a A'$ so there is $(v,\gamma)\in \M^{\mathrm{M}(\alpha')}$ such that $(u_k,\alpha')\sim_a (v,\gamma)$, $\M,v \satisfies \mathrm{pre}(\gamma)$ and $\M^{\mathrm{M}(\alpha')},(v,\gamma) \not \satisfies A'$.

     However, since the premise is valid $\M,w \satisfies f_\store(I_P)$ for all $\store$-assignments. Let $f'_\store$ be defined as in $(1)$.
     We now need to further distinguish two cases, depending on whether $(i) \ a_k \neq a$, so $na \notin \eta_i \nearrow \eta\hat\;na$ or $(ii) \ a_k = a$ so $\eta_i \nearrow \eta\hat\;na = \eta_i \nearrow na$, namely if $\eta_i \nearrow \eta\hat\;na$ is given by $\iota_1,\dots,\iota_{k+1}$ and $a_1,\dots,a_k$, we have $\iota_k \inter  \eta\hat\;na = \iota_k \inter \lbrace na\rbrace =\lbrace na\rbrace$.
     
     In case $(i)$ we get from $\M,w \satisfies f_\store(I_P)$:
    \begin{align*}
        \M,w \satisfies X_i'^\tau \lor K_{a_1}((\cdots K_{a_k}(X'^\tau \lor [\alpha'](\bigwedge \Gamma' \imp \bigvee \Delta') \lor K_a [\gamma]A') \lor \psi_k) \cdots) \lor \psi_1) \lor \phi
    \end{align*}
    so necessarily $\M,u_k \satisfies K_a[\gamma]A'$. Hence $\M,v \satisfies [\gamma]A'$, which contradicts $\M^{\mathrm{M}(\alpha')},(v,\gamma)\not\satisfies A'$.
    
     In case $(ii)$ we rather have:
    \begin{align*}
        \M,w \satisfies X_i'^\tau \lor K_{a_1}((\cdots K_{a_{k-1}}(K_{a}(X'^\tau \lor [\alpha'](\bigwedge \Gamma' \imp \bigvee \Delta') \lor \psi_{k}) \lor K_a [\gamma]A' \lor \psi_{k-1}) \cdots) \lor \psi_1) \lor \phi
    \end{align*}
    so necessarily $\M,u_{k-1}\satisfies K_a[\gamma]A'$. Since $u_{k-1}\sim_{a}u_k$ and $u_k \sim_a v$, $u_{k-1}\sim_a v$ so $\M,v\satisfies[\gamma]A'$, which again contradicts $\M^{\mathrm{M}(\alpha')},(v,\gamma)\not\satisfies A'$.

    In all cases, we reach a contradiction. Therefore, the conclusion must also be valid.
\end{proof}


\medskip

We now move to the completeness of the calculus. To obtain this result, we first need to show a  derivability lemma, see Lemma \ref{lemmaComposition}.\footnote{Note that this procedure is also adopted in other proof-theoretic works on DEL, e.g. \cite{Nomuraandothers2}.} 

Note that equivalent events have the same complexity: for all $\evt,\mathrm{f} \in \mathrm{Evt}$, if $\evt\sim_a \mathrm{f}\in \store$ then $c_\store(\evt)=c_\store(\mathrm{f})$. More generally, if $\evt\approx \mathrm{f}$, then $c_\store(\evt)=c_\store(\mathrm{f})$.\footnote{Indeed, $\evt \approx \mathrm{f}$ implies both events belong to a same model $\mathrm{M}$ so $c_\store(\evt)=c_\store(\mathrm{f})=c_\store(\mathrm{M})$.} Moreover, for all $\mathrm{e,f}\in \mathrm{Evt}$ if $\evt\approx_a \mathrm{f} \in \store$ then $c_\store(A[\mathrm{f/x}]) < c(\bigwedge_{\evt\sim_a \mathrm{x}}A)$. We can now show the required derivability lemma, which is crucial to obtain the completness of the calculus.

\begin{lemma}\label{lemmaComposition}
    Let $A \in \lDEL$ be a formula, $\evt,\mathrm{f}\in \Evt$ be two events and $\alpha,\beta$ be finite events sequences. Then, the following indexed dynamic hypersequents with event store are derivable:
    \begin{align*}
        &(1) \quad G \Sep \eta: X \update_{\alpha\cdot \evt \cdot \mathrm{f} \cdot \beta}\ A, \Gamma \Imp \Delta \update_{\alpha \cdot (\evt;\mathrm{f}) \cdot \beta} \ \Gamma' \Imp \Delta',A \st \store \\
        &(2) \quad G \Sep \eta: X \update_{\alpha\cdot \evt \cdot \mathrm{f} \cdot \beta} \ \Gamma \Imp \Delta, A \update_{\alpha \cdot (\evt;\mathrm{f}) \cdot \beta} \ A, \Gamma' \Imp \Delta' \st \store
    \end{align*}

\end{lemma}

\begin{proof}
 We simultaneously prove that $(1)$ and $(2)$ are derivable by induction on the sum of the DHS-complexity of each occurrence of $A$, denoted $\overline{c}(A):= c(A, \alpha\cdot \evt \cdot \mathrm{f} \cdot \beta) + c(A, \alpha \cdot (\evt;\mathrm{f}) \cdot \beta)$.

    Without loss of generality, we prove the result for empty $G$, $X$, $\eta$ and $\alpha$. In the following, $\beta = \mathrm{f_1} \cdot \mathrm{f_2} \cdots \mathrm{f_{k-1}}$. We distinguish the following cases, where we only present the details for $(1)$, whilst $(2)$ is treated similarly. In the following, we denote by $\Rule^\ast$ any finitely many applications of rule $\Rule$.

    \noindent$\bullet$ Case $A = p$. To derive $(1)$, we (bottom-up) apply $2\times k$ times (New) to go from $\update_{(\evt;\mathrm{b})\cdot \beta} \ \Gamma \Imp \Delta$ to $\update_\epsilon \ \Imp$ and from $\update_{\evt\cdot \mathrm{b} \cdot \beta} \ \Gamma' \Imp \Delta'$ to $\update_{\evt} \ \Imp $, then we apply $k+1$ times (Rat) and $k$ times (Lat) to obtain an axiom with $p\Imp p$.

    \begin{center}
    \begin{prooftree}[regular]
        \hypo{p \Imp p \update_{\evt} \ \Imp p \update_{\evt\cdot \mathrm{f}} \ \Imp p \update \cdots \update_{\evt\cdot \mathrm{f} \cdot \beta} \Gamma \Imp \Delta, p \update_{(\evt;\mathrm{f})} \ p\Imp  \update \cdots\update_{(\evt;\mathrm{f}) \cdot \beta} \ p, \Gamma' \Imp \Delta' \st \store}
        \infer1[(Lat)$^\ast$]{\Imp p \update_{\evt} \ \Imp p \update_{\evt\cdot \mathrm{f}} \ \Imp p \update \cdots \update_{\evt\cdot \mathrm{f} \cdot \beta} \Gamma \Imp \Delta, p  \update_{(\evt;\mathrm{f})} \ \Imp \update \cdots \update_{(\evt;\mathrm{f}) \cdot \beta} \ p, \Gamma' \Imp \Delta' \st \store}
        \infer1[(Rat)$^\ast$]{\Imp \update_{\evt} \  \Imp \update_{\evt\cdot \mathrm{f}} \ \Imp \update \cdots \update_{\evt\cdot \mathrm{f} \cdot \beta} \Gamma \Imp \Delta, p  \update_{(\evt;\mathrm{f})} \ \Imp \update \cdots \update_{(\evt;\mathrm{f}) \cdot \beta} \ p, \Gamma' \Imp \Delta' \st \store}
        \infer1[(New)$^\ast$]{\update_{\evt\cdot \mathrm{f} \cdot \beta} \ \Gamma \Imp \Delta, p \update_{(\evt;\mathrm{f})\cdot \beta}  \ p, \Gamma' \Imp \Delta' \st \store}
    \end{prooftree}
    \end{center}
    \hfill \\

    \noindent$\bullet$ Case $A =\lnot A'$. We apply the inductive hypothesis, then (L$\lnot$) and (R$\lnot$). This is straightforward.

    \noindent$\bullet$ Case $A = A' \land A''$. We apply the inductive hypothesis, then we apply admissibility of (IW) and use (L$\land$) and (R$\land$). This is straightforward. 

    \noindent$\bullet$ Case $A = \bigwedge_{\mathrm{g}\sim_a\mathrm{x}}A'(\mathrm{x})$. We start from $\update_{\mathrm{e\cdot f}\cdot \beta}  A[\mathrm{y/x}],\Gamma \Imp \Delta \update_{(\mathrm{e;f})\cdot\beta} \Gamma' \Imp \Delta', A[\mathrm{y/x}] \st \store, \mathrm{g}\sim_a \mathrm{y}$ that is derivable by inductive hypothesis. Then we apply admissibility of (IW) and then use (L$\bigwedge \evt$) and (R$\bigwedge\evt$).
    
    \begin{center}
    \begin{prooftree}[regular]
        \hypo{\update_{\mathrm{e\cdot f}\cdot \beta}  A[\mathrm{y/x}],\Gamma \Imp \Delta \update_{(\mathrm{e;f})\cdot\beta} \Gamma' \Imp \Delta', A[\mathrm{y/x}] \st \store, \mathrm{g}\sim_a \mathrm{y}}
        \infer1[(IW)]{\update_{\mathrm{e\cdot f}\cdot \beta} \bigwedge_{\mathrm{g}\sim_a \mathrm{x}}A, A[\mathrm{y/x}], \Gamma \Imp \Delta \update_{(\mathrm{e;f})\cdot\beta} \Gamma' \Imp \Delta', A[\mathrm{y/x}] \st \store, \mathrm{g}\sim_a \mathrm{y}}
        \infer1[(L$\bigwedge \evt$)]{\update_{\mathrm{e\cdot f}\cdot \beta} \bigwedge_{\mathrm{g}\sim_a \mathrm{x}}A, \Gamma \Imp \Delta \update_{(\mathrm{e;f})\cdot\beta} \Gamma' \Imp \Delta', A[\mathrm{y/x}] \st \store, \mathrm{g}\sim_a \mathrm{y} }
        \infer1[(R$\bigwedge \evt$)]{\update_{\mathrm{e\cdot f}\cdot \beta} \bigwedge_{\mathrm{g}\sim_a \mathrm{x}}A, \Gamma \Imp \Delta \update_{(\mathrm{e;f})\cdot\beta} \Gamma' \Imp \Delta',  \bigwedge_{\mathrm{g}\sim_a \mathrm{x}}A  \st \store}
    \end{prooftree}
    \end{center}

    \noindent$\bullet$ Case $A = K_a A'$. Note that $\overline{c}(A') < \overline{c}(K_a A')$ and, for all $\mathrm{e'}\cdot \mathrm{f'}\cdot\beta' \sim_a \evt\cdot \mathrm{f}\cdot\beta$ and all $\beta'_i\cdot \mathrm{f'_i} \prefix \beta'$, $c(\mathrm{pre(f'_i}), \mathrm{e'} \cdot \mathrm{f'} \cdot \beta'_i) + c(\mathrm{pre(f'_i}), (\mathrm{e'} ; \mathrm{f'}) \cdot \beta'_i) < \overline{c}(K_a A')$. This allows us to use the inductive hypothesis to obtain the following indexed dynamic hypersequents -- where $\store'=\store, \evt\cdot\mathrm{f}\cdot\beta \sim_a \mathrm{e'}\cdot\mathrm{f'}\cdot\beta'$:
    \begin{small}
    \begin{align*}
        (1) \quad &\update_{\mathrm{e'} \cdot \mathrm{f'}} \Imp \mathrm{pre(f'_1}) \update_{\mathrm{(e';f')}} \mathrm{pre(f'_1}) \Imp \st \store' \\
        (2) \quad &\update_{\mathrm{e'} \cdot \mathrm{f'} \cdot \mathrm{f'_1}} \Imp \mathrm{pre(f'_2}) \update_{\mathrm{(e';f')} \cdot \mathrm{f'_1}} \mathrm{pre(f'_2}) \Imp \st \store' \\
        \vdots \ \quad & \\
        (k-1) \quad &\update_{\mathrm{e'} \cdot \mathrm{f'} \cdots \mathrm{f'_{k-2}}} \Imp \mathrm{pre(f'_{k-1}}) \update_{\mathrm{(e';f')} \cdots \mathrm{f'_{k-2}}} \mathrm{pre(f'_{k-1}}) \Imp \st \store' \\
        (k) \quad &\update_{\mathrm{e'} \cdot \mathrm{f'} \cdot \beta'} A' \Imp \update_{\mathrm{(e';f')} \cdot\beta'} \Imp A' \st \store'
    \end{align*}
    \end{small}

    \noindent By admissibility of (EW) and (DW), from each $(i)$, $1\leq i\leq k$, we get the following $(i')$, where $X$ denotes the dynamic sequent $\update_{\evt \cdot \mathrm{f} \cdot \beta} \  K_a A', \Gamma \Imp \Delta \update_{\mathrm{(e;f)} \cdot \beta} \ \Gamma' \Imp \Delta'$:
    \begin{small}
    \begin{align*}
        (1') \quad &na: X \Sep na: \mathrm{pre}(\evt') \Imp \update_{\evt'} \mathrm{pre(f')} \Imp  \update_{\evt' \cdot \mathrm{f'}} \ \Imp \mathrm{pre(f'_1}) \ \update_{(\mathrm{e';f'}) \cdot \mathrm{f'_1}} \ \mathrm{pre(f'_2}) \Imp \ \update_{(\mathrm{e';f'})} \mathrm{pre(f'_1}) \Imp \update \cdots \update_{(\mathrm{e';f'})\cdot\beta'} \Imp A' \st \store' \\
        (2') \quad &na: X \Sep na: \mathrm{pre(e'}) \Imp \update_{\mathrm{e'}} \mathrm{pre(f'}) \Imp  \update_{\mathrm{e'} \cdot \mathrm{f'} \cdot \mathrm{f'_1}} \Imp \mathrm{pre(f'_2}) \update_{(\mathrm{e';f'})} \ \mathrm{pre(f'_1}) \Imp \ \update_{(\mathrm{e';f'}) \cdot \mathrm{f'_1}} \ \mathrm{pre(f'_2}) \Imp \ \update \cdots \update_{(\mathrm{e';f'})\cdot\beta'} \Imp A' \st \store \\
        \vdots \ \quad & \\
        (k') \quad &na: X \Sep na: \mathrm{pre(e'}) \Imp \update_{\mathrm{e'}} \ \mathrm{pre(f'}) \Imp \update_{\mathrm{e'} \cdot \mathrm{f'} \cdot \beta'} A' \Imp \update_{(\mathrm{e';f'}) } \ \mathrm{pre(f}_1') \Imp \ \update_{(\mathrm{e';f'}) \cdot \mathrm{f'_1}} \ \mathrm{pre(f'_2}) \Imp \ \update \cdots \update_{(\mathrm{e';f'}) \cdot\beta'} \Imp A' \st \store'
    \end{align*}
    \end{small}

    \noindent To apply ($\evt$LK$_2)'$ with $\update_{\evt\cdot \mathrm{f}\cdot \beta} \ K_a A', \Gamma' \Imp \Delta'$, we need $k+1$ premises, since the sequence of events here is $\mathrm{f} \cdot \mathrm{f_1} \cdots \mathrm{f_{k-1}}$. Hence, we also need the following IDHS, that is derivable by Lemma \ref{lemmaDerivabilityTrivialSequents}:
    \begin{small}
    \begin{displaymath}
        (0)' \quad na: X \Sep na: \mathrm{pre(e'}) \Imp \update_{\mathrm{e'}} \ \mathrm{pre(f'}) \Imp \mathrm{pre(f'}) \update_{(\mathrm{e';f'})} \ \mathrm{pre(f'_1}) \Imp \update \cdots \update_{(\mathrm{e';f'}) \cdot \beta'} \ \Imp A' \st \store'
    \end{displaymath}
    \end{small}
    We now apply ($\evt$LK$_2)'$)  to $(0)'-(k)'$ and continue the derivation as follows:
    \begin{center}
    \begin{prooftree}[small]
            \hypo{na: X \SepC na: \mathrm{pre(e'}) \Imp \mathrm{pre(e'})  \cdots}
                \hypo{(0)'}
                \hypo{(1)'}
                \hypo{\cdots}
                \hypo{(k)'}
            \infer[separation=4em]4[($\evt$LK$_2)'$]{ na: X \SepC na: \mathrm{pre(e'})\Imp \update_{\mathrm{e'}} \ \mathrm{pre(f'}) \Imp \update \cdots \update_{\mathrm{(e';f')} \cdot \beta'} \ \Imp A' \st \store'}
        \infer[separation=1.4em]2[(L$[\cdot]$)]{na: X \Sep na: \mathrm{pre(e'}),[\mathrm{e'}]\mathrm{pre(f'}) \Imp \update_{(\mathrm{e';f'})} \ \mathrm{pre(f_1'}) \Imp \update \cdots \update_{(\mathrm{e';f'}) \cdot \beta'} \Imp A' \st \store'}
        \infer1[(L$\land$)]{na: X \Sep na: \mathrm{pre(e'}) \land [\mathrm{e'}]\mathrm{pre(f'}) \Imp \update_{(\mathrm{e';f'})} \ \mathrm{pre(f_1'}) \Imp \update \cdots 
        \update_{(\mathrm{e';f'}) \cdot \beta'} \Imp A' \st \store'}
        \infer1[(Recall)$^\ast$]{na: X \Sep na: \ \Imp \update_{(\mathrm{e';f'})} \Imp \update \cdots \update_{(\mathrm{e';f'}) \cdot \beta'} \Imp A' \st \store'}
        \infer1[(New)$^\ast$]{na: \update_{\evt\cdot \mathrm{f}\cdot \beta} \ K_a A', \Gamma \Imp \Delta \update_{(\mathrm{e;f}) \cdot \beta}\ \Gamma' \Imp \Delta'  \Sep  na: \update_{(\mathrm{e';f'}) \cdot \beta'} \Imp A' \st \store' }
        \infer1[($\evt$RK)]{\update_{\evt\cdot \mathrm{f}\cdot \beta} \ K_a A', \Gamma \Imp \Delta \update_{(\mathrm{e;f}) \cdot \beta}\ \Gamma' \Imp \Delta', K_a A' \st \store}
    \end{prooftree}
    \end{center}

\noindent$\bullet$ Case $A=[\mathrm{g}]A'$. We use the inductive hypothesis and admissibility of (DW), then apply (L$[\cdot]$) and (R$[\cdot]$).This is straightforward.

\end{proof}

\begin{theorem}[Completeness]\label{theoremCompleteness}
     For all formulas $A \in \lDEL$, if $\ \satisfies_{DEL} A$ then $\vdash_{\calculus} \Imp A$.
\end{theorem}

\begin{proof}
    Since \textbf{DEL} is complete w.r.t. the semantics of DEL (see Theorem \ref{AxComp}), it is sufficient to demonstrate that all axioms of \textbf{DEL} are derivable in $\calculus$ and that the rules of \textit{modus ponens} and necessitation are admissible. \textit{Modus ponens} can be shown to be admissible using the Cut-rule (shown to be admissible in $\calculus$, see Section \ref{sectionCut}) whilst the necessitation rule can be shown to be admissible using external weakening. Derivations for the axioms of propositional logic and modal logic \textbf{S5m} are obtained by direct (bottom-up) proof-search, see  \cite{Poggiolesi2008} for details. Derivations for the reduction axioms are also obtained by (bottom-up) proof search (except for the axiom of event composition that we will consider separately). As an illustrative example, consider the axiom of announcement and knowledge, $[\evt^{\mathrm{M}}]K_a A \imp \mathrm{pre(\evt^{M})} \eq \bigwedge_{\evt^{\mathrm{M}} \sim_a \mathrm{x}}K_a[\mathrm{x}]A)$ for any given pointed event model $\evt^{\mathrm{M}}$, which we prove to be a theorem by means of the following derivation:

\vspace{0.3cm}

\begin{center}
    \begin{prooftree}[small]
    \hypo{ \mathrm{pre(\evt^M)}\Imp \mathrm{pre(\evt^M)}, \bigwedge_{\mathrm{x}\sim_a \evt^{\mathrm{M}}}K_a[\mathrm{x}]A }
        \hypo{1a: \mathrm{pre(\evt^M)} \Imp \update_{\evt^{\mathrm{M}}} K_a A \Imp \Sep 1a: \mathrm{pre(y)} \ \Imp \update_{\mathrm{y}} A \Imp A \st \evt^{\mathrm{M}}\sim_a \mathrm{y}}
        \infer1[($\evt$LK$_2$)]{1a: \mathrm{pre(\evt^M)} \Imp \update_{\evt^{\mathrm{M}}} K_a A \Imp \Sep 1a: \mathrm{pre(y}) \Imp \ \update_{\mathrm{y}} \ \Imp A \ \st \evt^{\mathrm{M}}\sim_a \mathrm{y}}
        \infer1[(R$[\cdot]$)]{1a: \mathrm{pre(\evt^M)} \Imp \update_{\evt^{\mathrm{M}}} K_a A \Imp  \Sep 1a: \  \Imp [\mathrm{y}]A \ \st \evt^{\mathrm{M}}\sim_a \mathrm{y} }
        \infer1[($\evt$RK)]{ \mathrm{pre(\evt^M)} \Imp K_a[\mathrm{y}]A \update_{\evt^{\mathrm{M}}} K_a A \Imp \ \st \evt^{\mathrm{M}}\sim_a \mathrm{y}}
        \infer1[(R$\bigwedge\evt$)]{ \mathrm{pre(\evt^M)} \Imp \bigwedge_{\evt^{\mathrm{M}} \sim_a \mathrm{x}}K_a[\mathrm{x}]A \update_{\evt^{\mathrm{M}}} K_a A \Imp}
    \infer[separation=0.8em]2[(L$[\cdot]$)]{\mathrm{pre(\evt^M)}, [\evt^{\mathrm{M}}]K_a A \Imp \bigwedge_{\evt^{\mathrm{M}} \sim_a \mathrm{x}}K_a[\mathrm{x}]A}
    \infer1[(R$\imp$)]{ [\evt^{\mathrm{M}}]K_a A \Imp \mathrm{pre(\evt^M)} \imp \bigwedge_{\evt^{\mathrm{M}} \sim_a \mathrm{x}}K_a[\mathrm{x}]A}
    \infer1[(R$\imp$)]{\Imp [\evt^{\mathrm{M}}]K_a A \imp (\mathrm{pre}(\evt^{\mathrm{M}}) \imp \bigwedge_{\evt^{\mathrm{M}} \sim_a \mathrm{x}} K_a[\mathrm{x}]A)}
    \end{prooftree} 
\end{center}

\vspace{0.3cm}

and the derivation below, where $X: \mathrm{pre}(\evt), \bigwedge_{\evt^{\mathrm{M}} \sim_a \mathrm{x}}  K_a[\mathrm{x}]A, K_a[\mathrm{y}]A \Imp \update_{\mathrm{pre(\evt^M)}} \Imp$
\vspace{0.3cm}

\begin{center}
\begin{prooftree}[small]
    \hypo{\mathrm{pre(\evt^M)} \Imp \mathrm{pre(\evt^M)} \cdots}
        \hypo{1a: X \Sep 1a: \mathrm{pre(y)} \Imp \mathrm{pre(y)} \cdots }
            \hypo{1a: X \Sep 1a: \mathrm{pre(y)} \Imp \update_{\mathrm{y}} \ A \Imp A \  \st \mathrm{\evt^M}\sim_a \mathrm{y}}
        \infer[separation=3em]2[(L$[\cdot]$)]{1a: \mathrm{pre}(\evt), \bigwedge_{\evt^{\mathrm{M}} \sim_a \mathrm{x}}  K_a[\mathrm{x}]A, K_a[\mathrm{y}]A \Imp \update_{\evt^{\mathrm{M}}} \Imp \Sep 1a: \mathrm{pre(y)}, [\mathrm{y}]A \Imp \update_{\mathrm{y}} \Imp A \ \st \evt^{\mathrm{M}} \sim_a \mathrm{y}}
        \infer1[(Recall)]{1a: \mathrm{pre(\evt^M)}, \bigwedge_{\evt^{\mathrm{M}} \sim_a \mathrm{x}}  K_a[\mathrm{x}]A, K_a[\mathrm{y}]A \Imp \update_{\evt^{\mathrm{M}}} \Imp \Sep 1a: \ [\mathrm{y}]A \Imp \ \update_{\mathrm{y}} \Imp A \  \st \evt^{\mathrm{M}} \sim_a \mathrm{y}}
        \infer1[($\evt$LK$_2$)]{1a: \mathrm{pre(\evt^M)}, \bigwedge_{\evt^{\mathrm{M}} \sim_a \mathrm{x}}  K_a[\mathrm{x}]A, K_a[\mathrm{y}]A \Imp \ \update_{\evt^{\mathrm{M}}}\Imp \Sep 1a: \ \Imp \ \update_{\mathrm{y}} \Imp A \ \st \mathrm{\evt^M} \sim_a \mathrm{y}}
        \infer1[(New)]{1a: \mathrm{pre(\evt^M)}, \bigwedge_{\evt^{\mathrm{M}} \sim_a \mathrm{x}}  K_a[\mathrm{x}]A, K_a[\mathrm{y}]A \Imp \update_{\evt^{\mathrm{M}}}\Imp \Sep 1a: \update_{\mathrm{y}} \Imp A \  \st \evt^{\mathrm{M}} \sim_a \mathrm{y}}
        \infer1[(L$\bigwedge \evt$)]{1a: \mathrm{pre(\evt^M)}, \bigwedge_{\evt^{\mathrm{M}} \sim_a \mathrm{x}}  K_a[\mathrm{x}]A \Imp \update_{\mathrm{\evt^M}} \Imp \Sep 1a: \update_{\mathrm{y}} \ \Imp A \  \st \evt^{\mathrm{M}} \sim_a \mathrm{y}}
        \infer1[($\evt$RK)]{\mathrm{pre(\evt^M)}, \bigwedge_{\evt^{\mathrm{M}} \sim_a \mathrm{x}} K_a[\mathrm{x}]A \Imp \update_{\evt^{\mathrm{M}}} \ \Imp K_a A }
    \infer[separation=-1em]2[(L$\imp$)]{\mathrm{pre(\evt^M)}, \mathrm{pre(\evt^M)} \imp \bigwedge_{\evt^{\mathrm{M}} \sim_a \mathrm{x}} K_a[\mathrm{x}]A \Imp \update_{\evt^{\mathrm{M}}}\Imp K_a A }
    \infer1[(R$[\cdot]$)]{ \mathrm{pre(\evt^M)} \imp \bigwedge_{\evt^{\mathrm{M}} \sim_a \mathrm{x}} K_a[\mathrm{x}]A \Imp [\evt^{\mathrm{M}}]K_a A }
    \infer1[(R$\imp$)]{\Imp (\mathrm{pre(\evt^M)} \imp \bigwedge_{\evt^{\mathrm{M}} \sim_a \mathrm{x}} K_a[\mathrm{x}]A) \imp [\evt^{\mathrm{M}}]K_a A }
\end{prooftree}
\end{center}
\vspace{0.3cm}

 We now move to the axiom of event composition $[\mathrm{e^M}][\mathrm{e^N}]A \eq [(\mathrm{e^M;f^N})]A$. Providing a direct proof for that axiom schema requires Lemma \ref{lemmaComposition} that demonstrates derivability of dynamic hypersequents $\update_{\mathrm{e^M} \cdot \mathrm{f^N}}  \Imp A \update_{(\mathrm{e^M;f^N})} A \Imp$ and $\update_{\mathrm{e^M} \cdot \mathrm{f^N}} \ A \Imp \update_{(\mathrm{e^M} ; \mathrm{f^N})} \Imp A$, for all formulas $A$. We only show the derivation for one direction, whilst the other can be obtained in a similar way.\\
    
\vspace{0.3cm}

    \noindent
    {\footnotesize
    \begin{prooftree}[regular]
                \hypo{\mathrm{pre(e^M}) \Imp \mathrm{pre(e^M}) \cdots}
                \hypo{\mathrm{pre(e^M}), \mathrm{pre(e^M}) \Imp \update_{\mathrm{e^M}} \ \mathrm{pre(f^N}) \Imp \mathrm{pre(f^N}) \update_{\mathrm{e^M} \cdot \mathrm{f^N}} \ \Imp A}
                \infer1[(R$[\cdot]$)]{\mathrm{pre(e^M}) \Imp [\mathrm{e^M}]\mathrm{pre(f^N}) \update_{\mathrm{e^M}} \mathrm{pre(f^N}) \ \Imp \ \update_{\mathrm{e^M} \cdot \mathrm{f^N}} \ \Imp A}
            \infer2[(R$\land$)]{\mathrm{pre(e^M}) \Imp \mathrm{pre(e^M}) \land [\mathrm{e^M}]\mathrm{pre(f^N}) \update_{\mathrm{e^M}} \ \mathrm{pre(f^N}) \Imp \update_{\mathrm{e^M} \cdot \mathrm{f^N}} \ \Imp A}
            \hypo{}
            \ellipsis{Lemma \ref{lemmaComposition}}{\cdots \update_{\mathrm{e^M} \cdot \mathrm{f^N}} \ \Imp A \update_{(\mathrm{e^M};\mathrm{f^N})} \ A \Imp}
        \infer2[(L$[\cdot]$)]{\mathrm{pre(e^M}), [(\mathrm{e^M;f^N})]A \Imp \ \update_{\mathrm{e^M}} \ \mathrm{pre(f^N}) \Imp \update_{\mathrm{e^M} \cdot \mathrm{f^N}} \Imp A}
        \infer1[(R$[\cdot]$)]{\mathrm{pre(e^M}), [(\mathrm{e^M;f^N})]A \Imp \ \update_{\mathrm{e^M}} \ \Imp [\mathrm{f^N}]A}
        \infer1[(R$[\cdot]$)]{[(\mathrm{e^M;f^N})]A \Imp [\mathrm{e^M}][\mathrm{f^N}]A}
        \infer1[(R$\imp$)]{\Imp [(\mathrm{e^M;f^N})]A \imp [\mathrm{e^M}][\mathrm{f^N}]A}
    \end{prooftree}
}

\end{proof}


\section{Cut admissibility}\label{sectionCut}

In this section we demonstrate that the Cut-rule is admissible in the calculus $\calculus$. To do so, we first need to prove a preliminary lemma.

\begin{lemma}\label{lemmaPermutationRules}
    Propositional rules, modal rules, event rules, as well as (New), (Recall) and (Lat) (resp. (Rat)) permute down with respect to $k$ applications of (Rat) (resp. (Lat)). Furthermore, any occurrence of (Rat) (resp. (Lat)) also permutes down with respect to $k$ applications of (Rat) (resp. (Lat)) when the atom $p$ that is removed in its conclusion is not active in either of these $k$ applications of (Rat) (resp. (Lat)).
\end{lemma}

\begin{proof}
    This is straightforward.
\end{proof}

\begin{theorem}\label{theoremCut}
    The rule of Cut is admissible.
    \begin{center}
    \begin{prooftree}[regular]
        \hypo{G \Sep \eta : X \update_\alpha \ \Gamma  \Imp \Delta, A \st \store}
        \hypo{G \Sep \eta : X \update_\alpha \ A, \Gamma \Imp \Delta \st \store}
        \infer[separation=3em]2[(Cut)]{G \Sep \eta : X \update_\alpha \ \Gamma \Imp \Delta \st \store}
    \end{prooftree}
    \end{center}
\end{theorem}

\begin{proof}
    The proof proceeds by main induction on the DHS-complexity of the cut-formula $c_\store(A,\alpha)$ (see Definition \ref{defComplexity}), and secondary induction on the sum of the heights of the derivations $d_1$ and $d_2$ of each premise. We distinguish cases on the last rule applied on the left-hand premise.

(i) If $G \ \vert \ X \update_\alpha \Gamma \Imp \Delta, A \st \store$ is an axiom, then either $G \ \vert \ X \update_\alpha A, \Gamma \Imp \Delta \st \store$ is also an axiom and so is the conclusion; or the conclusion can  be obtained by hp-admissibility of contraction.

(ii) If $G \ \vert \ X \update_\alpha \Gamma \Imp \Delta, A \st \store$ has been obtained by a rule $\Rule$ where the cut-formula $A$ is not principal, then we apply $\Rule^\star$ on $G \ \vert \ X \update_\alpha A, \Gamma \Imp \Delta \st \store$ obtaining an indexed dynamic hypersequent with event store $I\st\store^{\prime}$. We use $I\st\store^{\prime}$ as well as the premise of the rule $\Rule$ to cut the formula $A$, where this cut is admissible by  induction on the sum of the heights. We then apply the rule $\Rule$ again to obtain the desired conclusion.

(iii) If $G \ \vert \ X \update_\alpha \Gamma \Imp \Delta, A \st \store$ has been obtained by a rule $\Rule$ where $A$ is principal, and the right-hand premise has also been obtained by a rule where $A$ is principal, then we distinguish the following subcases: $(1)$ $\Rule$ is a propositional rule, $(2)$ $\Rule$ is a propositional event rule, $(3)$ $\Rule$ is a modal rule, $(4)$ $\Rule$ is an event rule, and $(5)$ $\Rule$ is a dynamic rule.
We treat in details cases (3) and (4). Cases (1) and (2) can be treated analogously. As for (5), the case where $\mathcal{R}$ is the rule (Rat) is detailed below -- the other cases of dynamic rules can be treated straightforwardly.

    $(3)$ Consider the case where $\mathcal{R}$ is the rule ($\evt$RK). Then the rule that has introduced the right-hand premise is either ($\evt$LK$_{1})$ or ($\evt$LK$_{2})$, or ($\evt$LK$_1)'$ or ($\evt$LK$_2)'$. We first analyse the case of the rule ($\evt$LK$_{2})$ and then treat the case of ($\evt$LK$_2)'$ -- the cases of rules ($\evt$LK$_{1})$ and ($\evt$LK$_{1})'$ can be dealt in a similar way, respectively. Suppose that we have the following situation, where $\eta\cap \theta = \lbrace na \rbrace$, $\alpha \approx_a \beta \in \store$ and $\overline{Y}:= \theta: Y \update_\beta \Gamma' \Imp \Delta'$:

\begin{center}
\begin{prooftree}[tiny]
	   \hypo{G \SepC \eta : X \update_\alpha \ \Gamma \Imp \Delta \Sep \overline{Y} \Sep na: \update_\chi \ \Imp A \st \store, \alpha \sim_a \chi}
	   \infer1[($\evt$RK)]{G \SepC \eta : X \update_\alpha \ \Gamma \Imp \Delta, K_a A \Sep \theta: Y \update_\beta \ \Gamma' \Imp \Delta'  \st \store}
 	      \hypo{G \SepC \eta : X \update_\alpha \ K_a A, \Gamma \Imp \Delta \Sep \theta: Y \update_\beta \ A, \Gamma' \Imp \Delta'  \st \store}
 	      \infer1[($\evt$LK$_2$)]{ G \SepC \eta : X \update_\alpha \ K_a A, \Gamma \Imp \Delta \Sep \theta: Y \update_\beta \ \Gamma' \Imp \Delta'  \st \store}
 	      \infer[separation=1.4em]2[(Cut)]{G \SepC \eta : X \update_\alpha \ \Gamma \Imp \Delta \Sep \theta: Y \update_\beta \ \Gamma' \Imp \Delta'  \st \store}
\end{prooftree}
\end{center}
\noindent
We first use the hp-admissibility of (IW) on the left-hand premise of the cut to obtain the indexed dynamic hypersequent with event store $G \ \vert \ \eta : X \update_\alpha \Gamma \Imp \Delta, K_a A \ \vert \ \theta: Y \update_\beta A, \Gamma' \Imp \Delta'  \st \store$. We use this IDHS together with the premise of the rule ($\evt$LK$_{2})$ to cut the formula $K_a A$, thereby obtaining $(\star):G \ \vert \ \eta : X \update_\alpha  \Gamma \Imp \Delta \ \vert \ \theta: Y \update_\beta A, \Gamma' \Imp \Delta'  \st \store$. This cut is admissible by induction on the sum of the heights.
We then proceed in the following way:\footnote{Note that in the derivation below, the occurrence of (Merge) uses the fact that $\beta\approx_a \chi \in \store$ because $\alpha \approx_a \beta\in \store$ and $\alpha\sim_a \chi \in \store$. {Moreover, since $\alpha \approx_a \beta \in \store$ the expression $\alpha \sim_a \beta$ obtained when replacing $\chi$ by $\beta$ is deleted when applying (Merge).}}
\begin{center}
\begin{prooftree}[tiny]
	   \hypo{G \SepC \eta : X \update_\alpha \ \Gamma \Imp \Delta \Sep \theta: Y \update_\beta \ \Gamma' \Imp \Delta' \Sep na: \update_\chi \ \Imp A \st \store, \alpha \sim_a \chi}
	   \infer1[(Merge)]{G \SepC \eta : X \update_\alpha \ \Gamma \Imp \Delta \Sep \theta: Y \update_\beta \ \Gamma' \Imp \Delta', A \st \store}
 	      \hypo{(\star)}
          \infer[no rule]1{G \SepC \eta : X \update_\alpha \ \Gamma \Imp \Delta \Sep \theta: Y \update_\beta \ A, \Gamma' \Imp \Delta'  \st \store}
 	      \infer[separation=1.4em]2[(Cut)]{G \SepC \eta : X \update_\alpha \ \Gamma \Imp \Delta \Sep \theta: Y \update_\beta  \ \Gamma' \Imp \Delta'  \st \store}
\end{prooftree}
\end{center}
\noindent where this cut is admissible by induction on the DHS-complexity of the formula.

Let us now consider the case where $\Rule$ is an application of ($\evt$RK) and the right-hand premise has been obtained by an application of ($\evt$LK$_2)'$. For the sake of clarity, we only show the case where the rule ($\evt$LK$_2)'$ has three premises. The general case is instead shown in the Appendix. Suppose we have the following situation, where $\store' = \store, \alpha\cdot \gamma \sim_a \chi$.
    
    \begin{center}
        \begin{prooftree}[tiny]
            \hypo{G \SepC \eta: X \update_{\alpha\cdot \gamma} \ \Gamma \Imp \Delta \Sep \theta: Y \update_\beta \ \Gamma' \Imp \Delta \Sep na: \update_\chi \ \Imp A \st \store'}
            \infer1[($\evt$RK)]{G \SepC \eta: X \update_{\alpha\cdot \gamma} \ \Gamma \Imp \Delta,  K_a A \Sep \theta: Y \update_\beta \ \Gamma' \Imp \Delta' \st \store}
            \hypo{G \SepC \eta: X \update_{\alpha\cdot \gamma} \ K_a A, \Gamma \Imp \Delta \Sep \theta: Y \update_\beta \ \Gamma' \Imp \Delta' \st \store}
            \infer2[(Cut)]{G \SepC \eta: X \update_{\alpha\cdot \gamma}  \ \Gamma \Imp \Delta \Sep \theta: Y \update_\beta \ \Gamma' \Imp \Delta' \st \store}
        \end{prooftree}
    \end{center}
    
  \noindent  and the right-hand premise has been obtained by the following application of the rule ($\evt$LK$_2)'$, where $\gamma = \mathrm{c_1}\cdot \mathrm{c_{2}}$, $\delta= \mathrm{d_1\cdot d_2}$, $\alpha\cdot\gamma \approx_a \beta \cdot \delta \in \store$. Moreover, $\eta \cap \theta = \lbrace na \rbrace$ and $\overline{X}= \eta: X \update_{\alpha\cdot\gamma} \ K_a A, \Gamma \Imp \Delta$.

    \begin{center}
    \begin{prooftree}[tiny]
            \hypo{(1)}
            \infer[no rule]1{G\SepC \overline{X} \Sep \theta: Y \update_\beta \ \Gamma' \Imp \Delta',  \mathrm{pre(d_1}) \st \store}
            \hypo{(2)}
            \infer[no rule]1{G\SepC \overline{X} \Sep \theta:Y \update_\beta \ \Gamma' \Imp \Delta' \update_{\beta\cdot \mathrm{d_1}} \ \Imp \mathrm{pre(d_2}) \st \store}
            \hypo{(3)}
            \infer[no rule]1{G\SepC \overline{X} \Sep \theta:Y \update_\beta \ \Gamma' \Imp \Delta' \update_{\beta \cdot \delta} \ A \Imp \st \store}
        \infer[separation=1.2em]3[($\evt$LK$_2)'$]{G \SepC \eta: X \update_{\alpha\cdot \gamma} \ K_a A, \Gamma \Imp \Delta \Sep \theta: Y \update_\beta \Gamma^{\prime} \Imp \Delta^{\prime} \st \store}
  \end{prooftree}
  \end{center}
  
    \medskip\noindent
    We first use the left-hand premise of original cut and hp-admissibility of (DW) to operate the following cut, which is admissible by induction on the sum of the heights:
    \begin{center}
        \begin{prooftree}[tiny]
            \hypo{G \SepC \eta: X \update_{\alpha\cdot \gamma} \ \Gamma \Imp \Delta, K_a A \Sep \theta: Y \update_\beta \ \Gamma' \Imp \Delta' \st \store  }
            \infer1[(DW)]{G \SepC \eta: X \update_{\alpha\cdot \gamma} \ \Gamma \Imp \Delta, K_a A \Sep \theta: Y \update_\beta \ \Gamma' \Imp \Delta' \update_{\beta\cdot\delta} \ A \Imp \st \store  }
            \hypo{(3)}
            \infer[no rule]1{ G \SepC \eta: X \update_{\alpha\cdot \gamma} \ , K_a A, \Gamma \Imp \Delta \Sep \theta: Y \update_\beta \ \Gamma' \Imp \Delta' \update_{\beta \cdot \delta} \ A \Imp  \st \store}
            \infer[separation=1.3em]2[(Cut)]{\ G \SepC \eta: X \update_{\alpha\cdot \gamma} \ \Gamma \Imp \Delta \Sep \theta: Y \update_\beta \ \Gamma' \Imp \Delta' \update_{\beta\cdot\delta} \ A \Imp \st \store}
        \end{prooftree}
    \end{center}

    \medskip\noindent where we denote the conclusion by $(3)'$.  
    We proceed likewise to cut the formula $K_a A$ in premises $(1)$ and $(2)$, thereby obtaining $(1)': G \ \vert \ \eta: X \update_{\alpha\cdot \gamma} \ \Gamma \Imp \Delta \ \vert \ \theta: Y \update_\beta \ \Gamma' \Imp \Delta', \mathrm{pre(d_1})  \st \store $\footnote{To obtain $(1)'$ we need hp-admissibility of (IW) instead of that of (DW).} and $(2)': G \ \vert \ \eta: X \update_{\alpha\cdot \gamma} \ \Gamma \Imp \Delta \ \vert \ \theta: Y \update_\beta \ \Gamma' \Imp \Delta' \update_{\beta \cdot \mathrm{d_1}} \ \Imp \mathrm{pre(d_2}) \st \store$, respectively. Analogously to what we have seen above, both cuts are eliminable by induction on the sum of the heights. We then consider the premise of the rule ($\evt$RK) and proceed as follows:\footnote{In the derivation below, the occurrence of (Merge) uses the fact that $\beta\cdot \delta \approx_a \chi$ since both $\alpha\cdot \gamma \approx_a \beta\cdot \delta$ and $\alpha\cdot\gamma \sim_a \chi$ occur in $\store^\prime$. We can therefore substitute $\beta\cdot \delta$ to $\chi$. Moreover, the expression $\alpha\cdot\gamma \sim_a \beta\cdot \delta$ so obtained can be deleted because $\alpha \cdot \gamma \approx_a \beta \cdot \delta \in \store$.}
    
    \begin{center}
        \begin{prooftree}[tiny]
            \hypo{G \SepC \eta: X \update_{\alpha\cdot \gamma} \ \Gamma \Imp \Delta \Sep na: \update_\chi \ \Imp A \Sep \theta: Y \update_\beta \ \Gamma' \Imp \Delta' \st \store'}
            \infer1[(DW)]{G \SepC \eta: X \update_{\alpha\cdot \gamma} \ \Gamma \Imp \Delta \Sep na: \update_\chi \ \Imp A \Sep \theta: Y \update_\beta \ \Gamma' \Imp \Delta' \, \update_{\beta \cdot \delta} \ \Imp \st \store'}
            \infer1[(Merge)]{G \SepC \eta: X \update_{\alpha \cdot \gamma} \ \Gamma \Imp \Delta \Sep \theta: \update_\beta \ \Gamma' \Imp \Delta' \, \update_{\beta\cdot \delta} \ \Imp A \st \store}
        \end{prooftree}
    \end{center}

    \noindent We use the conclusion of this derivation together with $(3)'$ in the following way:
    \begin{center}
        \begin{prooftree}[tiny]
            \hypo{G \SepC \eta: X \update_{\alpha \cdot \gamma} \ \Gamma \Imp \Delta \Sep \theta: \update_\beta \ \Gamma' \Imp \Delta' \, \update_{\beta\cdot \delta} \ \Imp A \st \store}
            \hypo{(3)'}
            \infer[no rule]1{G \SepC \eta: X \update_{\alpha \cdot \gamma} \ \Gamma \Imp \Delta \Sep \theta: Y \update_\beta \ \Gamma' \Imp \Delta' \, \update_{\beta\cdot \delta} \ A \Imp \st \store}
            \infer[separation=3em]2[(Cut)]{G \SepC \eta: X \update_{\alpha \cdot \gamma} \ \Gamma \Imp \Delta \Sep \theta: Y \update_\beta \ \Gamma' \Imp \Delta' \, \update_{\beta\cdot \delta} \ \Imp \st \store}
        \end{prooftree}
    \end{center}

    \noindent where this cut is admissible by induction on the DHS-complexity of the cut-formula. We then continue the derivation as follows:
    
        \noindent \begin{prooftree}[tiny]
        \hypo{(2)'}
        \infer[no rule]1{G \SepC \eta: X \update_{\alpha \cdot \gamma} \ \Gamma \Imp \Delta \Sep \theta: Y \update_\beta \ \Gamma' \Imp \Delta' \, \update_{\beta\cdot \mathrm{d_1}} \ \Imp \mathrm{pre(d_2}) \st \store}
            \hypo{G \SepC \eta: X \update_{\alpha \cdot \gamma} \ \Gamma \Imp \Delta \Sep \theta: Y \update_\beta \ \Gamma' \Imp \Delta' \ \update_{\beta\cdot \delta} \ \Imp \st \store}
            \infer1[(DW)]{G \SepC \eta: X \update_{\alpha \cdot \gamma} \ \Gamma \Imp \Delta \Sep \theta: Y \update_\beta \ \Gamma' \Imp \Delta' \, \update_{\beta\cdot \mathrm{d_1}} \ \mathrm{pre(d_2}) \Imp \ \update_{\beta\cdot\delta} \ \Imp \st \store}
            \infer1[(New$^\evt$)]{ G \SepC \eta: X \update_{\alpha \cdot \gamma} \ \Gamma \Imp \Delta \Sep \theta: Y \update_\beta \ \Gamma' \Imp \Delta' \, \update_{\beta\cdot \mathrm{d_1}} \ \mathrm{pre(d_2}) \Imp \st \store}
            \infer2[(Cut)]{G \SepC \eta: X \update_{\alpha \cdot \gamma} \ \Gamma \Imp \Delta \Sep \theta: Y \update_\beta \ \Gamma' \Imp \Delta' \ \update_{\beta\cdot \mathrm{d_1}}  \Imp  \st \store}
        \end{prooftree}
     
    \medskip\noindent where this cut is admissible  by induction on the DHS-complexity of the cut-formula. In order to obtain the conclusion, we repeat the same procedure, namely we first apply the rules (DW) and (New$^\evt$) on the conclusion of this last derivation, obtaining  $G \ \vert \ \eta: X \update_{\alpha \cdot \gamma} \Gamma \Imp \Delta \ \vert \ \theta: Y \update_\beta \, \mathrm{pre(d_1}), \Gamma' \Imp \Delta' \st \store$. Finally, we use this IDHS together with $(1)'$ to cut the formula $\mathrm{pre(d_1})$. This cut is admissible by induction on the DHS-complexity of the cut-formula and it gives the desired conclusion.

    \medskip

    $(4)$ Consider the case where we have the pair (R$[\cdot])'$-(L$[\cdot])'$  -- the case of  the pair (R$[\cdot])$-(L$[\cdot])$ can be treated in a similar though simpler manner. The situation is as follows:
    \begin{center}
        \begin{prooftree}[small]
            \hypo{G \SepC \eta: X \update_\alpha \ \mathrm{pre}(\evt), \Gamma \Imp \Delta \update_{\alpha \cdot \evt} \ \Imp A \st \store}
            \infer1[(R$[\cdot])'$]{G \SepC \eta: X \update_\alpha \ \Gamma \Imp \Delta, [\evt]A \st \store}
            \hypo{G \SepC \eta: X \update_\alpha \ \Gamma, \Imp \Delta , \mathrm{pre}(\evt) \st \store}
            \hypo{G \SepC \eta: X \update_\alpha \ \Gamma \Imp \Delta \update_{\alpha \cdot \evt} \ A \Imp \st \store}
            \infer[separation=1.3em]2[(L$[\cdot])'$]{G \SepC \eta: X \update_\alpha \ [\evt]A, \Gamma  \Imp \Delta \st \store}
        \infer[separation=1.3em]2[(Cut)]{G \SepC \eta: X \update_\alpha \ \Gamma  \Imp \Delta \st \store}
        \end{prooftree}
    \end{center}
    \noindent We proceed as below, where both cuts are admissible by the DHS-complexity of the cut-formula:
    \begin{center}
        \begin{prooftree}[small]
        \hypo{G \SepC \eta: X \update_\alpha \ \Gamma \Imp \Delta, \mathrm{pre}(\evt) \st \store}
            \hypo{G \SepC \eta: X \update_\alpha \ \mathrm{pre}(\evt), \Gamma \Imp \Delta \update_{\alpha \cdot \evt} \ \Imp A \st \store}
            \hypo{G \SepC \eta: X \update_\alpha \ \Gamma \Imp \Delta \update_{\alpha \cdot \evt} \ A \Imp \st \store}
            \infer1[(IW)]{G \SepC \eta: X \update_\alpha \ \mathrm{pre}(\evt), \Gamma \Imp \Delta \update_{\alpha \cdot \evt} \ A \Imp \st \store}
        \infer[separation=1.5em]2[(Cut)]{G \SepC \eta: X \update_\alpha \ \mathrm{pre}(\evt), \Gamma \Imp \Delta \update_{\alpha \cdot \evt} \ \Imp \st \store}
        \infer1[(New$^\evt$)]{G \SepC \eta: X \update_\alpha \ \mathrm{pre}(\evt), \Gamma \Imp \Delta \st \store}
    \infer[separation=-4em]2[(Cut)]{G \SepC \eta: X \update_\alpha \ \Gamma \Imp \Delta \st \store}
    \end{prooftree}
    \end{center}

\medskip

    $(5)$ We now move to the very last case where the left-hand premise has been obtained by an application of (Rat). We have the following situation, where $\alpha = \mathrm{e}_1 \cdots \mathrm{e}_k$ and for all $1 \leq i \leq k$, $\alpha_i = \mathrm{e}_1 \cdots \mathrm{e}_i$ -- so $\alpha_k = \alpha$:
    \begin{center}       
    \begin{prooftree}[small]
        \hypo{}
        \ellipsis{$d_1'$}{G \Sep \eta: X \update \cdots \update_{\alpha_{k-1}} \Gamma_{k-1} \Imp \Delta_{k-1}, p \update_{\alpha} \Gamma \Imp \Delta, p}
        \infer1[(Rat)]{G \Sep \eta: X \update \cdots \update_{\alpha_{k-1}}  \Gamma_{k-1} \Imp \Delta_{k-1} \update_{\alpha} \Gamma \Imp \Delta, p}
        \hypo{}
        \ellipsis{$d_2$}{G \Sep \eta: X \update \cdots \update_{\alpha_{k-1}}  \Gamma_{k-1} \Imp \Delta_{k-1} \update_{\alpha} p, \Gamma \Imp \Delta}
        \infer[separation=1.5em]2[(Cut)]{G \Sep \eta: X \update \cdots \update_{\alpha_{k-1}}  \Gamma_{k-1} \Imp \Delta_{k-1} \update_{\alpha} \Gamma \Imp \Delta}
    \end{prooftree}\\
    \end{center}
    
    \noindent Here we go up derivation $d_1'$ until we find an application of a rule $\mathcal{R}^{\prime}$ different from (Rat) on that atom $p$. There are two possibilities:
    
    $\bullet$ After $m$ applications of (Rat), we indeed find a rule $\Rule^{\prime}$ different from (Rat) on that atom $p$. Thanks to Lemma \ref{lemmaPermutationRules}, we can permute $\Rule^{\prime}$ down with the $m$ applications of (Rat). We also apply $\Rule^{\prime\star}$ on $G \ \vert \ \eta: X \update \cdots \update_\alpha p,\Gamma \Imp \Delta$. Then, we cut the atom $p$ from those two  IDHS. This cut is admissible by induction on the sum of the heights. We then apply rule $\mathcal{R}^{\prime}$ again, thereby obtaining the desired result.

    $\bullet$ Alternatively, there is no such rule because we find an axiom:
    \begin{center}       
    \begin{prooftree}[small]
        \hypo{G \Sep \eta: X' \update \cdots \update_{\alpha_{k-1}} \Gamma_{k-1} \Imp \Delta_{k-1}, p \update_\alpha \Gamma \Imp \Delta, p}
        \ellipsis{(Rat)$^\ast$}{G \Sep \eta: X \update \cdots \update_{\alpha_{k-1}} \Gamma_{k-1} \Imp \Delta_{k-1}, p \update_\alpha \Gamma \Imp \Delta, p}
        \infer1[(Rat)]{G \Sep \eta: X \update \cdots \update_{\alpha_{k-1}} \Gamma_{k-1} \Imp \Delta_{k-1} \update_{\alpha} \Gamma \Imp \Delta, p}
        \hypo{}
        \ellipsis{$d_2$}{G \Sep \eta: X \update \cdots \update_{\alpha_{k-1}} \Gamma_{k-1} \Imp \Delta_{k-1} \update_{\alpha} p, \Gamma \Imp \Delta}
        \infer[separation=2em]2[(Cut)]{G \Sep \eta: X \update \cdots \update_{\alpha_{k-1}} \Gamma_{k-1} \Imp \Delta_{k-1} \update_{\alpha} \Gamma \Imp \Delta}
    \end{prooftree}\\
    \end{center}

    \medskip We distinguish the three following cases.
    $(5a)$ $p$ is principal in $\update_\alpha \ \Gamma \Imp \Delta$. We then start from the right-hand premise and obtain the conclusion by contraction -- analogously to case (i) above.
    $(5b)$ $p$ is principal in none of  the $\update_{\alpha_i} \ \Gamma_i \Imp \Delta_i$. Then the conclusion is also an axiom and we thus already have the desired result.
    $(5c)$ $p$ is principal in another dynamic sequent $\update_{\alpha_i} \ \Gamma_i \Imp \Delta_i$ (for $1 \leq i < k$). We then need to look at the right-hand premise and further distinguish cases.
    If it is an axiom, we reason as in case (i) above. If it has been obtained by the application of a rule $\Rule^\prime$ where $p$ is not principal, then we proceed analogously to case (ii). Finally, if it is the result of $l$ applications of (Lat), we go up the $l$ applications of (Lat) until we find a rule other that (Lat) on the atom $p$. We again distinguish two possibilities, analogous to those we have enumerated for (Rat). If we find a rule $\Rule''$, we proceed by using permutation. Otherwise, if we find an axiom, we need to distinguish the following cases.
    $(5c,i)$ If $p$ is principal in $\update_\alpha \ \Gamma \Imp \Delta$, we proceed as in $(5a)$.
    $(5c, ii)$ If $p$ is not the principal atom of the axioms, or it is and it occurs in the same sequent $\update_{\alpha_i} \ \Gamma_i \Imp \Delta_i$ as in the other premise, then the conclusion is also an axiom and we thus have the desired result.
    $(5c, iii)$ If $p$ is principal in a sequent $\update_{\alpha_j} \ \Gamma_j \Imp \Delta_j$, where $i\neq j$, we use hp-admissibility of contraction and then we  apply a cut on the atom $p$, which is admissible by induction on the sum of the heights.
    More details can be found in the Appendix.
   
\end{proof}

\section{Conclusions} \label{sectionConclusion}

In this paper, we have developed a proof-theoretic calculus for DEL based on the framework of indexed dynamic hypersequents with event store. Starting from the proof-theoretic treatment of multi-agent S5, we introduced a new structural formalism capable of representing epistemic events directly within the proof system, thereby capturing the dynamics of DEL without resorting to auxiliary semantic devices such as labels or relational atoms.

 We proved the admissibility of the structural rules of the calculus, including weakening, contraction, and merge, established soundness and completeness with respect to DEL, and showed the admissibility of the cut rule.
The resulting system provides, to the best of our knowledge, the first proof-theoretic calculus for DEL over a classical S5 base in which all structural rules -- including contraction and cut -- are admissible. More generally, it shows that the dynamic evolution of epistemic models can be internalised at the structural level, yielding a natural proof-theoretic account of epistemic events that mirrors the expressive power of DEL itself.

Several directions for future work naturally arise from this framework. On the proof-theoretic side, it would be interesting to investigate whether the present approach can be extended to richer variants of DEL, such as logics with common knowledge or distributed knowledge. It would also be worthwhile to explore decidability, counter-model construction, proof-search optimisations and complexity bounds for the calculus. Finally, we believe that the structural methodology introduced here may provide a useful foundation for developing proof systems for other dynamic modal logics in which model transformations play a central role.

\bibliographystyle{cas-model2-names}

\bibliography{biblioDHSforDEL}

@InCollection{sepDEL,
	author       =	{Baltag, Alexandru and Renne, Bryan},
	title        =	{{Dynamic Epistemic Logic}},
	booktitle    =	{The {Stanford} Encyclopedia of Philosophy},
	editor       =	{Edward N. Zalta},
	howpublished =	{\url{https://plato.stanford.edu/archives/win2016/entries/dynamic-epistemic/}},
	year         =	{2016},
	edition      =	{{W}inter 2016},
	publisher    =	{Metaphysics Research Lab, Stanford University}
}

@inproceedings{aucherandothers1,
  author={Aucher, Guillaume and Maubert, Bastien and Schwarzentruber, François},
  title={Generalized {DEL}-sequents},
  booktitle={JELIA 2012. LNCS, vol.7519},
  pages={54--66},
  year={2012},
  publisher={Springer},
   editor={Cerro, L.F. and Herzig, A. and Mengin, J.}
}

@inproceedings{aucherandothers2,
  author={Aucher, Guillaume and Schwarzentruber, François},
  title={On the complexity of dynamic epistemic logic.},
  booktitle={Proceedings of TARK},
  pages={54--66},
  year={2013},
   editor={Schipper, B.C.}
}

@article{aumann,
	author = {Aumann, Robert J.},
	title = {Agreeing to disagree},
	journal = {The Annals of Statistics},
    publisher = {Institute of Mathematical Statistics},
    volume = {4},
	year = {1976},
	pages = {1236--1239},
}

@inproceedings{baltagandothers,
    author = {Baltag, Alexandru and Moss, Lawrence and Solecki, Slawomir},
    title = {The Logic of Public Announcements, Common Knowledge, and Private Suspicions},
    year = {1998},
    booktitle = {Proceedings of the 7th Conference on Theoretical Aspects of Rationality and Knowledge (TARK VII)},
    editor = {Gilboa, I.},
    pages = {43--56}
}

@misc{bílková2025agentinterpolationknowledge,
      title={Agent Interpolation for Knowledge}, 
      author={Marta Bílková and Wesley Fussner and Roman Kuznets},
      year={2025},
      eprint={2505.23401},
      archivePrefix={arXiv},
      primaryClass={cs.LO},
      url={https://arxiv.org/abs/2505.23401}, 
}

@book{BlackburnModalLogic,
  title={Modal logic},
  author={Blackburn, Patrick and De Rijke, Maarten and Venema, Yde},
  volume={53},
  year={2001},
  publisher={Cambridge University Press}
}

@book{DEL,
  title={Dynamic epistemic logic},
  author={Ditmarsch, Hans van and Hoek, Wiebe and Kooi van Der, Barteld},
  volume={337},
  year={2007},
  publisher={Springer Science \& Business Media}
}

@article{lerouvillois2025dynamichypersequentspublicannouncement,
      title={Dynamic Hypersequents for Public Announcement Logic}, 
      author={Clara Lerouvillois and Francesca Poggiolesi},
      year={2026},
      journal={The Review of Symbolic Logic},
      volume={19},
      pages={491-518},
      }

@article{Sano2,
	author = {Liu, S. and Sano, K.},
	title = {Sequent calculi for public announcement logic and action model logic},
	year = {2026},
	journal = {Synthese},
    volume = {208},
	pages = {1--46},

}

@article{Frittellandothers1,
  title={Multi-type display calculus for propositional dynamic logic},
  author={Frittella, Sabine and Greco, Giuseppe and Kurz, Alexander and Palmigiano, Alessandra},
  journal={Journal of Logic and Computation},
  year={2016},
pages={2067–2104},
volume={26}
}

@article{Frittellandothers2,
  title={A proof-theoretic semantic analysis of dynamic epistemic logic},
  author={Frittella, Sabine and Giuseppe, Greco and Kurz, Alexander and Palmigiano, Alessandra and Sikimic, Vlasta},
  journal={Journal of Logic and Computation},
  year={2016},
pages={1961–2105},
volume={26}
}

@article{Nomuraandothers2,
  title={A cut-free labelled sequent calculus for dynamic epistemic logic},
  author={Nomura, Shoshin and  Ono, Hiroakira and Sano,Katsuhiko},
  journal={Journal of Logic and Computation},
  year={2020},
pages={321–348},
volume={30}
}

@inproceedings{nomura2015revising,
  title={Revising a labelled sequent calculus for public announcement logic},
  author={Nomura, Shoshin and Sano, Katsuhiko and Tojo, Satoshi},
  booktitle={Structural Analysis of Non-Classical Logics: The Proceedings of the Second Taiwan Philosophical Logic Colloquium},
  pages={131--157},
  year={2015},
  organization={Springer}
}

@inproceedings{plaza1989,
  title={Logics of public announcements},
  author={Plaza, Jan},
  booktitle={Proceedings 4th international symposium on methodologies for intelligent systems},
  pages={201--216},
  year={1989}
}

@article{Poggiolesi2008,
  title={A Cut-Free Simple Sequent Calculus
for Modal Logic S5},
  author={Poggiolesi, Francesca},
  journal={Review of Symbolic Logic},
  year={2008},
pages={3-15},
volume={1}
}

@book{Poggiolesi2010,
    author={Poggiolesi, Francesca},
    title={Gentzen Calculi for Modal Propositional Logic},
    publisher={Springer},
    year={2010},
    place={Berlin},
    collection={Trends in Logic}
}

@article{Poggiolesi2013fromSingleToMany,
  title={From single agent to multi-agent via hypersequents},
  author={Poggiolesi, Francesca},
  journal={Logica Universalis},
  volume={7},
  number={2},
  pages={147--166},
  year={2013},
  publisher={Springer}
}

@inbook{Restall,
	author = {Restall, Greg},
	title = {Proofnets for S5: sequents and circuits for modal logic},
	booktitle = { Logic Colloquium 2005: Proceedings of the Annual European Summer Meeting of the Association for Symbolic Logic},
	year = {2007},
	publisher = {Cambridge University Press},
	pages = {151--172},
	editor = {C. Dimitracopoulos and L. Newelski and D. Normann}
}

@book{BasicProofTheoryTroelstra,
    author={Troelstra, Anne S. and Schwichtenberg, Helmut},
    title={Basic Proof Theory},
    publisher={Cambridge University Press},
    year={2000},
    place={Cambridge},
    edition={2},
    series={Cambridge Tracts in Theoretical Computer Science},
    collection={Cambridge Tracts in Theoretical Computer Science}
}

\appendix
\section{Appendix}

\paragraph*{Cut admissibility (Theorem \ref{theoremCut})}

Let us go back to case (3) where $\Rule$ is an application of ($\evt$RK) and the right-hand premise has been obtained by an application of ($\evt$LK$_2)'$ with $k$ premises. Suppose we have the following situation, where $\delta=\mathrm{d}_1, \dots, \mathrm{d}_{k-1}$, $\eta \cap \theta = \lbrace na \rbrace$ and $\alpha\cdot\gamma \approx_a \beta \cdot \delta \in \store$.  Moreover, $\overline{G}:= G \SepC \eta: X \update_{\alpha\cdot\gamma} K_a A, \Gamma \Imp \Delta$ and $\overline{G'}:= G \SepC \theta: Y \update_\beta \Gamma' \Imp \Delta'$. 
    \begin{center}
    \scalebox{0.75}{
        \begin{prooftree}[regular]
            \hypo{\overline{G'} \SepC \eta: X \update_{\alpha\cdot \gamma} \Gamma \Imp \Delta \SepC na: \update_\chi \Imp A \st \store, \alpha \cdot \gamma \sim_a \chi}
            \infer1[($\evt$RK)]{G \SepC \eta: X \update_{\alpha\cdot \gamma}  \Gamma \Imp \Delta,  K_a A \SepC \theta: Y \update_\beta \Gamma' \Imp \Delta' \st \store}
                \hypo{(1)}
                \infer[no rule]1{\overline{G} \SepC \theta: Y \update_\beta  \Gamma' \Imp \Delta',  \mathrm{pre(d_1}) \st \store}
                \hypo{\cdots}
                \hypo{(k)}
                \infer[no rule]1{\overline{G} \SepC \theta: Y \update_\beta \Gamma' \Imp \Delta' \update_{\beta \cdot \delta} A \Imp  \st \store}
            \infer3[($\evt$LK$_2)'$]{G \SepC \eta: X \update_{\alpha\cdot \gamma} K_a A, \Gamma \Imp \Delta  \SepC \theta: Y \update_\beta \Gamma' \Imp \Delta' \st \store}
            \infer2[(Cut)]{G \SepC \eta: X \update_{\alpha\cdot \gamma} \Gamma \Imp \Delta  \SepC \theta: Y \update_\beta \Gamma' \Imp \Delta' \st \store}
        \end{prooftree}
    }
    \end{center}    
    \hfill \\

  \noindent    
    We first operate the following cut, which is admissible by induction on the sum of the heights:
    \begin{center}
        \begin{prooftree}[small]
            \hypo{G \SepC \eta: X \update_{\alpha\cdot \gamma} \Gamma \Imp \Delta, K_a A \SepC \theta: Y \update_\beta \Gamma' \Imp \Delta' \st \store  }
            \infer1[(DW)]{G \SepC \eta: X \update_{\alpha\cdot \gamma} \Gamma \Imp \Delta, K_a A \SepC \theta: Y \update_\beta \Gamma' \Imp \Delta' \update_{\beta\cdot\delta} A \Imp \st \store  }
            \hypo{(k)}
            \infer[no rule]1{G \SepC \eta: X \update_{\alpha\cdot \gamma} \Gamma \Imp \Delta, K_a A \SepC \theta: Y \update_\beta  \Gamma' \Imp \Delta' \update_{\beta \cdot \delta} A \Imp  \st \store}
            \infer2[(Cut)]{\ G \SepC \eta: X \update_{\alpha\cdot \gamma} \Gamma \Imp \Delta \SepC \theta: Y \update_\beta \Gamma' \Imp \Delta' \update_{\beta\cdot\delta} A \Imp \st \store}
        \end{prooftree}
    \end{center}
    \noindent where we denote the conclusion by $(I_k)$.
    We then consider the premise of the rule ($\evt$RK), apply the rule (DW) and (Merge)\footnote{Note that this application of (Merge) uses the fact that $\alpha \cdot \gamma \approx_a \beta \cdot \delta$ and $\alpha \cdot \gamma \sim_a \chi$ are both in the event store, thus so is $\beta \cdot \delta \approx_a \chi$. Also note that when substituting $\chi$ by $\beta\cdot \delta$ we get $\alpha \cdot \gamma \sim_a \beta \cdot \delta$ which is deleted from the event store since $\alpha \cdot \gamma \approx_a \beta \cdot \delta$ already occurs in $\store$.} as below.
    
    \begin{center}
        \begin{prooftree}[small]
            \hypo{G \SepC \eta: X \update_{\alpha\cdot \gamma} \Gamma \Imp \Delta \SepC na: \update_\chi \Imp A \SepC \theta: Y \update_\beta \Gamma' \Imp \Delta' \st \store, \alpha \cdot \gamma \sim_a \chi}
            \infer1[(DW)]{G \SepC \eta: X \update_{\alpha\cdot \gamma} \Gamma \Imp \Delta \SepC na: \update_\chi \Imp A \SepC \theta: Y \update_\beta \Gamma' \Imp \Delta' \update_{\beta \cdot \delta} \Imp \st \store, \alpha \cdot \gamma \sim_a \chi}
            \infer1[(Merge)]{G \SepC \eta: X \update_{\alpha \cdot \gamma}\  \Gamma \Imp \Delta \SepC \theta: \update_\beta \ \Gamma' \Imp \Delta' \update_{\beta\cdot \delta} \ \Imp A \st \store}
        \end{prooftree}
    \end{center}

    \noindent We then use this IDHS together with $(I_k)$ resulting from the previous cut, to operate a novel cut on the formula $A$, and we go on by applying admissibility of (DW) and (New$^\evt$):
    \begin{center}
        \begin{prooftree}[small]
            \hypo{G \SepC \eta: X \update_{\alpha \cdot \gamma} \Gamma \Imp \Delta \SepC \theta: \update_\beta \Gamma' \Imp \Delta' \update_{\beta\cdot \delta} \Imp A \st \store}
            \hypo{(I_k)}
            \infer[no rule]1{G \SepC \eta: X \update_{\alpha \cdot \gamma} \Gamma \Imp \Delta \SepC \theta: Y \update_\beta \Gamma' \Imp \Delta' \update_{\beta\cdot \delta} A \Imp \st \store}
            \infer[separation=2em]2[(Cut)]{G \SepC \eta: X \update_{\alpha \cdot \gamma} \Gamma \Imp \Delta \SepC \theta: Y \update_\beta \, \Gamma' \Imp \Delta' \, \update_{\beta\cdot \delta} \Imp \st \store}
            \infer1[(DW)]{G \SepC \eta: X \update_{\alpha \cdot \gamma} \Gamma \Imp \Delta \SepC \theta: Y \update_\beta \Gamma' \Imp \Delta' \update_{\beta\cdot \mathrm{d_1\cdots d}_{k-2}} \mathrm{pre(d}_{k-1}) \Imp \update_{\beta\cdot \delta} \Imp \st \store}
            \infer1[(New$^\evt$)]{ G \SepC \eta: X \update_{\alpha \cdot \gamma} \Gamma \Imp \Delta \SepC \theta: Y \update_\beta \, \Gamma' \Imp \Delta' \update_{\beta\cdot \mathrm{d_1\cdots d}_{k-2}} \mathrm{pre(d}_{k-1}) \Imp \st \store}
        \end{prooftree}
    \end{center}
    \noindent where we call the conclusion $(I_{k-1})$. This new cut is admissible by induction on complexity of the cut-formula.
    We now use $(I_{k-1})$ and premise $(k-1)$ to cut $\mathrm{pre(d}_{k-1})$, by induction on the complexity of the cut formula. We apply the admissibility of (DW) and (New$^\evt$) one more time to obtain
    $$(I_{k-2}): G \Sep \eta: X \update_{\alpha \cdot \gamma} \Gamma \Imp \Delta \Sep \theta: Y \update_\beta \Gamma' \Imp \Delta' \update_{\beta\cdot \mathrm{d_1\cdots d}_{k-3}} \mathrm{pre(d}_{k-2}) \Imp \st \store.$$
    \noindent We use it together with $(k-2)$ to cut $\mathrm{pre(d}_{k-2})$, by induction on the complexity, and use (DW) and (New$^\evt$) as before. We proceed this way until obtaining
    $$(I_1): G \Sep \eta: X \update_{\alpha \cdot \gamma} \Gamma \Imp \Delta \Sep \theta: Y \update_\beta \ \mathrm{pre(d_1}), \Gamma' \Imp \Delta' \st \store$$
    \noindent and apply one last time the inductive hypothesis with premise $(1)$ to cut $\mathrm{pre(d_1})$, by induction on the complexity.

    \begin{center}
        \begin{prooftree}[small]
            \hypo{(1)}
            \infer[no rule]1{G \SepC \eta: X \update_{\alpha \cdot \gamma} \Gamma \Imp \Delta \SepC \theta: Y \update_\beta \, \Gamma' \Imp \Delta',\mathrm{pre(d_1}) \st \store}
            \hypo{(I_1)}
            \infer[no rule]1{G \SepC \eta: X \update_{\alpha \cdot \gamma} \Gamma \Imp \Delta \SepC \theta: Y \update_\beta \ \mathrm{pre(d_1}), \Gamma' \Imp \Delta' \st \store}
            \infer[separation=2em]2[(Cut)]{G \SepC \eta: X \update_{\alpha \cdot \gamma} \Gamma \Imp \Delta \SepC \theta: \update_\beta \, \Gamma' \Imp \Delta' \st \store}
        \end{prooftree}
    \end{center}

    \medskip

    Let us now move to the very last case $(5c, iii)$ in the proof of Theorem \ref{theoremCut}, where both premises have been obtained by finitely many applications of (Rat) and (Lat), respectively, and axioms where $p$ is principal in different dynamic components. More precisely, we consider here axioms obtained after $k$ applications of (Rat) and $l$ applications of (Lat), respectively in $\update_{\alpha_i} \Gamma_i \Imp \Delta_i$ in the left-hand premise and $\update_{\alpha_j} \Gamma_j \Imp \Delta_j$ in the right-hand premise. Without loss of generality, let us consider $i < j$ -- so $k > l$.
    We explain the procedure in the general case and then provide one simple example.
    We have the following situation:

    \begin{center}
    \scalebox{.75}{
    \begin{prooftree}[regular]
                \hypo{G \SepC \eta: X \update_{\alpha_i} p, \Gamma_i \Imp \Delta_i,p \update \cdots \update_{\alpha_j} \Gamma_j \Imp \Delta_j, p, p \update \cdots  \update_\alpha \Gamma \Imp \Delta, p \st \store}
            \ellipsis{(Rat)$\ \times \ (k-l)$}{G \SepC \eta: X \update_{\alpha_i} p, \Gamma_i \Imp \Delta_i \update \cdots \update_{\alpha_j}\Gamma_j \Imp \Delta_j, p, p \update \cdots  \update_\alpha \Gamma \Imp \Delta, p \st \store}
            \infer1[(Rat)]{G \SepC \eta: X \update_{\alpha_i} p, \Gamma_i \Imp \Delta_i \update \cdots \update_{\alpha_j}\Gamma_j \Imp \Delta_j, p \update \cdots  \update_\alpha \Gamma \Imp \Delta, p \st \store}
            \ellipsis{(Rat)$\ \times \ (l-1)$}{G \SepC \eta: X \update_{\alpha_i} p, \Gamma_i \Imp \Delta_i \update \cdots \update_{\alpha_j}\Gamma_j \Imp \Delta_j, p \update \cdots  \update_\alpha \Gamma \Imp \Delta, p \st \store}
                \hypo{G \SepC \eta: X \update_{\alpha_i} p, \Gamma_i \Imp \Delta_i \update \cdots \update_{\alpha_j} p, \Gamma_j \Imp \Delta_j, p \update \cdots  \update_\alpha p, \Gamma \Imp \Delta \st \store}
                \ellipsis{(Lat)$\ \times \ l$}{G \SepC \eta: X \update_{\alpha_i} p, \Gamma_i \Imp \Delta_i \update \cdots \update_{\alpha_j} \Gamma_j \Imp \Delta_j, p \update \cdots  \update_\alpha p, \Gamma \Imp \Delta \st \store}
            \infer2[(Cut)]{G \SepC \eta: X \update_{\alpha_i} p, \Gamma_i \Imp \Delta_i \update \cdots \update_{\alpha_j} \Gamma_j \Imp \Delta_j, p \update \cdots  \update_\alpha \Gamma \Imp \Delta \st \store}
        \end{prooftree}
    }
        \end{center}
    \medskip

    Let $h$ be the height of the left-hand premise. Going up the corresponding derivation, the $l^{\text{th}}$ application of (Rat) that we find is the following, where the premise's height is $h-l$:
        \begin{center}
        \begin{prooftree}[small]
            \hypo{G \Sep \eta: X \update_{\alpha_i} p, \Gamma_i \Imp \Delta_i \update \cdots \update_{\alpha_j}\Gamma_j \Imp \Delta_j, p, p \update \cdots  \update_\alpha \Gamma \Imp \Delta, p \st \store}
            \infer1[(Rat)]{G \Sep \eta: X \update_{\alpha_i} p, \Gamma_i \Imp \Delta_i \update \cdots \update_{\alpha_j}\Gamma_j \Imp \Delta_j, p \update \cdots  \update_\alpha \Gamma \Imp \Delta, p \st \store}
        \end{prooftree}
    \end{center}
    \noindent so we apply (RC) -- that is hp-admissible by Lemma \ref{lemmaAdmissibilityC} -- to the premise and get the same indexed dynamic hypersequent with event store as the conclusion $G \ \vert \ \eta: X \update_{\alpha_i} p, \Gamma_i \Imp \Delta_i \update \cdots \update_{\alpha_j}\Gamma_j \Imp \Delta_j, p \update \cdots  \update_\alpha \Gamma \Imp \Delta, p \st \store$ but now of equal height $h-l$. We then apply $l-1$ times (Rat) to get $G \ \vert \ \eta: X \update_{\alpha_i} p, \Gamma_i \Imp \Delta_i \update \cdots \update_{\alpha_j}\Gamma_j \Imp \Delta_j, p \update \cdots  \update_\alpha \Gamma \Imp \Delta, p \st \store$, which is now of height $h-1$. We can therefore apply the inductive hypothesis with the right-hand premise to cut atom $p$ and obtain the desired conclusion. This application of the Cut-rule is admissible by induction on the sum of heights.
    
    Below is an example with $k= 2$ and $l=1$.

\begin{center}
    \begin{prooftree}[small]
        \hypo{G\SepC \eta: X \update_{\alpha_1} p, \Gamma_1 \Imp \Delta_1, p \update_{\alpha_2} \Gamma_2 \Imp \Delta_2,p,p \update_\alpha \Gamma \Imp \Delta, p \st \store }
        \infer1[(Rat)]{G\SepC \eta: X \update_{\alpha_1} p, \Gamma_1 \Imp \Delta_1 \update_{\alpha_2} \Gamma_2 \Imp \Delta_2,p,p \update_\alpha \Gamma \Imp \Delta, p \st \store}
        \infer1[(Rat)]{G\SepC \eta: X \update_{\alpha_1} p, \Gamma_1 \Imp \Delta_1 \update_{\alpha_2} \Gamma_2 \Imp \Delta_2,p \update_\alpha \Gamma \Imp \Delta, p \st \store}
        \hypo{G\SepC \eta: X \update_{\alpha_1} p, \Gamma_1 \Imp \Delta_1 \update_{\alpha_2} p, \Gamma_2 \Imp \Delta_2,p \update_\alpha p, \Gamma \Imp \Delta  \st \store}
        \infer1[(Lat)]{G\SepC \eta: X \update_{\alpha_1} p, \Gamma_1 \Imp \Delta_1 \update_{\alpha_2} \Gamma_2 \Imp \Delta_2,p \update_\alpha p, \Gamma \Imp \Delta \st \store}
        \infer2[(Cut)]{G\SepC \eta: X \update_{\alpha_1} p, \Gamma_1 \Imp \Delta_1 \update_{\alpha_2} \Gamma_2 \Imp \Delta_2,p \update_\alpha \Gamma \Imp \Delta \st \store}
    \end{prooftree}
\end{center}


    \noindent We apply hp-admissibility of contraction and proceed with the following cut, which is admissible by induction on the sum of the heights:
    
\begin{center}
        \begin{prooftree}[small]
         \hypo{G\SepC \eta: X \update_{\alpha_1} p, \Gamma_1  \Imp \Delta_1, p \update_{\alpha_2} \Gamma_2 \Imp \Delta_2,p,p \update_\alpha \Gamma \Imp \Delta, p \st \store }
        \infer1[(Rat)]{G\SepC \eta: X \update_{\alpha_1} p, \Gamma_1 \Imp \Delta_1 \update_{\alpha_2} \Gamma_2 \Imp \Delta_2,p,p \update_\alpha \Gamma \Imp \Delta, p \st \store}
        \infer1[(RC)]{G\SepC \eta: X \update_{\alpha_1} p, \Gamma_1 \Imp \Delta_1 \update_{\alpha_2} \Gamma_2 \Imp \Delta_2,p \update_\alpha \Gamma \Imp \Delta, p \st \store}
        \hypo{G\SepC \eta: X \update_{\alpha_1} p, \Gamma_1 \Imp \Delta_1 \update_{\alpha_2} p, \Gamma_2 \Imp \Delta_2,p \update_\alpha p, \Gamma \Imp \Delta \st \store}
        \infer1[(Lat)]{G\SepC \eta: X \update_{\alpha_1} p, \Gamma_1 \Imp \Delta_1 \update_{\alpha_2} \Gamma_2 \Imp \Delta_2,p \update_\alpha p, \Gamma \Imp \Delta \st \store}
        \infer2[(Cut)]{G\SepC \eta: X \update_{\alpha_1} p, \Gamma_1 \Imp \Delta_1 \update_{\alpha_2} \Gamma_2 \Imp \Delta_2,p \update_\alpha \Gamma \Imp \Delta \st \store}
    \end{prooftree}
\end{center}


\end{document}